\documentclass[12pt]{article}

\usepackage[left=1in,top=1in,right=1in,bottom=1in,head=.1in,nofoot]{geometry}
\usepackage{setspace,url,bm,amsmath} 
\usepackage{titlesec} 
\titlelabel{\thetitle.\quad} 
\usepackage[margin=20pt]{subcaption}
\usepackage{amsmath,amsfonts,amsthm,amssymb}
\usepackage{graphicx}
\usepackage[colorlinks=true,linkcolor=black,citecolor=blue,urlcolor=blue]{hyperref}
\usepackage[table]{xcolor}
\usepackage{multirow}
\usepackage{array}
\usepackage{booktabs}
\usepackage{threeparttable}
\usepackage{comment}
\usepackage{float}
\usepackage{natbib}
\usepackage{indentfirst}

\newtheorem*{theorem*}{Theorem}
\newtheorem{theorem}{Theorem}
\newtheorem{proposition}{Proposition}
\newtheorem{lemma}{Lemma}
\newtheorem{remark}{Remark}

\newtheorem{corollary}{Corollary}
\newtheorem*{corollary*}{Corollary}
\newtheorem{condition}{Condition}

\def\T{\text{T}}

\def\adj{\text{adj}}

\def\sumk{\sum_{k=1}^{K} }
\def\sumik{\sum_{i \in [k]}}
\def\nkt{n_{[k]1}}
\def\nkc{n_{[k]0}}
\def\pik{\pi_{[k]}}

\def\tauk{ \tau_{[k]}}

\def\pk{p_{[k]}}

\def\nk{n_{[k]}}

\def\nt{n_1}
\def\nc{n_0}

\def\SkXY{S_{[k]XY}}

\DeclareMathOperator*{\argmin}{arg\,min}

\def\T{{ \mathrm{\scriptscriptstyle T} }}
\def\cov{\mathrm{cov}}
\def\var{\mathrm{var}}
\def\pr{\mathrm{pr}}

\def\Mtau{\mathcal{M}_{\hat \tau_X}}
\def\MI{\mathcal{M}_s}

\def\P{P}
\def\dd{\mathrm{d}}

\def\p{p}
\def\S{K}
\def\pk{p_{[k]}}
\def\nk{n_{[k]}}
\def\nkz{n_{[k]z}}
\def\pik{\pi_{[k]}}
\def\obs{\text{obs}}
\def\nkt{n_{[k]1}}
\def\nkc{n_{[k]0}}
\def\nt{n_1}
\def\nc{n_0}

\def\Yi{Y_{i}}
\def\Yk{\bar Y_{[k]}}

\def\bYk{\bar Y_{[k]}}
\def\Yobs{\bar Y^{\obs}}

\def\Xi{X_i}

\def\bXk{\bar X_{[k]}}
\def\Xobs{\bar X^{\obs}}

\def\ttX{{\tilde \tau_X}}
\def\Md{\mathcal{M}_\ttX}
\def\taui{\tau_i}
\def\tauk{\tau_{[k]}}
\def\taukX{\tau_{[k]X}}

\def\tmj{\tau_{m_j}}
\def\htmj{\hat\tau_{m_j}}
\def\skmj{\sum_{k\in A_{ss}:\nk=m_j}}

\def\Wi{W_i}
\def\tW{\tau_W}
\def\htW{\hat{\tau}_W}

\def\bWk{\bar W_{[k]}}
\def\SkW{S_{[k]WW}}

\def\SkWY{S_{[k]WY}}

\def\SkY{S^2_{[k]Y}}
\def\SkX{S_{[k]XX}}

\def\SkXY{S_{[k]XY}}
\def\Skt{S^2_{[k]\tau}}
\def\Sktr{S^2_{[k]\tau_R }}

\def\skY{s^2_{[k]Y}}

\def\skYX{s_{[k]YX}}
\def\skXY{s_{[k]XY}}

\def\Ymul{R}
\def\bYmulk{\bar R_{[k]}}
\def\tmul{\tau_{\Ymul}}
\def\htmul{\hat\tau_{\Ymul}}
\def\htkmul{\hat\tau_{\Ymul,[k]}}
\def\SkYmul{S^2_{[k]\Ymul}}
\def\Sktmul{S^2_{[k]\tmul}}
\def\Sigmul{\Sigma_{\Ymul}}

\def\Ynew{R^{\new}}
\def\bYnewk{\bar R^{\new}_{[k]}}
\def\tnew{\tau^{\new}}
\def\htnew{\hat\tau^{\new}}
\def\SkYnew{S^2_{[k]\Ynew}}
\def\Sktnew{S^2_{[k]\tnew}}

\def\Zi{Z_i}
\def\htk{\hat\tau_{[k]}}
\def\htX{\hat{\tau}_X}
\def\htkX{\hat{\tau}_{[k]X}}
\def\Mk{M_{[k]}}

\def\Ai{A_i}
\def\Bi{B_i}
\def\bAk{\bar A_{[k]}}
\def\bBk{\bar B_{[k]}}
\def\SkAB{S_{[k]AB}}
\def\skAB{s_{[k]AB}}
\def\rkz{r_{[k]z}}
\def\bAobsk{\bar A^{\obs}_{[k]}}
\def\bBobsk{\bar B^{\obs}_{[k]}}

\def\Vktt{\Sigma_{[k]\tau\tau}}
\def\Vktx{\Sigma_{[k]\tau x}}
\def\Vkxt{\Sigma_{[k]x\tau}}
\def\Vkxx{\Sigma_{[k]xx}}
\def\hVktt{\hat\Sigma_{[k]\tau\tau}}
\def\tVktt{\tilde\Sigma_{[k]\tau\tau}}
\def\Rk{R_{[k]}^2}
\def\hRk{\hat R_{[k]}^2}
\def\tRk{\tilde R_{[k]}^2}

\def\pa{p_a}
\def\tpa{p_{\tilde a}}

\def\ntau{n^{1/2}(\hat\tau-\tau)}
\def\skS{\sum_{k=1}^{K}}
\def\sumk{\sum_{i\in[k]}}
\def\ik{i\in[k]}
\def\CI{(1-\alpha)}

\def\new{\text{new}}

\def\opt{\textnormal{opt}}
\def\adj{\textnormal{adj}}

\def\argmin{\mathop{\arg\min}}

\def\Sktr{S^2_{[k]\tau_R }}
\def\Sktrij{S^2_{[k]{\tau_R},ij }}
\def\Sktrii{S^2_{[k]{\tau_R},ii }}

\def\skrij{s^2_{[k]R,ij}}

\def\Skrij{S^2_{[k]R,ij}}
\def\Skrii{S^2_{[k]R,ii}}
\def\trk{\tau_{R,[k]}}
\def\trs{\tau_{R,ss}}
\def\trki{\tau_{R,[k],i}}
\def\trsi{\tau_{R,ss,i}}
\def\trkj{\tau_{R,[k],j}}
\def\trsj{\tau_{R,ss,j}}
\def\htrk{\hat\tau_{R,[k]}}

\def\htrsi{\hat\tau_{R,ss,i}}
\def\htrki{\hat\tau_{R,[k],i}}
\def\htrsj{\hat\tau_{R,ss,j}}
\def\htrkj{\hat\tau_{R,[k],j}}
\def\Ri{R_i}
\def\Rij{R_{i,j}}
\def\oRk{\bar{R}_{[k]}}
\def\oRkj{\bar{R}_{[k],j}}
\def\argmin{\mathop{\arg\min}}

\def\new{\text{new}}

\begin{document}
\begin{singlespace}
\title{\bf Rerandomization in stratified randomized experiments}

\author{
\small
{
Xinhe Wang$^{1}$\thanks{
    The authors gratefully acknowledge \textit{the Tsinghua University Initiative Scientific Research Program and  the National Natural Science Foundation of China grant 11701316 and 12071242}. Xinhe Wang and Tingyu Wang contribute equally to this work.}, Tingyu Wang$^{2}$, Hanzhong Liu$^{3}$\thanks{\small{Correspondence: \texttt{lhz2016@tsinghua.edu.cn}}}
}
\\ \\
{\small $^{1}$ Department of Mathematical Sciences, Tsinghua University, Beijing, China}\\
{\small $^{2}$ Department of Physics, Tsinghua University, Beijing, China}\\
{\small $^{3}$ Center for Statistical Science, Department of Industrial Engineering, Tsinghua University, Beijing, China}
}

\date{}
\maketitle
\end{singlespace}

\thispagestyle{empty}
\vskip -8mm 

\begin{singlespace}
\begin{abstract}

Stratification and rerandomization are two well-known methods used in randomized experiments for balancing the baseline covariates. Renowned scholars  in experimental design have recommended combining these two methods; however, limited studies have addressed the statistical properties of this combination. This paper proposes two rerandomization methods to be used in stratified randomized experiments, based on the overall and stratum-specific Mahalanobis distances. The first method is applicable for nearly arbitrary numbers of strata, strata sizes, and stratum-specific proportions of the treated units. The second method, which is generally more efficient than the first method, is suitable for  situations in which the number of strata is fixed with their sizes tending to infinity.  Under the randomization inference framework, we obtain the  asymptotic distributions of estimators used in these methods and the formulas of variance reduction when compared to stratified randomization. Our analysis does not require any modeling assumption regarding the potential outcomes. Moreover, we provide asymptotically conservative variance estimators and confidence intervals for the average treatment effect.  The advantages of the proposed methods are exhibited through an extensive simulation study and a real-data example.

\vspace{12pt}
\end{abstract}

\end{singlespace}

\newpage

\clearpage
\setcounter{page}{1}

\allowdisplaybreaks
\baselineskip=24pt

\begin{singlespace}
\section{Introduction}
\label{sec:intro}

The application of randomized experiments has recently gained increasing popularity in various fields, including industry,  social sciences, and clinical trials \citep[e.g.,][]{box2005,GerberGreen2012,Rosenberger2015}. Often, there are covariates  that are likely to be unbalanced in completely randomized experiments \citep{Fisher1926, Senn1989, Morgan2012,Athey2017}. \citet{Fisher1926} first recognised this issue and introduced the use of blocking, or stratification, for balancing discrete covariates. In stratified randomized experiments, units are divided into strata according to the discrete covariates and complete randomization is conducted within each stratum. Appropriate stratification improves the covariate balance and inference efficiency; see \citet{imai2008variance}, \citet{miratrix2013}, and \citet{Imbens2015} for an overview. 

Whereas stratification balances only discrete covariates, rerandomization is a more powerful tool that excludes allocations causing covariate imbalance. Covariate balance can be measured by a predetermined criterion, and only the allocations that meet this criterion are accepted \citep{Morgan2012}. \citet{Morgan2012} used the Mahalanobis distance of the sample means of the covariates in the treatment and control groups for measuring covariate balance and set a threshold in advance to rule out unsatisfactory allocations. The authors showed that the difference-in-means average treatment effect estimator remains unbiased under the symmetric balance criterion for the treatment and control groups, and that rerandomization enhances efficiency when the treatment effect is additive (i.e., all the units have the same treatment effect) and the covariates are correlated with the potential outcomes. For more general situations,  \citet{Li2018} obtained an asymptotic distribution of the difference-in-means estimator under rerandomization, and developed a method to construct large-sample confidence intervals for the average treatment effect.

Renowned scholars, such as R. A. Fisher, have recommended combining the  rerandomization and stratification methods. This design strategy was summarized by D. B. Rubin as `Block what you can and rerandomize what you cannot'. Recently, \citet{Schultzberg2019} developed a stratified rerandomization design where stratification on binary covariates was followed by rerandomization on continuous covariates. They demonstrated that for binary covariates, stratification is equivalent to rerandomization, and that stratified rerandomization enhances both inference and computation efficiencies  under equal-sized treatment and control groups and an additive treatment effect, or the Fisher sharp null hypothesis. However, when the sizes of the treatment and control groups are not equally sized, or the treatment effect is not additive, especially when the number of strata tends towards infinity, the efficient strategy of stratified rerandomization and its statistical behaviour are unknown. 

The present paper  proposes two rerandomization strategies in stratified randomized experiments and establishes their asymptotic theory by using the Neyman--Rubin potential outcomes model \citep{Neyman:1923,Rubin:1974} and randomization inference framework \citep{kempthorne1955randomization, fcltxlpd2016, zhao2016randomization}, without any modeling assumption regarding the potential outcomes. The proposed methods are termed the overall strategy and the stratum-specific strategy. Both use the  Mahalanobis distance for measuring covariate imbalance. However, the first computes  the overall covariate imbalance and rerandomizes over the entire strata together, and the second computes  the stratum-specific covariate imbalance and rerandomizes within each stratum independently. The overall strategy is flexible and applicable to nearly arbitrary numbers of strata and their sizes, and does not require the same propensity scores (proportions of the treated units) across different strata.
This strategy is essential for stratified experiments with strata containing single treated or control unit and hybrid experiments with both large and small strata. These scenarios can easily appear in modern social science experiments. For instance, multisite trials in education often have several sites, including only a few schools on each site \citep[e.g.,][]{Wills2018}.
The stratum-specific strategy is a straightforward extension of \citet{Li2018}, and is generally more efficient than the first method; however, it requires the number of strata to be fixed with their sizes tending to infinity. 

We prove that, under mild conditions, the stratified difference-in-means estimators are asymptotically unbiased and truncated-normal distributed under both stratified rerandomization strategies. In addition, we show that stratified rerandomization improves, or at least does not degrade, the precision as compared to stratified randomization (SR). We further  provide asymptotically conservative estimators for the variances and confidence intervals under both strategies. Finally, we illustrate the performances of the proposed methods through an extensive simulation study and a real-data example.

\section{Framework, notation, and stratified rerandomization}
\label{sec2}

In stratified randomized experiments with $n$ units, $\p_0  $ discrete covariates and $\p$ additional (discrete or continuous) covariates are collected before the physical implementation of randomization. The units are divided into $\S$ strata according to the $p_0$ discrete covariates, each having $\nk \ge2\ (k=1,\ldots,\S)$ units, such that $n=n_{[1]}+\cdots+n_{[\S]}$. Let $X\in\mathbb{R}^{n\times\p}$ denote an additional covariate matrix, whose $i$th row, denoted by $\Xi^\T$, indicates the observations of the additional covariates of unit $i$. In stratum $k$, $\nkt = \pk\nk$ units are randomly selected and assigned to the treatment group, and the remaining $\nkc = (1-\pk)\nk$ units are assigned to the control group, where $\pk \in (0,1)$ is called the propensity score. The total numbers of treated and control units are $\nt = \sum_{k=1}^{K} \nkt$ and $\nc = \sum_{k=1}^{K} \nkc$, respectively. 
For each unit $i=1,\ldots,n$, let $\Zi$ be the treatment assignment indicator, where $\Zi=1$ if it is assigned to the treatment group and $\Zi=0$ if it is assigned to the control group.  We use $i\in[k]$ to denote the indices taken over the stratum $k$.  Let $Y_i(z)$ be the potential outcomes for unit $i$ under the treatment arm $z\ (z=0,1)$, where $z=1$ indicates treatment and $z=0$ indicates control. The unit level treatment effect is defined as $\taui=\Yi(1)-\Yi(0)$, and the average treatment effect is defined as
$
	\tau= n^{-1} \skS\sumk\taui=\skS\pik\tauk,
$
where $\pik=\nk/n$ is the proportion of stratum size and $\tauk = \nk^{-1} \sumk \taui$ is the stratum-specific average treatment effect in stratum $k\ (k=1,\ldots,\S)$. 

In stratum $k$, the stratum-specific means of covariates and potential outcomes are denoted as
$\bXk= \nk^{-1} \sumk\Xi$ and $\bYk(z)= \nk^{-1} \sumk\Yi(z)$, $z = 0, 1,$ and the stratum-specific variances and covariances are denoted as
$$ \SkY(z)= \frac{1}{\nk - 1} \sumk\{\Yi(z)-\bYk(z)\}^2, \quad  \SkX=   \frac{1}{\nk - 1}  \sumk(\Xi-\bXk)(\Xi-\bXk)^\T ,  $$
$$  \SkXY(z)=  \frac{1}{\nk - 1} \sumk (\Xi-\bXk) \{\Yi(z)-\bYk(z)\} , \quad  \Skt= \frac{1}{\nk - 1} \sumk(\taui-\tauk)^2. $$

Under the stable unit treatment value assumption \citep{Rubin:1980}, for any realised value of $\Zi$, the observed outcome of unit $i$ is $Y_i^\obs = \Zi\Yi(1)+(1-\Zi)\Yi(0)$. For the treatment arm $z=1$, the observed stratum-specific means of the potential outcomes and covariates are denoted as
$   \Yobs_{[k]1} =  \nkt^{-1} \sumk \Zi \Yi(1)$ and $\Xobs_{[k]1} =  \nkt^{-1} \sumk \Zi \Xi$. Similarly, we define $\Yobs_{[k]0}$ and $ \Xobs_{[k]0} $ for the control arm $z=0$. The stratified difference-in-means estimator of the average treatment effect is 
\begin{equation}
	\hat{\tau}=\skS\pik\Big\{ \Yobs_{[k]1}-\Yobs_{[k]0} \Big\}=\skS\pik\htk,
	\label{tauhat}
\end{equation} 
where $\htk = \Yobs_{[k]1}-\Yobs_{[k]0} $ is the difference-in-means estimator of $\tauk$.

This paper proposes two stratified rerandomization criteria, one based on the overall Mahalanobis distance and the other based on the stratum-specific Mahalanobis distance. 

(1) Stratified rerandomization based on the overall Mahalanobis distance. Because covariates can be viewed as potential outcomes that are unaffected by the treatment assignment with zero treatment effect,  the Mahalanobis distance of the  stratified sample means of the covariates under two treatment arms can be used to measure the covariate imbalance. More specifically, denote 
\begin{equation*}
	\htX=\skS\pik  \Big\{ \Xobs_{[k]1} - \Xobs_{[k]0} \Big\}=\skS\pik\htkX,
\end{equation*}
where $\htkX=\Xobs_{[k]1}-\Xobs_{[k]0}$ indicates the difference-in-means of the covariates in stratum $k$. The overall Mahalanobis distance is defined as $M_{\htX}=(\htX)^\T\cov(\htX)^{-1}\htX$. Here, a random assignment is accepted only when  $M_{\htX}<a$, where $a$ is a predetermined threshold.

(2) Stratified rerandomization based on the stratum-specific Mahalanobis distance. When each stratum comprises a large number of units, rerandomizing within each stratum separately and independently can be more efficient than the overall rerandomization. Thus, we use this rerandomization criterion in our study, where the stratum-specific Mahalanobis distance is defined as $\Mk=(\htkX)^\T\cov(\htkX)^{-1}\htkX,\ k=1,\ldots,\S$. Here, a random assignment is accepted only when $\Mk<a_k$, where  $a_k$ is a predetermined threshold for the stratum $k$.

To investigate the asymptotic properties of the above two stratified rerandomization strategies and obtain valid inferences for the average treatment effect, we first establish the joint asymptotic normality of the stratified difference-in-means estimator for vector potential outcomes. Our analysis is conducted under the randomization inference framework, where both $Y_i(z)$ and $\Xi$ are fixed quantities, and randomness originates only from the treatment assignment $Z_i$.

\section{Joint asymptotic normality of stratified difference-in-means estimator}

Let us consider (fixed) $d$-dimensional potential outcomes  $\Ymul_{i}(z)=(R_{i,1}(z),\cdots,R_{i,d}(z))^\T$, $i=1,\ldots,n,$ $z=0,1$. In what follows, $\Ymul_i(z)$ can take the form of $Y_i(z)$, $\Xi$, or $(Y_i(z), \Xi^\T)^\T$. Similar to the definitions established in Section~\ref{sec2}, we can define the vector-form average treatment effect $\tmul$, its stratified difference-in-means estimator $\htmul$, and the covariances of $\Ymul_{i}(z)$ and $\tmul$. 
\begin{proposition}
\label{prop0}
Under stratified randomization, the covariance of $n^{1/2}(\htmul-\tmul)$ is 
\[\Sigmul=\skS\pik \Big \{ \frac{\SkYmul(1)}{\pk}+\frac{\SkYmul(0)}{1-\pk}-\Sktmul \Big\}.\]
\end{proposition}

Because covariates can be considered potential outcomes with no treatment effect, we can apply Proposition~\ref{prop0} to $\Ymul_i(z) = (Y_i(z), \Xi^\T)^\T$ and obtain the following proposition.

\begin{proposition}
\label{prop::covariance}
Under stratified randomization, the covariance of $n^{1/2}(\hat\tau-\tau, \htX^\T)^\T$ is 
\begin{equation}
\label{cov1}
\Sigma=\left(
	\begin{array}{cc}
	\Sigma_{\tau\tau}&\Sigma_{\tau x}\\
	\Sigma_{x\tau}&\Sigma_{xx}\\
	\end{array}\right)
	=\skS\pik \left(
	\begin{array}{cc}
	\frac{\SkY(1)}{\pk}+\frac{\SkY(0)}{1-\pk}-\Skt & \frac{\SkXY^\T(1)}{\pk}+\frac{\SkXY^\T(0)}{1-\pk} \\
	\frac{\SkXY(1)}{\pk}+\frac{\SkXY(0)}{1-\pk} &	\frac{\SkX}{\pk(1-\pk )} \\
	\end{array}
	\right).
\end{equation}
\end{proposition}

To establish the joint asymptotic normality of $\htmul$, the following conditions need to be satisfied. Without further explanation, limits  are taken as $n$ tends to infinity with no restriction on $\S$ and $\nk$. Let $\| \cdot \|_\infty$ denote the infinity norm of a vector, and let $\mathcal{N}(\mu, \Sigma)$ denote a normal distribution with mean $\mu$ and covariance matrix $\Sigma$. 

\begin{condition}\label{cond::propensity}
	For $k= 1,\dots,K$, there exist constants $\pk^{\infty}$ and $c\in(0,0.5)$ such that $\pk^{\infty}\in(c,1-c)$ and $\max_{k=1,\ldots,\S}|\pk -\pk^{\infty}|\to 0$.
\label{cond1pk}
\end{condition}

\begin{condition}\label{cond::max}
	For $z=0,1$, $\max_{k=1,\ldots,\S}\max_{i\in[k]}\|\Ymul_i(z)-\bYmulk(z)\|_\infty^2/n\to 0$.
\label{cond2}
\end{condition}

\begin{condition}\label{cond::covariance}
	The following three matrices have finite limits:
	 $$  \skS\pik \frac{ \SkYmul(1) }{\pk} , \quad \skS\pik   \frac{ \SkYmul(0) }{ 1-\pk }, \quad \skS\pik\Sktmul, $$
	and the limit of $\Sigmul$, denoted as $\Sigmul^\infty$,  is (strictly) positive definite.
\label{cond3}
\end{condition}

\begin{remark}
Condition~\ref{cond::propensity} assumes that the propensity scores for all strata uniformly converge to limits between zero and one. Condition~\ref{cond::max} requires that the maximum squared distance between each component of the potential outcomes and its stratum-specific means, divided by $n$, tends to zero. When $K=1$, Condition~\ref{cond::max} reduces to that proposed  in \citet{fcltxlpd2016} for establishing the finite-population central limit theorem for simple randomization. Condition~\ref{cond::covariance} is a technical condition. When $d=1$, Conditions~\ref{cond::max} and \ref{cond::covariance}  reduce to those proposed  in \citet{Liu2019} for analysing the properties of regression adjustments in stratified randomized experiments. 
\end{remark}

\begin{theorem}
\label{thm1CLT}
Under Conditions \ref{cond1pk}--\ref{cond3} and stratified randomization, $n^{1/2}(\htmul-\tmul)$ converges in distribution to $\mathcal{N}(0,\Sigmul^\infty)$ as $n$ tends to infinity.
\end{theorem}

Theorem~\ref{thm1CLT} provides a normal approximation for the distribution of $\htmul$. It  generalizes the asymptotic normality of the stratified difference-in-means estimator from one-dimensional outcomes \citep{Liu2019} to $d$-dimensional vector outcomes, as well as the result of \citet{fcltxlpd2016} from simple randomization to stratified randomization. The generalization is straightforward in case of a fixed $K$ with each $\nk$ tending to infinity, but novel for an asymptotic regime where both $K$ and $\nk$ can tend to infinity, including the special cases of paired randomized experiments, finely stratified randomized experiments \citep{fogarty2018finely}, and threshold blocking design \citep{Higgins2016}.

To construct a large-sample confidence set of $\tmul$ using Theorem~\ref{thm1CLT}, we propose an asymptotically conservative estimator of $\Sigmul$.
When there are at least two treated and two control units in each stratum, we can replace the stratum-specific population covariances, $ \SkYmul(1)$ and $ \SkYmul(0)$, by their  sample analog and ignore the term $\Sktmul$ to conservatively estimate  $\Sigmul$. When there is only one treated or control unit in some strata, we can generalize the variance estimator proposed by \citet{Pashley2017} to handle $d$-dimensional potential outcomes.  Define $s^2_{[k]R}(z)$ as the sample covariance of $R_{i}$ in stratum $k$ under treatment arm $z$ when $n_{[k]z} \ge 2$, $z=0,1$. Let $A_{ss} = \{ k: \nkt = 1 \ \textnormal{or} \ \nkc = 1 \}$ be the set of small strata, where ``ss" stands for ``small strata".  Let $n_{ss} = \sum_{k \in A_{ss} } \nk$ be the total number of units in small strata, and let $\tau_{R,ss} = \sum_{k \in A_{ss}} (\nk / n_{ss} ) \trk$ and $\hat \tau_{R,ss} = \sum_{k \in A_{ss}} (\nk / n_{ss} ) \htkmul$. Throughout the paper, we assume that $\nk < n_{ss} / 2$ for all $k \in A_{ss}$. Then an asymptotically conservative estimator of $\Sigmul$ is 
\begin{eqnarray}
\label{eq:Sigmul}
\hat \Sigma_{R} &= &   \sum_{k \notin A_{ss} } \pik\Big\{\frac{s^2_{[k]R}(1)}{\pk}+\frac{s^2_{[k]R}(0)}{1-\pk}\Big\} \nonumber  \\
& & +  \left(  \frac{n_{ss} }{n} \right)^2  \sum_{k \in A_{ss} }  \frac{ n \nk^2 } { ( n_{ss} - 2 \nk  ) \Big( n_{ss} + \sum_{h \in A_{ss}} \frac{ n_{[h]}^2  }{ n_{ss} - 2  n_{[h]} }  \Big)  }  ( \htkmul -  \hat \tau_{R,ss} )( \htkmul -  \hat \tau_{R,ss} )^\T.\nonumber
\end{eqnarray}

\begin{condition}
\label{cond::ss}
There exists a constant $C$ such that $n^{-1}\sum_{i=1}^n R_i(z)^\T R_i(z) \leq C,\ z=0,1 $.
\end{condition}

\begin{theorem}
\label{thm:cov-est}
Under Conditions \ref{cond1pk}--\ref{cond::ss} and stratified randomization, $E( \hat  \Sigma_{R} ) = \tilde{\Sigma}_{R} $ and $\hat \Sigma_{R} - \tilde{\Sigma}_{R}$ converges to zero in probability, where
\begin{eqnarray}
\tilde{\Sigma}_{R} & = &  \Sigma_{R} + \sum_{k \notin A_{ss}} \pik\Sktr  \nonumber \\
&& + \frac{n_{ss}^2}{n} \sum_{k \in A_{ss} }  \frac{  \nk^2 } { ( n_{ss} - 2 \nk  ) \Big( n_{ss} + \sum_{h \in A_{ss}} \frac{ n_{[h]}^2  }{ n_{ss} - 2  n_{[h]} }  \Big)  }  ( \trk - \trs )( \trk - \trs )^\T \nonumber
\end{eqnarray}
with $\tilde{\Sigma}_{R} -  \Sigma_{R} $ being positive semidefinite.
\end{theorem}


Next, we apply Theorem \ref{thm1CLT} to $\Ymul_i(z) = (Y_i(z), \Xi^\T)^\T$. The following conditions should be met.

\begin{condition}
	For each treatment arm $z=0,1$, 
	$$ \max_{k=1,\ldots,\S}\max_{i\in[k]}\{\Yi(z)-\bYk(z)\}^2/n\to 0,\ \text{and}\ \max_{k=1,\ldots,\S}\max_{i\in[k]}\|\Xi-\bXk\|_\infty^2/n\to0 . $$
\label{cond4}
\end{condition}

\vspace{-1cm}

\begin{condition}
	The following two matrices have finite limits:
	\begin{eqnarray*}
	\skS\frac{\pik}{\pk} \left(
	\begin{array}{cc}
	\SkY(1) & \SkXY^\T(1)\\
	\SkXY(1) & \SkX\\
	\end{array}
	\right),\quad 
	\skS\frac{\pik}{1-\pk} \left(
	\begin{array}{cc}
	\SkY(0) & \SkXY^\T(0)\\
	\SkXY(0) & \SkX\\
	\end{array}
	\right),
	\end{eqnarray*}
	$\skS\pik\Skt$ has a limit, and the limit of $\Sigma$, denoted by $\Sigma^\infty$,  is (strictly) positive definite.
\label{cond5}
\end{condition}

\begin{corollary}
	\label{cor1}
	Under stratified randomization, if Conditions \ref{cond1pk}, \ref{cond4}, and \ref{cond5} hold, then $n^{1/2}(\hat{\tau}-\tau, \htX^\T)^\T$
converges in distribution to $\mathcal{N}(0,{\Sigma}^\infty)$ as $n$ tends to infinity.
\end{corollary}


\section{Asymptotics of stratified rerandomization}

\subsection{Stratified rerandomization based on the overall Mahalanobis distance}
\label{SRRoM}

According to Proposition~\ref{prop::covariance}, the overall Mahalanobis distance
$
	M_{\htX} =(\htX)^\T\cov(\htX)^{-1}\htX= n (\htX)^\T \Sigma_{xx}^{-1} \htX,
$
where $\Sigma_{xx}=\skS\pik\SkX/\{\pk(1-\pk)\}$ is the lower right block matrix of $\Sigma$ known at the design stage of the experiment. Denote $\Mtau = \{ (Z_1,\dots,Z_n): \  M_{\htX} < a   \}$ as an event that an assignment is accepted under the \underline{s}tratified \underline{r}e\underline{r}andomization based on the \underline{o}verall \underline{M}ahalanobis distance  $M_{\htX}$, which is abbreviated as SRRoM. 

\begin{proposition}
	\label{prop1pa}
	Under SRRoM, if Conditions \ref{cond1pk}, \ref{cond4}, and \ref{cond5} hold, then the asymptotic probability of accepting a random assignment is $p_a=\pr(\chi^2_{\p}<a)$, where $\chi^2_{\p}$ represents a chi-square distribution with $\p$ degrees of freedom.
\end{proposition}

The asymptotic distribution of $\ntau\mid\Mtau$ can be derived from Corollary \ref{cor1}. Let $R^2=\cov(\hat\tau,\htX)\var(\htX)^{-1}\cov(\htX,\hat\tau)/\var(\hat\tau)=\Sigma_{\tau x}\Sigma_{xx}^{-1}\Sigma_{x\tau}/\Sigma_{\tau\tau}$ be the squared multiple correlation between $\hat\tau$ and $\htX$ under stratified randomization. Let us denote independent random variables as $\epsilon_0\sim\mathcal{N}(0,1)$ and $L_{\p,a}\sim(D_1\mid D^\T D<a)$, where $D=(D_1,\ldots,D_{p})^\T$ is a  $\p$-dimensional $\mathcal{N}(0,I)$ distributed random vector.
In what follows, the notation $\dot{\sim}$ will be used for two sequences of random vectors converging to the same distribution as $n$ tends to infinity.

\begin{theorem}
	\label{thm2}
	Under SRRoM, if Conditions \ref{cond1pk}, \ref{cond4}, and \ref{cond5} hold, then
	\begin{equation}
	n^{\frac{1}{2}}(\hat{\tau}-\tau)\mid\Mtau\ \dot{\sim}\  \Sigma_{\tau\tau}^{\frac{1}{2}}\big\{(1-R^2)^{\frac{1}{2}}\epsilon_0+(R^2)^{\frac{1}{2}}L_{\p,a}\big\}. \nonumber
	\label{thm2rv}
	\end{equation}
\end{theorem}

When the number of strata is fixed with their sizes tending to infinity, Theorem~\ref{thm2} becomes a direct extension of the asymptotic result of rerandomization in completely randomized experiments \citep{Li2018}, and can also be obtained from the asymptotic theory of rerandomization for tiers of covariates \citep{morgan2015rerandomization,Li2018}. The novelty of this theorem lies in the fact that it makes few restrictions on the number of strata and their sizes, allowing the number of strata to tend to infinity with their sizes fixed. According to Theorem~\ref{thm2}, the asymptotic distribution of the stratified estimator under SRRoM is a truncated-normal, which has the same formula as that of the difference-in-means estimator under merely rerandomization; however, both $\Sigma_{\tau\tau}$ and $R^2$ have distinct meanings due to different sources of randomness.

Theorem \ref{thm2} implies the asymptotic unbiasedness and improvement in the efficiency of stratified rerandomization, as summarized in the next corollary. Let $v_{\p,a}=\pr(\chi^2_{\p+2}\le a)/\pr(\chi^2_{\p}\le a)$ denote the variance of $L_{\p,a}$.

\begin{corollary}
	\label{cor3}
	Under SRRoM, if Conditions \ref{cond1pk}, \ref{cond4}, and \ref{cond5} hold, then $\hat\tau$ is an asymptotically unbiased estimator of  $\tau$. The asymptotic variance of $\ntau$ under SRRoM is the limit of $\Sigma_{\tau\tau}\big\{1-(1-v_{\p,a})R^2\big\}$, whereas the percentage of reduction in asymptotic variance compared to stratified randomization is the limit of $(1-v_{\p,a})R^2$.
\end{corollary}

\begin{remark}
\citet{Schultzberg2019} proposed a stratified rerandomization strategy using the Mahalanobis distance  $ M_{\tilde \tau_X} = (\tilde \tau_X )^\T\cov(\htX)^{-1} \tilde \tau_X  $, where
$
\tilde \tau_X = ( 1 / \nt ) { \sum_{i=1}^{n} \Zi \Xi} -  ( 1 / \nc ) \sum_{i=1}^{n} ( 1 - \Zi ) \Xi
$
is the difference-in-means estimator of the covariates. They showed that there is no guarantee that the overall difference-in-mean estimator is unbiased (non-asymptotically) unless $\nkt=\nkc$ in each stratum. We generalize this result using asymptotic theory in the Supplementary Material. We show that the overall difference-in-means estimator can be asymptotically biased if the propensity scores differ across strata. The aim of \citet{Schultzberg2019} was to show that substantial computational and efficiency gains  can be obtained by first stratifying and then finding the ``optimal" allocations within each stratum. For this procedure, the asymptotic theory in \citet{Li2018} has not yet been proved to be valid, which is why \citet{Schultzberg2019} suggested a Fisher randomization test for inferences.
\end{remark}

As the threshold $a$ tends  to $0$, the asymptotic variance $\Sigma_{\tau\tau}\big\{1-(1-v_{\p,a})R^2\big\}$ tends to its minimum value 
$\Sigma_{\tau\tau} ( 1 - R^2) = n \times \min_{\gamma} E\left( \hat \tau - \tau - \htX ^\T \gamma   \right)^2,$
which is equal to the variance of the errors in the linear projection of $\sqrt{ n} (\hat \tau - \tau) $ onto $ \sqrt{n} \htX$. Let us define the projection coefficient vector as
$
\gamma_{\opt} = \argmin_{\gamma} E\left( \hat \tau - \tau - \htX ^\T \gamma   \right)^2 =  \Sigma_{xx}^{-1}\Sigma_{x\tau}.
$
This motivates us to consider a covariate-adjusted estimator $\hat \tau_{\adj} = \hat \tau -   \htX ^\T \gamma_{\opt} $, which has the smallest asymptotic variance among the adjusted estimators of the same form. Asymptotically, the efficiency of $\hat \tau_{\adj}$ under stratified randomization is equal to that of $\hat \tau$ under SRRoM with a threshold $a \rightarrow 0$. That is, stratified rerandomization can be viewed as covariate/regression adjustment in design. In practice, however, $\gamma_{\opt}$ is usually unknown because it depends on the unknown potential outcomes. We need to derive a consistent estimator of $\gamma_{\opt}$. Define $s_{[k]XY}(1)$ as the sample covariance between $X_{i}$ and $Y_{i}(1)$ in stratum $k$ under the treatment  when $\nkt \geq 2$ and define it as $ \{ \nk /  (\nk - 1)   \} \sumk   Z_i   ( X_i -  \bar{X}_{[k]} )  Y_i^\obs  $ when $\nkt = 1$. The intuition comes from the fact that, when $\nkt = 1$,
$$
E \bigg\{ \frac{\nk}{\nk - 1} \cdot \sumk Z_i (X_i-\bXk) Y_i^\obs \bigg\} =  \frac{1}{\nk -1} \sumk (X_i-\bXk) Y_i(1) = \SkXY(1) .
$$
This relies on our knowledge of $\bXk$. Similarly, we can define $s_{[k]XY}(0)$. Then a consistent estimator of $\gamma_{\opt}$ is
$$
\hat \gamma_{\opt} =  \Sigma_{xx}^{-1} \skS \pik \left(  \frac{ s_{[k]XY}(1) }{\pk} +  \frac{ s_{[k]XY}(0) }{ 1 - \pk}   \right).
$$

\begin{condition}
\label{cond::small-strata}
There exists a constant $C$ such that $n^{-1}\sum_{i=1}^n \Yi^2(z) \leq C,\ z=0,1 $.
\end{condition}
\begin{proposition}
\label{prop::gamma}
Under stratified randomization or SRRoM, if Conditions \ref{cond1pk} and \ref{cond4}--\ref{cond::small-strata} hold, then $\hat \gamma_{\opt}  - \gamma_{\opt}$ converges to zero in probability.
\end{proposition}

\begin{remark}
When there are at least two treated and two control units in each stratum, Proposition~\ref{prop::gamma} still holds without Condition~\ref{cond::small-strata}.
\end{remark}

\begin{remark}
Let  $I_{i \in [k]}$ be the stratification indicator. With unequal propensity scores, we cannot obtain $\hat \tau_{\adj}$ by running  a single regression of $Y_i$ on $Z_i$, $I_{i \in [k]} $, $\Xi$, and/or their interactions, while with equal propensity scores, let $\pk = e$,  we have
$
\gamma_{\opt}   =  (1 - e ) \beta(1) + e \beta(0), 
$
where $\beta(z) =   ( \skS \pik S_{[k] XX}  ) ^{-1} \skS \pik S_{[k]XY}(z)$ is the projection coefficient of $\Xi$ in the weighted projection of $Y_i(z)$ onto $I_{i \in [k]}  $ and  $\Xi$ with weights $\pik / (\nk - 1)$. When $n_{[k]z} \geq 2$ for all $k$ and $z$, we can estimate $\beta(z)$ using the sample, i.e., by running a weighted regression of $Y_i(z)$ on $I_{i \in [k]}  $ and $\Xi$ under  treatment arm $z$ with weights $\pik / (n_{[k]z} - 1)$. Let $\hat \beta(z)$ be  the OLS estimator of the coefficient of $\Xi$. Then we can derive another consistent estimator of $ \gamma_{\opt}  $:
$ \tilde \gamma_{\opt} =   (1 - e )  \hat \beta(1) + e \hat \beta(0). $
The resulting regression-adjusted estimator is
$$
\hat \tau_{\adj}   = \hat \tau -   \htX ^\T \tilde \gamma_{\opt} = \skS \pik  \Big[  \big\{ \Yobs_{[k]1}  -  (  \Xobs_{[k]1} - \bXk ) ^\T \hat \beta(1) \big\}   -  \big\{ \Yobs_{[k]0}  -  (  \Xobs_{[k]0} - \bXk ) ^\T \hat \beta(0) \big\}  \Big].
$$
As shown by \citet{Liu2019}, $\hat \tau_{\adj} $ is the OLS estimator of $Z_i$ in a weighted regression of $Y_i$ on $Z_i$, $I_{i \in [k]}  $, $\Xi$,  $Z_i \times ( I_{i \in [k]} - \pik ) $, and $Z_i \times ( \Xi - \bXk ) $.
\end{remark}

\begin{remark}
In completely randomized experiments, rerandomization followed by regression adjustment can further improve the estimation efficiency if the analyzer has access to more covariates than the designer \citep{Li2020}. This property is inherited in stratified randomized experiments. Under SRRoM, the regression-adjusted estimator $\hat \tau_{\adj} $  using all the covariates that are used in the design stage and additional covariates that are collected in the analysis stage can further improve the efficiency.  
\end{remark}




Next, we compare the quantile ranges of $\ntau$ under SRRoM and stratified randomization. Let $\nu_\xi(R^2,p_a,p)$ be the $\xi$th quantile of the random variable $(1-R^2)^{1/2}\epsilon_0+(R^2)^{1/2}L_{\p,a}$, then under SRRoM, the asymptotic $\CI$ quantile range of $\ntau$ is the limit of 
\begin{equation*}
	\big[\Sigma_{\tau\tau}^{\frac{1}{2}}\nu_{\alpha/2}(R^2,p_a,p),\ \Sigma_{\tau\tau}^{\frac{1}{2}}\nu_{1-\alpha/2}(R^2,p_a,p)\big],
\end{equation*}
for the length of which we present the following corollary.

\begin{corollary}
	\label{thm3}
	If Conditions \ref{cond1pk}, \ref{cond4}, and \ref{cond5} hold, then the length of the $\CI$ quantile range of the asymptotic distribution of $\ntau$ under SRRoM is less than or equal to  that under stratified randomization; this length is non-increasing in $R^2$ and non-decreasing in $p_a$ and $\p$.
\end{corollary}

\begin{remark}
This result is similar with Theorem 2 in \citet{Li2018} where they compared the lengths of the quantile ranges of $\ntau$ under rerandomization using Mahalanobis distance and complete randomization, while we compare the lengths under SRRoM and stratified randomization. 
Corollary~\ref{thm3} suggests that a smaller value of $p_a$ leads to better improvement; however, setting $p_a$ to a very small value can be problematic if   very few assignments are acceptable, which renders little power to   randomization inference. Thus, how to choose the value of $p_a$ remains an open issue and should be investigated in the future.   In practice, we suggest to choose a small value of $p_a$, for example, $p_a = 0.001$.
\end{remark}

As the experimental results yield only part of the potential outcomes, the precise variance of $\hat\tau$ and the theoretical confidence interval of $\tau$ are unknown; however,  we can construct asymptotically conservative estimators. According to Proposition~\ref{prop::gamma}, a consistent estimator of $  \Sigma_{x\tau} $ is
\begin{equation*}
	\hat \Sigma_{x\tau}=\hat \Sigma_{\tau x}^\T=\skS\pik\Big\{\frac{\skXY(1)}{\pk}+\frac{\skXY(0)}{1-\pk}\Big\}. 
\end{equation*}
Applying Theorem~\ref{thm:cov-est} to $Y_i(z)$, we can estimate  $\Sigma_{\tau\tau}$ as follows. 
Define $s^2_{[k]Y}(z)$ as the sample variance of $Y_{i}$'s in stratum $k$ under treatment arm $z$ when $n_{[k]z} \ge 2$, $z=0,1$. Let $\hat \tau_{ss} = \sum_{k \in A_{ss}} (\nk / n_{ss} ) \htk$. Then a conservative estimator of $\Sigma_{\tau\tau}$ is 
{\footnotesize
\begin{eqnarray}
\hat \Sigma_{\tau\tau} &= &   \sum_{k \notin A_{ss} } \pik\Big\{\frac{\skY(1)}{\pk}+\frac{\skY(0)}{1-\pk}\Big\} +  \left(  \frac{n_{ss} }{n} \right)^2  \sum_{k \in A_{ss} }  \frac{ n \nk^2 } { ( n_{ss} - 2 \nk  ) \Big( n_{ss} + \sum_{h \in A_{ss}} \frac{ n_{[h]}^2  }{ n_{ss} - 2  n_{[h]} }  \Big)  }  ( \htk -  \hat \tau_{ss} )^2.\nonumber
\end{eqnarray}
}
Let $\hat{R}^2=\hat{\Sigma}_{\tau x}\Sigma_{xx}^{-1}\hat \Sigma_{x\tau}/\hat \Sigma_{\tau\tau}$ and  $\nu_\xi(\hat R^2,p_a,p)$ be the $\xi$th quantile of $(1-\hat R^2)^{1/2}\epsilon_0+(\hat R^2)^{1/2}L_{\p,a}$. 
\begin{theorem}
	\label{thmvaro}
	Under SRRoM, if Conditions \ref{cond1pk} and \ref{cond4}--\ref{cond::small-strata} hold, then $\hat \Sigma_{\tau\tau}\{1-(1-v_{\p,a})\hat{R}^2\}$ is an asymptotically conservative estimator for the asymptotic variance of $\ntau$ and
	\begin{equation*}
		\big[\hat\tau-(\hat \Sigma_{\tau\tau}/n)^{\frac{1}{2}}\nu_{1-\alpha/2}(\hat R^2,p_a,p),\ \hat\tau-(\hat \Sigma_{\tau\tau}/n)^{\frac{1}{2}} \nu_{\alpha/2}(\hat R^2,p_a,p)\big]
	\end{equation*}
	is an asymptotically conservative $\CI$ confidence interval of $\tau$.
\end{theorem}
\begin{remark}
Based on \citet{Pashley2017}, we obtain another conservative variance estimator, which is applicable when there exist at least two strata of each (small) stratum size; see the  Supplementary Material for detailed discussions.
\end{remark}
\begin{remark}
The above results can be applied to paired randomized experiments, finely stratified randomized experiments \citep{fogarty2018finely}, and threshold blocking designs \citep{Higgins2016}. Moreover, when there are at least two treated and two control units in each stratum, we do not require Condition~\ref{cond::small-strata}.
\end{remark}



\subsection{Stratified rerandomization based on the stratum-specific Mahalanobis distance}

In the special case where $\S$ is fixed and all $\nk$'s tend to infinity, we can rerandomize in each stratum separately and independently.  Let $\MI = \{ (Z_1,\dots,Z_n): \  \Mk < a_k, \ k=1,\dots,\S   \}$ denote an event in which an assignment is accepted under the \underline{s}tratified \underline{r}e\underline{r}andomization based on the \underline{s}tratum-specific \underline{M}ahalanobis distance $\Mk$,  which is abbreviated as SRRsM. In this section, we assume that $\S$ is fixed and $\nk\to\infty,$ $k=1,\ldots,\S$ as $n\to\infty$ unless stated otherwise.

Since $ \ntau=\skS\pik^{1/2}\nk ^{1/2}(\htk-\tauk)$, and each stratum is rerandomized independently under SRRsM, we can simply apply the asymptotic distribution of $\nk ^{1/2}(\htk-\tauk)$ under complete rerandomization \citep{Li2018} to derive the asymptotic distribution of $\ntau$. 

\begin{condition}
	\label{cond7}
	For each $k=1,\ldots,\S$, as $\nk\to\infty$, $\SkY(z)\ (z=0,1)$ and $\Skt$ have finite limits; the limit of $\var\{ \nk ^{1/2}(\htk-\tauk) \}$ is positive; $\SkX$ converges to a (strictly) positive definite matrix; and $\SkXY(z)\ (z=0,1)$ converges to finite limits.
\end{condition}

\begin{proposition}
	\label{prop3pa}
	Under SRRsM, if Conditions \ref{cond1pk}, \ref{cond4}--\ref{cond5}, and \ref{cond7} hold, then the asymptotic probability of accepting a random assignment is $\prod_{k=1}^\S p_{a_k} = \prod_{k=1}^\S\pr(\chi^2_{\p}<a_k)$.
\end{proposition}

Let us denote the covariance matrix of $\nk^{1/2}\{ \htk-\tauk, (\htkX - \taukX)^\T \}^\T$ as
\begin{equation}
	\label{cov2}
	\left(
	\begin{array}{cc}
	\Vktt&\Vktx\\
	\Vkxt&\Vkxx\\
	\end{array}\right)
	=\left(
	\begin{array}{cc}
	\frac{\SkY(1)}{\pk}+\frac{\SkY(0)}{1-\pk}-\Skt & \frac{\SkXY^\T(1)}{\pk}+\frac{\SkXY^\T(0)}{1-\pk} \\
	\frac{\SkXY(1)}{\pk}+\frac{\SkXY(0)}{1-\pk} &	\frac{\SkX}{\pk(1-\pk)}\\
	\end{array}
	\right),
\end{equation}
and let $\Rk=\Vktx\Vkxx^{-1}\Vkxt/\Vktt$ be the squared correlation between $\htk$ and $\htkX$ under stratified randomization. Let $\epsilon_0$ be a $\mathcal{N}(0,1)$ distributed random variable and let $L_{\p,a_1}^1,\ldots,L_{\p,a_\S}^\S$ be independent and $L_{\p,a_k}^k\sim L_{\p,a_k}$ for $k=1,\ldots,\S$, where $L_{\p,a_k}$ is defined in Section \ref{SRRoM}. Suppose that  $\epsilon_0$ and $L_{p,a_1}^1,\ldots,L_{p,a_\S}^\S$ are independent.

\begin{theorem}
	\label{thm5}
	Under SRRsM, if Conditions \ref{cond1pk}, \ref{cond4}--\ref{cond5}, and \ref{cond7} hold, then
	\begin{equation}
		\label{eqRMR}
		n^{\frac{1}{2}}(\hat\tau-\tau)\mid\MI\ \dot{\sim}\ \Big\{\skS\pik\Vktt(1-\Rk)\Big\}^{\frac{1}{2}}\epsilon_0+\skS(\pik\Vktt \Rk)^{\frac{1}{2}}L_{p,a_k}^k.
	\end{equation}
	
\end{theorem}

The asymptotic unbiasedness of $\hat\tau$, asymptotic variance of $\ntau$ under SRRsM, and variance reduction are summarized in the following corollary. 

\begin{corollary}
	\label{cor4}
	Under SRRsM, if Conditions \ref{cond1pk}, \ref{cond4}--\ref{cond5}, and \ref{cond7} hold, then $\hat\tau$ is an asymptotically unbiased estimator of $\tau$, and the asymptotic variance of $\ntau$ is the limit of $\skS\pik\Vktt\big\{1-(1-v_{p,a_k})\Rk\big\}$, and the percentage of reduction in asymptotic variance compared to stratified randomization is the limit of $\skS\pik\Vktt(1-v_{\p,a_k})\Rk / \Sigma_{\tau\tau}$.
\end{corollary}

SRRoM is also applicable in this case, whereas intuitively, SRRsM achieves better covariance balance because it balances covariates in each stratum.
According to Propositions \ref{prop1pa} and \ref{prop3pa}, asymptotically, the proportions  of all possible assignments $p_a$ and $\prod_{k=1}^\S p_{a_k}$ are acceptable under SRRoM and SRRsM, respectively. Therefore, if we use identical thresholds, that is, $a_1=\cdots=a_\S=a$, SRRsM appears stricter than SRRoM because $\prod_{k=1}^\S p_{a_k}=(p_a)^\S<p_a$. 

\begin{theorem}
	\label{prop3}
	When the thresholds $a_1,\ldots,a_\S$ and $a$ are identical or tend to $0$, the asymptotic variance of $\ntau$ under SRRsM is smaller than or equal to that under SRRoM. Particularly, $\skS\pik\Vktt\big\{1-(1-v_{\p,a})\Rk\big\}\le \Sigma_{\tau\tau}\big\{1-(1-v_{\p,a})R^2\big\}$, where the equality holds if and only if $\Vkxx^{-1}\Vkxt=\Sigma_{xx}^{-1}\Sigma_{x\tau}$ for $k=1,\ldots,\S$.  
\end{theorem}

Theorem~\ref{prop3} implies that SRRsM  improves the efficiency of SRRoM in the situation where there are only a few large strata and the thresholds $a_1,\ldots,a_\S$ and $a$ are identical or tend to $0$. The only exception (that is, they have the same efficiency) is the case that the strata are homogeneous in the sense that the stratum-specific projection coefficients $\Vkxx^{-1}\Vkxt$ ($k=1,\dots,K$) are the same as the overall projection coefficients $\Sigma_{xx}^{-1}\Sigma_{x\tau}$, when projecting the treatment effect onto the covariates. In other situations, the relative reduction in asymptotic variance is related, in a complicated form, to $v_{\p,a}$, $v_{\p,a_k}$'s, and the covariance matrices defined in \eqref{cov1} and \eqref{cov2}. In our simulation study,  the SRRsM with $ p_{a_k} = (p_a)^{1/\S}$ (which ensures the same acceptance probabilities) performs better than the SRRoM when there are a few heterogeneous strata. In contrast,  when there exist many small strata, SRRsM performs worse than SRRoM, even with $ p_{a_k} = p_a$.

Now we  compare the quantile ranges of $\ntau$ under SRRsM and stratified randomization. Denote $q_\xi(R_{[1]}^2,\ldots,R_{[\S]}^2,p_{a_1},\ldots,p_{a_\S},\p)$ as the $\xi$th quantile of the random variable on the right hand side of (\ref{eqRMR}), then the asymptotic $\CI$ quantile range of $\ntau$ under SRRsM is the limit of
\[ [q_{\alpha/2}(R_{[1]}^2,\ldots,R_{[\S]}^2,p_{a_1},\ldots,p_{a_\S},\p),\ q_{1-\alpha/2}(R_{[1]}^2,\ldots,R_{[\S]}^2,p_{a_1},\ldots,p_{a_\S},\p)].\]

\begin{corollary}
	\label{thm6}
	Under SRRsM, if Conditions \ref{cond1pk}, \ref{cond4}--\ref{cond5}, and \ref{cond7} hold, then the length of the $\CI$ quantile range of the asymptotic distribution of $\ntau$ is less than or equal to that under stratified randomization, with the length non-increasing in $R_{[1]}^2,\ldots,R_{[\S]}^2$ and non-decreasing in $p_{a_1},\ldots,p_{a_\S}$ and $\p$.
\end{corollary}

To obtain a valid inference of $\tau$ based on $\ntau$ under SRRsM, we need to estimate the asymptotic variance and quantile range. To achieve this, we follow \citet{Li2018}. Let 
\begin{equation*}
	s^2_{[k]\tau|X}=\{\skXY(1)-\skXY(0)\}^\T (\SkX)^{-1}\{\skXY(1)-\skXY(0)\}
\end{equation*}
be an estimator of the variance of the linear projection of $\tau$ on $X$ in stratum $k$. Then, $\Vktt$ is estimated by $\hVktt=\pk ^{-1}\skY(1)+(1-\pk)^{-1}\skY(0)-s^2_{[k]\tau|X}$. Let $s^2_{[k]Y|X}(z)=\skYX(z)\SkX^{-1}\skXY(z)$ be the sample variance of linear projection of $Y$ on $X$. Then, the estimator of $\Rk$ is
\begin{equation*}
	\hRk=\hVktt^{-1}\Big\{\frac{s^2_{[k]Y|X}(1)}{\pk}+\frac{s^2_{[k]Y|X}(0)}{1-\pk}-s^2_{[k]\tau|X}\Big\},
\end{equation*}
which is set to 0 if the right hand side is negative.

With the above-constructed estimators, the asymptotic distribution of $\ntau$ can be estimated conservatively by
$
	\Big\{\skS\pik\hVktt(1-\hRk)\Big\}^{\frac{1}{2}}\epsilon_0+\skS(\pik\hVktt\hRk)^{\frac{1}{2}}L_{\p,a}^k. 
$
Let $\hat q_\xi$ be the $\xi$th quantile of the above random variable.

\begin{theorem}
	\label{thmvars}
	Under SRRsM, if Conditions \ref{cond1pk}, \ref{cond4}--\ref{cond5}, and \ref{cond7} hold, then $\skS\pik\hVktt\big\{1-(1-v_{\p,a})\hRk\big\}$ is an asymptotically conservative estimator for the asymptotic variance of $\ntau$ and $ [\hat\tau-n^{-1/2}\hat q_{1-\alpha/2},\ \hat\tau-n^{-1/2}\hat q_{\alpha/2}]$ is an asymptotically conservative $\CI$ confidence interval of $\tau$.
\end{theorem}

\section{Simulation study}

We conduct a simulation study to evaluate the finite-sample performance of the point and interval estimators for the average treatment effect under stratified rerandomization strategies, SRRoM and SRRsM, and compare them with those under stratified randomization. Data are generated from the  model: $ \Yi(z)=\Xi^\T\beta_{z1}+\mathrm{exp}(\Xi^\T\beta_{z2})+\epsilon_{i}(z),\ i=1,\ldots,n,\ z=0,1,$
where the covariate vectors $\Xi$'s are eight-dimensional vectors drawn independently from normal distribution with mean zero and covariance matrix $\Sigma$, whose entries $\Sigma_{ij}=0.5^{|i-j|}$, $i,j=1,\dots,8$, and the disturbances $\epsilon_{i}(z)$ are normally distributed with mean zero and variance $10$. The $j$th components of the coefficients are generated independently from the distributions:
$\beta_{11,j}\sim t_3,\ \beta_{12,j}\sim0.1t_3,\ \beta_{01,j}\sim\beta_{11,j}+t_3,\ \beta_{02,j}\sim\beta_{12,j}+0.1t_3,\ j=1,\ldots,8, $
where $t_3$ denotes $t$ distribution with three degrees of freedom. 

The number of strata $\S$ and strata sizes $\nk$ are set in four cases: Case 1, there are many small strata, with $\S=25,\ 50,\ 100$ and $\nk =10$; Case 2, there are many small strata and two large strata, with $\S=10+2,\ 20+2,\ 50+2$, and $\nk =10$ for small strata and $\nk = 100$ for the two large strata; Case 3, there are two large homogeneous strata, with $\S=2$ and $\nk =100,\ 200,\ 500$; Case 4, there are two large heterogeneous strata where the coefficients $\beta_{z1}$ and $\beta_{z2}$ are generated independently for each stratum, with $\S=2$ and $\nk=100,\ 200,\ 500$.

For the given $\S$ and $\nk$, first, we generate the covariates and potential outcomes, and randomly assign $\nkt$ units in stratum $k$ to the treatment group, where the propensity scores are equal, $\pk =0.5\ (k=1,\ldots,\S)$, or unequal, $\pk =0.4\ (k\le \S/2),\ 0.6\ (k>\S/2)$. If the covariate balance criterion is not met, we generate new assignments until we find an assignment that meets the criterion. Then, based on the assignment, we compute the stratified difference-in-means estimator and the  $95\%$ confidence interval. In stratified randomization, we use $[\hat\tau-(\hat \Sigma_{\tau\tau}/n)^{1/2}z_{1-\alpha/2},\ \hat\tau-(\hat \Sigma_{\tau\tau}/n)^{1/2} z_{\alpha/2}]$ as the conservative  confidence interval of $\tau$, where $z_\xi$ is the $\xi$th quantile of a standard normal distribution.
The preceding process of allocation and computation is repeated for $10^4$ times to evaluate the bias, standard deviation, root mean squared error (RMSE), mean confidence interval length, and empirical coverage probability under stratified randomization and  stratified rerandomization. The threshold for SRRoM is set such that $p_a=0.001$, and the thresholds for SRRsM are set such that $p_{a_k}=(0.001)^{1/\S}$ for a fair comparison or $p_{a_k}=0.001$ for an unfair comparison, $k=1,\ldots,\S$. For a stratum of size ten, there are only 252 possible assignments, and SRRsM sometimes rejects all possible assignments under an unfair comparison. In this case, we perform stratified randomization instead of SRRsM.


\begin{figure}[ht]
\begin{center}
\includegraphics[width=6in]{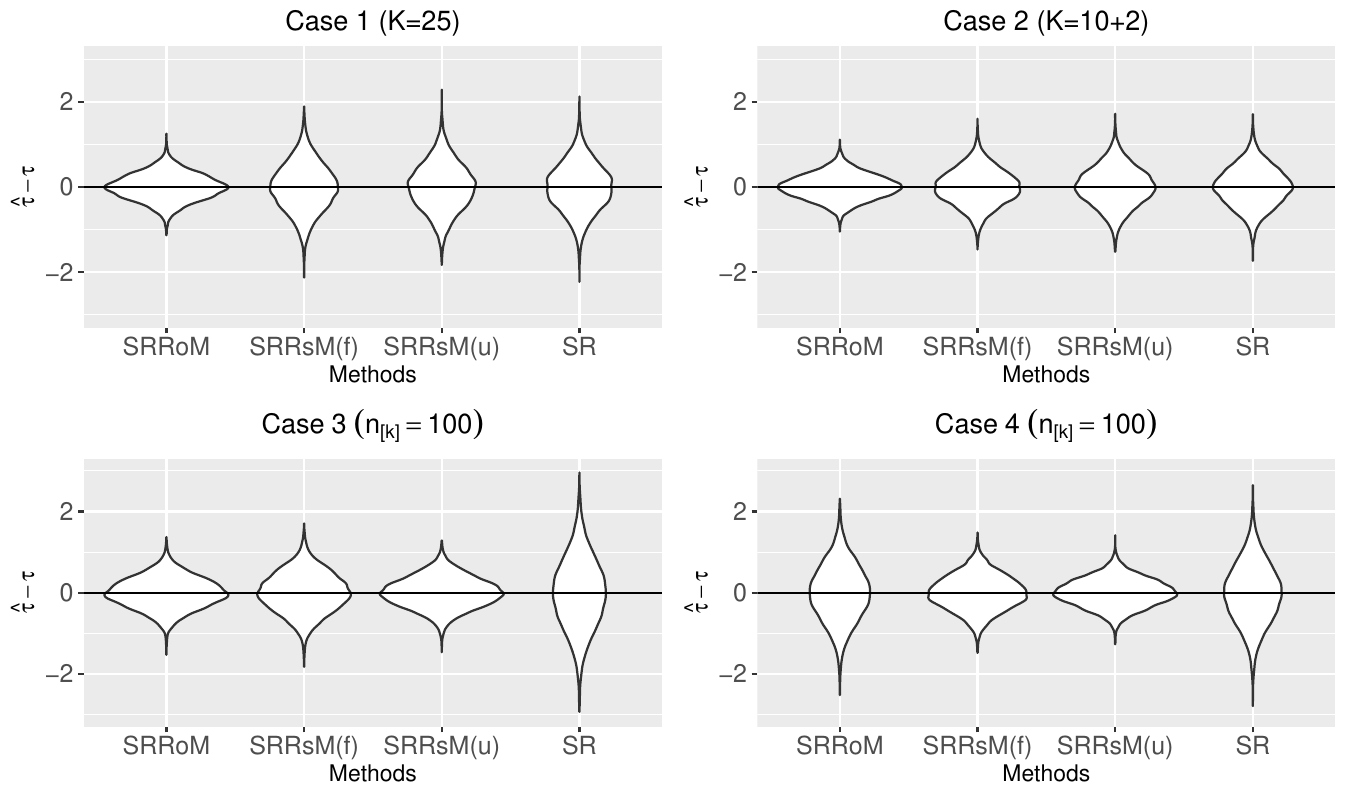}
\end{center}
\caption{Violin plot of the (centered) average treatment effect estimator, $\hat\tau-\tau$, under SRRoM, SRRsM(f) for a fair comparison, SRRsM(u) for an unfair comparison, and stratified randomization (SR). The  propensity scores are equal across strata $(\pk =0.5, \ k=1,\ldots,\S)$. \label{fig1}}
\end{figure}

Figure \ref{fig1} shows the results for equal propensity scores. Additional results are available in the Supplementary Material. Our findings are summarized as follows. First, the  treatment effect estimators under all assignment mechanisms have small finite-sample biases, which are more than ten times smaller than the standard deviations. Second, compared to stratified randomization, SRRoM always reduces the RMSEs and confidence interval lengths, regardless of the stratum numbers and sizes. The percentages of reduction are  $2.6\% - 56.0\%$ and $3.6\% - 40.7\%$, respectively. Third, when there exist small strata (Cases 1 and 2), fair SRRsM performs similarly to, or slightly better than, stratified randomization in terms of RMSE, and is less efficient than SRRoM. In this setting, fair SRRsM can still reduce the confidence interval lengths ($3.4\% - 11.6\%$) compared to stratified randomization because it uses less conservative variance estimators. Fourth, when unfair SRRsM can be implemented (Cases 3 and 4), it is generally better than fair SRRsM because it uses stricter thresholds. When there are two large homogeneous strata (Case 3), fair SRRsM is less efficient than SRRoM, and unfair SRRsM is comparable to SRRoM in terms of RMSE but gives slightly longer confidence intervals. In contrast, when there are two large heterogeneous strata (Case 4), fair SRRsM is better than SRRoM, with percentages of reduction being $27.6\% - 38.9\%$ in RMSEs and $13.5\% - 25.5\%$ in confidence interval lengths. Finally, all the interval estimators are conservative, with the empirical coverage probabilities being larger than the confidence level.

In general, we recommend SRRoM when there exist small strata and SRRsM when there are only a few large strata. We also conduct simulations for paired randomized experiments and finely stratified randomized experiments, and the results are given in  the Supplementary Material.

\section{Application}
In this section, we analyse the `Opportunity Knocks' experiment data \citep{Angrist2014} using two stratified rerandomization methods and compare them with stratified randomization. The Opportunity Knocks data are obtained from an experiment that aims at evaluating the influence of a financial incentive demonstration plan on college students' academic performance. The research subjects of this experiment included first- and second-year students of a large Canadian commuter university who applied for financial aid. Stratification was conducted according to the year, sex, and high school GPA quartile. Students were randomly assigned to the
treatment and control groups within each stratum, and those who fell in the treated group had peer advisors and received cash bonuses for attaining the given grades. Students with missing outcomes or covariates were excluded, resulting in 16 strata, a treatment group of size 382, and a control group of size 821. The propensity scores $\pk$ varied from 0.22 to 0.51.

The  outcome of interest is the average grade for the semester right after the experiment (2008 fall). From the original dataset, we cannot determine the true gains of stratified rerandomization. To evaluate the repeated sampling properties of SRRoM and SRRsM, we generate a synthetic dataset, with the missing values of the potential outcomes imputed by a linear model of regressing the observed outcomes on the  treatment indicator, average grade in 2008 spring, gender, and  treatment-by-covariate interactions. The resulting average treatment effect is $0.205$. To conduct stratified rerandomization, we select seven covariates: high school grade, average grade in 2008 spring, number of college graduates in the family, whether  the first/second question in the survey is correctly answered, whether the mother tongue is English, and credits earned in 2008 fall.  We center the covariates at their stratum-specific means. Stratified rerandomization is conducted under the same stratification and propensity scores as the original dataset, with acceptance probability $p_{a}=0.001$ for SRRoM and $p_{a_k}=(p_a)^{1/16}$ for SRRsM for a fair comparison, and  $p_{a_k}= 0.001$ for an unfair comparison, $k=1,\dots,K$. We repeat the stratified rerandomization for $10^4$ times and compute the bias, standard deviation, RMSE, mean confidence interval length, and empirical coverage probability of different methods.

\begin{table}
\caption{Results of stratified rerandomization applied to the Opportunity Knocks data \label{tab3}}
\begin{center}
\begin{tabular}{p{2cm}<{\centering}p{2cm}<{\centering}p{2cm}<{\centering}p{2cm}<{\centering}p{2cm}<{\centering}p{2cm}<{\centering}}
Method & Bias & SD & RMSE & CI length & CP (\%) \\[5pt] \midrule
	SRRoM &  0.0042 & 0.3935 & 0.3935 & 1.9739 & 98.69 \\ 
	SRRsM(f)& 0.0366 & 0.4821 & 0.4835 & 2.2240 & 97.96 \\ 
	SRRsM(u)& 0.0308 & 0.3901 & 0.3913 & 1.8099 & 97.77 \\ 
	SR& 0.0024 & 0.5283 & 0.5283 & 2.4095 & 97.58 \\  \hline
\end{tabular}
\begin{tablenotes}

Note: SRRoM, stratified rerandomization based on overall Mahalanobis distance; SRRsM, stratified rerandomization based on stratum-specific Mahalanobis distance; SRRsM(f), SRRsM for a fair comparison; SRRsM(u), SRRsM for an unfair comparison; SR, stratified randomization; SD, standard deviation; RMSE, root mean squared error; CI length, mean confidence interval length; CP, empirical coverage probability.

\end{tablenotes}	
\end{center}	
\end{table}

Table~\ref{tab3} and Figure~\ref{figrd} show the results, where the bias of each method is more than 10 times smaller than the standard deviation. Among the considered methods, 
SRRoM performs similarly to unfair SRRsM, and both of them are better than the other two methods. They reduce the RMSE of the stratified difference-in-mean estimator by approximately $26\%$ when compared to stratified randomization. Fair SRRsM is less efficient than SRRoM and it does not substantially improve efficiency compared to stratified randomization. Moreover, all confidence intervals are conservative, with the coverage probabilities being larger than the confidence level.

\begin{figure}
\begin{center}
\includegraphics[width=4in]{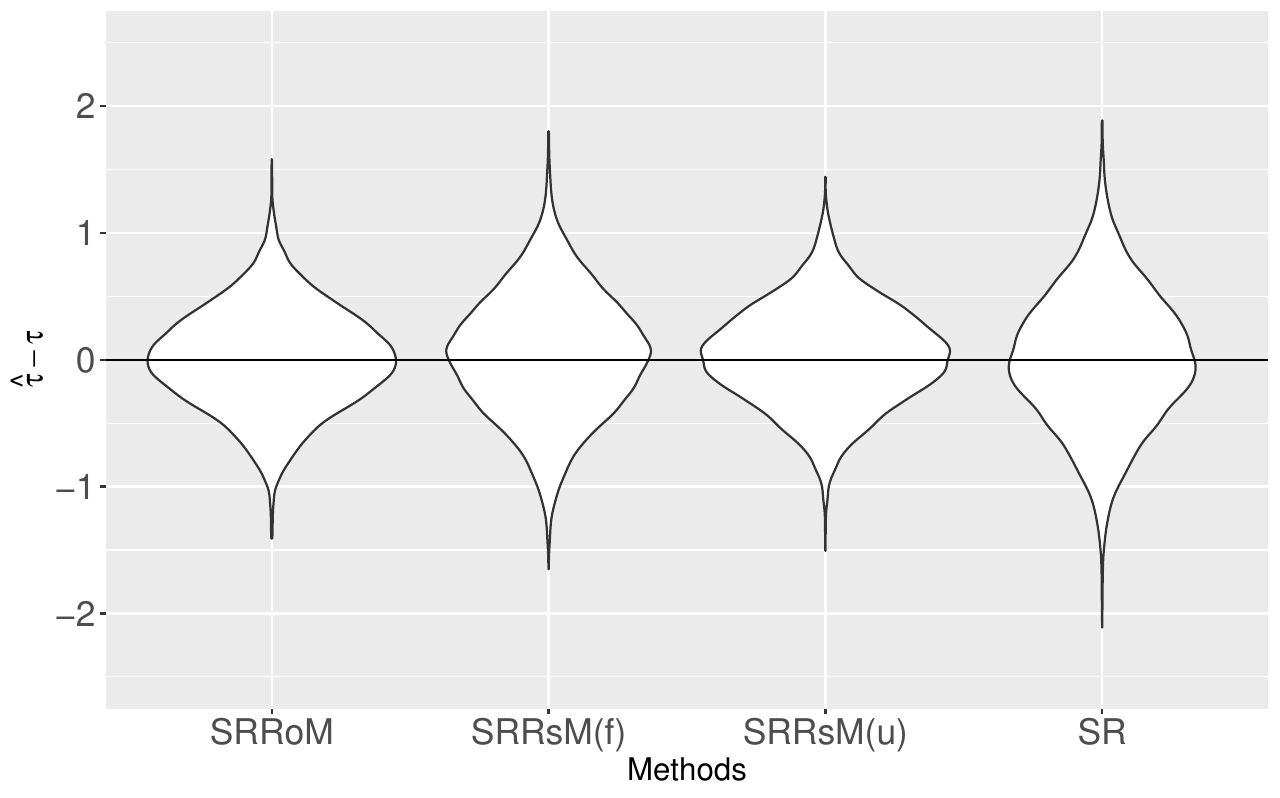}
\end{center}
\caption{Violin plot of the (centered) average treatment effect estimator applied to the Opportunity Knocks data, $\hat\tau-\tau$, under SRRoM, SRRsM(f) for fair comparison, SRRsM(u) for unfair comparison, and stratified randomization (SR). \label{figrd}}
\end{figure}

\section{Conclusion}
\label{sec:conc}

In this paper, we propose two rerandomization strategies based on the overall and stratum-specific Mahalanobis distances in stratified randomized experiments. We derive the asymptotic distributions of the stratified difference-in-means estimator under both strategies.  We demonstrate the advantages of the proposed  strategies  through theoretical investigations, a simulation study, and a real data analysis. In addition, we provide asymptotically conservative estimators for the variances and asymptotic distributions, which can be used to construct large-sample conservative confidence intervals for the average treatment effect. Our work provides statistical support for the recommendation of renowned scholars, such as R. A. Fisher and D. B. Rubin, to ``Block what you can and rerandomize what you cannot". 

In theory, when there are only a few large strata, the stratum-specific stratified rerandomization strategy is more efficient than the overall stratified rerandomization strategy. Thus, when there are a few large strata and many small strata, it might be more efficient to pool small strata into large ones, and then, use the stratum-specific stratified rerandomization. It would be interesting to study how to efficiently pool small strata together and investigate the theoretical properties of the stratum-specific stratified rerandomization after pooling. Moreover, this paper focuses on stratified rerandomization for a binary treatment and continuous potential outcomes.  It would be interesting to generalize our  results to stratified rerandomization with multiple treatments (including factorial experiments) or binary potential outcomes. In addition,  we assume that the number of the additional covariates is fixed; in practice, however, the number of the additional covariates can be large, even larger than the sample size. It is worthy of further investigation to develop stratified rerandomization method using high-dimensional covariates.


\bigskip
\begin{center}
{\large\bf SUPPLEMENTARY MATERIAL}
\end{center}

The Supplementary Material provides a detailed discussion on Schultzberg and Johansson's stratified rerandomization criterion, another conservative variance estimator, additional simulation results, and proofs.

\bibliographystyle{apalike}
\bibliography{causal}

\newpage

\begin{center}
{\LARGE{\bf Supplementary material}}
\end{center}

\appendix

\setcounter{equation}{0}
\renewcommand{\theequation}{S.\arabic{equation}}
\setcounter{theorem}{0}
\renewcommand{\thetheorem}{S.\arabic{theorem}}
\setcounter{proposition}{0}
\renewcommand{\theproposition}{S.\arabic{proposition}}
\setcounter{condition}{0}
\renewcommand{\thecondition}{S.\arabic{condition}}
\setcounter{corollary}{0}
\renewcommand{\thecorollary}{S.\arabic{corollary}}
\setcounter{remark}{0}
\renewcommand{\theremark}{S.\arabic{remark}}

\section{Stratified rerandomization based on difference-in-means}

In this section, we discuss the stratified rerandomization criterion proposed by \citet{Schultzberg2019}, which is based on  a widely used measurement of covariate imbalance $M_{\ttX}=\ttX^\T\cov(\ttX)^{-1}\ttX$, where
\begin{equation*}
	\ttX =\frac{1}{n_1}\sum_{i=1}^n\Zi\Xi-\frac{1}{n_0}\sum_{i=1}^n(1-\Zi)\Xi
\end{equation*}
is the difference-in-means estimator of the covariates.  $\ttX$ pools the treated units (and the control units) together and ignores the stratification used in the design stage. It is identical to the stratified difference-in-means estimator $\htX$ when the propensity scores are the same across strata, otherwise it is  different from  $\htX$ and is a biased estimator for the average treatment effect of the covariates $\tau_X = 0$.  In what follows,  we derive the asymptotic normality of the random vector $n^{1/2}(\hat\tau-\tau,\ttX^\T)^\T$ using Theorem \ref{thm1CLT} in the main text. Let $p_1=n_1/n$ and $p_0=n_0/n$ be the overall proportions of units in the treatment and control groups, respectively. 

\begin{proposition}
\label{props1}
Under stratified randomization, the covariance of $n^{1/2}(\hat\tau-\tau,\ttX^\T)^\T$ is 
	\begin{equation*}
		U=\left(
		\begin{array}{cc}
		U_{\tau\tau}&U_{\tau x}\\	
		U_{x\tau}&U_{xx}\\
		\end{array}\right)
		=\skS\pik \left(
		\begin{array}{cc}
		\frac{\SkY(1)}{\pk}+\frac{\SkY(0)}{1-\pk}-\Skt & \frac{(1-\pk)\SkXY^\T(1)}{p_1p_0}+\frac{\pk\SkXY^\T(0) }{p_1p_0}\\
		\frac{(1-\pk)\SkXY(1)}{p_1p_0}+\frac{\pk\SkXY(0) }{p_1p_0}&
		\frac{\pk(1-\pk)}{p_1^2p_0^2}\SkX\\
		\end{array}
		\right).
	\end{equation*}
\end{proposition}

\begin{condition}
\label{cond6}
	The following three matrices have finite limits:
	\begin{eqnarray*}
		\skS\frac{\pik}{\pk} \left(
		\begin{array}{cc}
		\SkY(1) & ({\pk}/{p_1})\SkXY^\T(1)\\
		({\pk}/{p_1})\SkXY(1) & ({\pk}/{p_1})^2\SkX\\
		\end{array}
		\right),
	\end{eqnarray*}
	\begin{eqnarray*}
		\skS\frac{\pik}{1-\pk} \left(
		\begin{array}{cc}
		\SkY(0) & \{({1-\pk})/p_0 \} \SkXY^\T(0)\\
		\{({1-\pk})/p_0\} \SkXY(0) & \{ ({1-\pk})^2/p_0^2\} \SkX\\
		\end{array}
		\right),
	\end{eqnarray*}
	\begin{equation*}
		\skS\frac{\pik(\pk-p_1)}{p_1p_0} \left(
		\begin{array}{cc}
		\{ {p_1p_0}/(\pk-p_1) \} \Skt & \SkXY^\T(1)-\SkXY^\T(0)\\
		\SkXY(1)-\SkXY(0)& \{ ({\pk-p_1})/(p_1p_0) \} \SkX\\
		\end{array}
		\right),
	\end{equation*}
	and the limit of $U$, denoted by $U^\infty$, is (strictly) positive definite.
\end{condition}

\begin{corollary}
	\label{cor2}
	Under stratified randomization, if Conditions \ref{cond::propensity}, \ref{cond4}, and \ref{cond6} hold and 
	$${n^{1/2}}/(p_1p_0)\skS\pik(\pk-p_1)\bXk\to\omega ,$$ 
	where $\omega\in\mathbb{R}^{\p\times1}$ is a constant vector, then $n^{1/2}(\hat{\tau}-\tau, \ttX^\T)^\T$ converges in distribution to $\mathcal{N}((0,\omega^\T)^\T,U^\infty)$ as $n$ tends to infinity.
\end{corollary}

As the covariance matrix of the difference-in-means of the covariates is
\begin{equation*}
	\cov(\ttX)=\frac{1}{np_1^2p_0^2}\skS\pik\pk(1-\pk)\SkX=\frac{1}{n}U_{xx},
\end{equation*}
the Mahalanobis distance based on $\ttX$ is
\begin{equation*}
\begin{split}
	M_{\ttX}=\ttX^\T\cov(\ttX)^{-1}\ttX
	=(n^{1/2}\ttX)^\T U_{xx}^{-1}(n^{1/2}\ttX).
\end{split}
\end{equation*}
The rerandomization criterion based on $M_\ttX$ accepts an assignment only when $M_\ttX<a$, where $a  $  is a predetermined threshold. Let us denote $\Md = \{(Z_1,\ldots,Z_n):\ M_{\ttX}<a\}$ as an event in which an assignment is accepted under the \underline{s}tratified \underline{r}e\underline{r}andomization based on the \underline{d}ifference-in-means \underline{M}ahalanobis distance $M_{\ttX}$, which is abbreviated as SRRdM.

\begin{proposition}
	\label{prop2pa}
	Under SRRdM, if the conditions in Corollary \ref{cor2} hold, then the asymptotic probabiliy of accepting a random assignment is $p'_a=\pr\{\chi^2_{\p}(\omega^\T U_{xx}^{-1}\omega)<a\}$, where $\chi^2_\p(\lambda)$ represents a noncentral chi-square distribution with $\p$ degrees of freedom and noncentrality parameter $\lambda$.
\end{proposition}

The asymptotic biasedness of $\ntau$ under SRRdM can be derived from Corollary \ref{cor2}. Let $D_\omega\sim\mathcal{N}(U_{xx}^{-1/2}\omega,I)$ be a $\p$-dimensional normal random vector with mean $U_{xx}^{-1/2}\omega$ and identity covariance matrix $I$.

\begin{theorem}
	\label{thm4}
	Under SRRdM, if the conditions in Corollary \ref{cor2} hold, then the asymptotic expectation of $\ntau$ is $U_{\tau x}U_{xx}^{-1/2}\big\{U_{xx}^{-1/2}\omega+E(D_\omega\mid D_\omega^\T D_\omega<a)\big\}$. Moreover, when the propensity scores are the same across strata, SRRdM is equivalent to SRRoM if we use the same threshold for these two criteria.
\end{theorem}


According to Theorem \ref{thm4}, the asymptotic expectation of $\ntau$ under SRRdM is usually not equal to zero when $\omega\neq0$, that is,  $\hat\tau$ is a biased estimator for $\tau$ under SRRdM, even asymptotically, thus we do not recommend SRRdM to be used in stratified randomized experiments.

\section{Another conservative estimator of $\Sigma_{\tau\tau}$}
This section provides another conservative variance estimator based on the results of \citet{Pashley2017}. This variance estimator is applicable when at least two strata of each (small) stratum size exist. 
Let $m_j$ ($j=1,\ldots,J$) be different sizes of small strata and $K_j$ be the number of strata of size $m_j$. We assume that $K_j\ge2$ for all $j$. Let $\tmj=K_j^{-1}\sum_{k\in A_{ss}:\nk=m_j}\tauk$ be the average treatment effect over strata of size $m_j$. Then
$$\hat\Sigma^*_{\tau\tau}=\sum_{k \notin A_{ss} } \pik\Big\{\frac{\skY(1)}{\pk}+\frac{\skY(0)}{1-\pk}\Big\} +  \left(  \frac{n_{ss} }{n} \right)^2 \frac{n}{n_{ss}^2 }\sum_{j=1}^J(m_jK_j)^2 \frac{1}{K_j(K_j-1)}\sum_{k\in A_{ss}:\nk=m_j}(\htk-\htmj)^2
$$ 
is a conservative estimator of $\Sigma_{\tau\tau}$, where $\htmj= (1 / K_j ) \sum_{k\in A_{ss}:\nk=m_j}\htk$ is an unbiased estimator of $\tmj$.
Let $\hat{R}^2_*=\hat{\Sigma}_{\tau x}\Sigma_{xx}^{-1}\hat \Sigma_{x\tau}/\hat \Sigma_{\tau\tau}^*$ and  $\nu_\xi(\hat R^2_*,p_a,p)$ be the $\xi$th quantile of $(1-\hat R^2_*)^{1/2}\epsilon_0+(\hat R^2_*)^{1/2}L_{\p,a}$. 
\begin{theorem}
	\label{thmvaro2}
	Under SRRoM, if Conditions \ref{cond1pk} and \ref{cond4}--\ref{cond::small-strata} hold, then $\hat \Sigma_{\tau\tau}^*\{1-(1-v_{\p,a})\hat{R}^2_*\}$ is an asymptotically conservative estimator for the asymptotic variance of $\ntau$ and
	\begin{equation*}
		\big[\hat\tau-(\hat \Sigma_{\tau\tau}^*/n)^{\frac{1}{2}}\nu_{1-\alpha/2}(\hat R^2_*,p_a,p),\ \hat\tau-(\hat \Sigma_{\tau\tau}^*/n)^{\frac{1}{2}} \nu_{\alpha/2}(\hat R^2_*,p_a,p)\big]
	\end{equation*}
	is an asymptotically conservative $\CI$ confidence interval of $\tau$.
\end{theorem}
\begin{remark}
 We have two conservative estimators of $\Sigma_{\tau\tau}$. It is natural to ask which one is better. In fact,  these two estimators are equivalent when the stratum sizes and propensity scores are the same across strata. Generally, $ \hat\Sigma^*_{\tau\tau}$ has a smaller bias when the treatment effects of similar-sized small strata are similar. However, the use of $ \hat\Sigma^*_{\tau\tau}$ is limited to the situations when there are at least two strata of each size. See \citet{Pashley2017} for more detailed discussions. 
\end{remark}

\section{Additional simulation results}
Figure \ref{fig2} shows the results of different methods for unequal propensity scores in Cases 1-4 in the main text. We present the bias, standard deviation (SD), root mean squared error (RMSE), mean confidence interval length (CI length), and empirical coverage probability (CP) in  Tables \ref{tab11}--\ref{tab22}.

We also conduct simulation studies for paired and finely stratified randomized experiments. Since these two kinds of experiments contain only small strata, we study the performance of only SRRoM and compare it with stratified randomization.  The simulation setup is similar to  Case 1 with the difference that we let $\nk=2$ and $\pk=0.5$ for paired randomized experiments, and let $\nk=4$ and $\pk=0.25$ for finely stratified experiments. Since the stratum sizes and propensity scores are the same across strata, two conservative variance estimators $\hat \Sigma_{\tau \tau}$ and $ \hat\Sigma^*_{\tau\tau}$ are equal to each other.

The results are shown in Figure~\ref{pair} and Table~\ref{tabpair}. The conclusions are similar to those in the main text. First, the treatment effect estimators under both assignment mechanisms have small finite-sample biases, which are more than ten times smaller than the standard deviations.  Second, for paired randomized experiments,  SRRoM  reduces the RMSE and confidence interval length by  $90.67\% - 93.78\%$ and $18.03\% - 36.15\%$, respectively, while for finely stratified randomized experiments, the percentages of reduction in RMSE and confidence interval length are $61.13\% - 133.61\%$ and $39.50\% - 90.33\%$, respectively. Third, all the interval estimators are conservative, with the empirical coverage probabilities being larger than the confidence level.


\begin{figure}[ht]
\begin{center}
\includegraphics[width=6in]{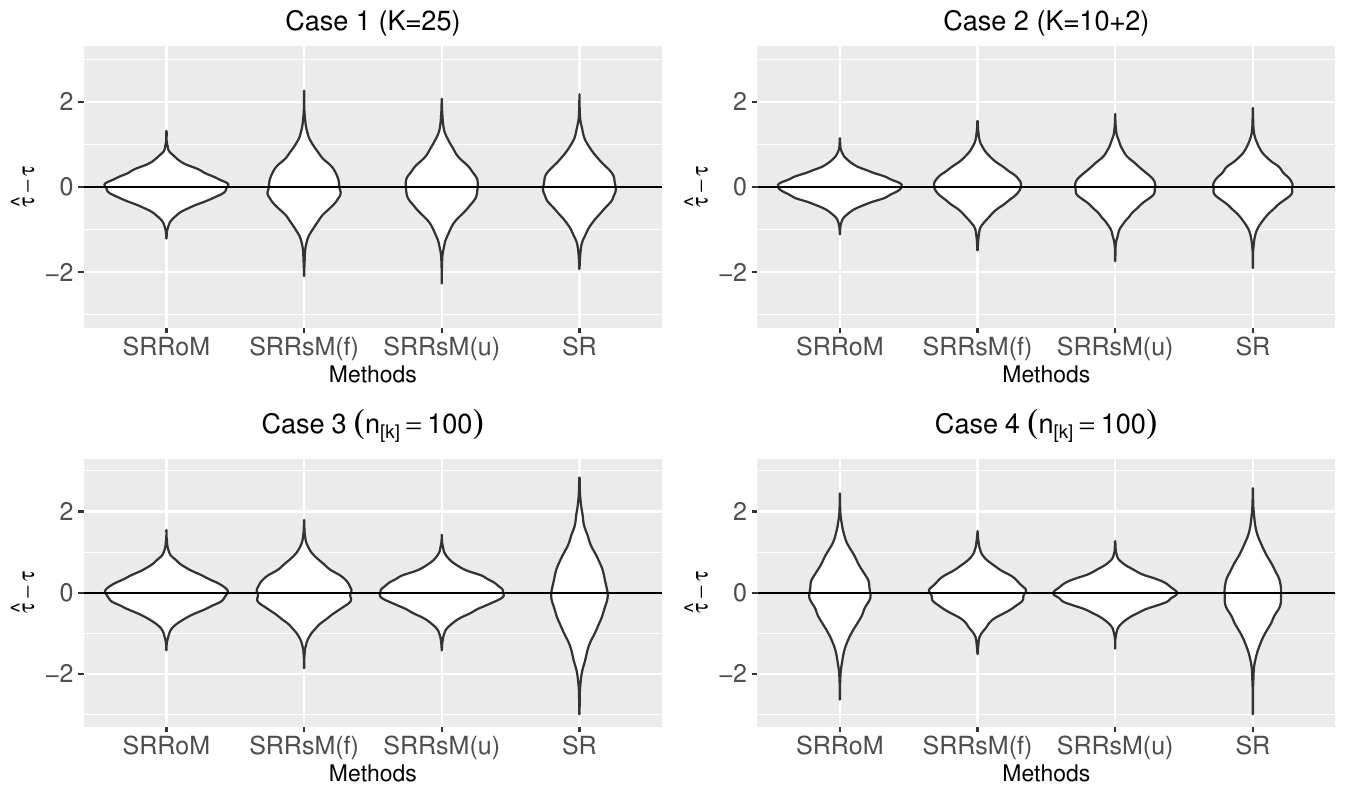}
\end{center}
\caption{Violin plot of the (centered) average treatment effect estimator, $\hat\tau-\tau$, under SRRoM, SRRsM(f) for a fair comparison, SRRsM(u) for an unfair comparison, and stratified randomization (SR). The  propensity scores are unequal across strata ($\pk =0.4, k\le \S/2,$ and $\pk =  0.6, \ k>\S/2)$. \label{fig2}}
\end{figure}

\newpage
\ \\

\begin{table}[H]
\caption{{Simulation results for Case 1} \label{tab11}}
\begin{center}
{
\begin{tabular}{ccccccccc}
  $K$ & Propensity score & Method & Bias & SD & RMSE & CI length & CP (\%) \\[5pt] \hline
  25 & equal & SRRoM & -0.0015 & 0.3128 & 0.3128 & 2.2538 & 99.95 \\ 
 & & SRRsM(f) & -0.0049 & 0.5543 & 0.5544 & 2.6563 & 98.20 \\ 
 & & SRRsM(u) & 0.0008 & 0.5523 & 0.5523 & 2.9207 & 99.13 \\ 
 & & SR &-0.0089 & 0.5569 & 0.5570 & 2.9189 & 98.99 \\ \hline
 & unequal & SRRoM &  0.0002 & 0.3435 & 0.3435 & 2.2857 & 99.89 \\ 
 & & SRRsM(f) &  0.0063 & 0.5729 & 0.5729 & 2.6818 & 97.98 \\ 
 & & SRRsM(u) & -0.0007 & 0.5763 & 0.5763 & 2.9763 & 98.93 \\ 
 & & SR & 0.0011 & 0.5728 & 0.5728 & 2.9806 & 98.97 \\ \hline
 50 & equal & SRRoM  & 0.0016 & 0.2216 & 0.2216 & 1.3816 & 99.79 \\ 
 & & SRRsM(f) & 0.0024 & 0.3561 & 0.3561 & 1.6568 & 98.13 \\ 
 & & SRRsM(u) & 0.0011 & 0.3518 & 0.3518 & 1.7589 & 98.69 \\ 
 & & SR & -0.0029 & 0.3561 & 0.3561 & 1.7585 & 98.57 \\ \hline
 & unequal & SRRoM & 0.0015 & 0.2394 & 0.2394 & 1.4146 & 99.60 \\ 
 & & SRRsM(f) & -0.0018 & 0.3659 & 0.3659 & 1.6859 & 97.75 \\ 
 & & SRRsM(u) & 0.0273 & 0.3518 & 0.3529 & 1.8126 & 99.03 \\ 
 & & SR & -0.0058 & 0.3671 & 0.3672 & 1.8020 & 98.53 \\ \hline
 100 & equal & SRRoM & -0.0023 & 0.1547 & 0.1547 & 0.8522 & 99.33 \\ 
 & & SRRsM(f) & 0.0020 & 0.2625 & 0.2625 & 1.1528 & 97.30 \\ 
 & & SRRsM(u) & -0.0015 & 0.2647 & 0.2647 & 1.1935 & 97.82 \\ 
 & & SR &  0.0004 & 0.2622 & 0.2622 & 1.1933 & 97.66 \\ \hline
 & unequal & SRRoM &  0.0017 & 0.1589 & 0.1590 & 0.8662 & 99.55 \\ 
 & & SRRsM(f) & -0.0006 & 0.2727 & 0.2727 & 1.1715 & 96.61 \\ 
 & & SRRsM(u) & -0.0236 & 0.2665 & 0.2675 & 1.2216 & 97.52 \\ 
 & & SR & 0.0022 & 0.2678 & 0.2678 & 1.2190 & 97.90 \\ \hline
\end{tabular}
}
\begin{tablenotes}
Note: SRRoM, stratified rerandomization based on overall Mahalanobis distance; SRRsM, stratified rerandomization based on stratum-specific Mahalanobis distance; SRRsM(f), SRRsM for fair comparison; SRRsM(u), SRRsM for unfair comparison; SR, stratified randomization; SD, standard deviation; RMSE, root mean squared error; CI length, mean confidence interval length; CP, empirical coverage probability.
\end{tablenotes}
\end{center}	
\end{table}

\begin{table}[H]
\caption{{ Simulation results for Case 2} \label{tab12}}
\begin{center}
{
\begin{tabular}{ccccccccc}
  $K$ & Propensity score & Method & Bias & SD & RMSE & CI length & CP (\%) \\[5pt] \hline
  10+2 & equal & SRRoM & 0.0038 & 0.2917 & 0.2918 & 1.7266 & 99.74 \\ 
 & & SRRsM(f) & -0.0028 & 0.4172 & 0.4172 & 2.0047 & 98.24 \\ 
 & & SRRsM(u) & 0.0031 & 0.4575 & 0.4575 & 2.2276 & 98.40 \\ 
 & & SR & 0.0018 & 0.4583 & 0.4583 & 2.2292 & 98.40 \\ \hline
 & unequal & SRRoM &  -0.0000 & 0.3055 & 0.3055 & 1.7446 & 99.60 \\ 
 & & SRRsM(f) &  0.0266 & 0.4337 & 0.4345 & 2.0207 & 97.91 \\ 
 & & SRRsM(u) & 0.0017 & 0.4742 & 0.4742 & 2.2859 & 98.22 \\ 
 & & SR & 0.0114 & 0.4763 & 0.4765 & 2.2870 & 98.13 \\ \hline
 20+2 & equal & SRRoM  & -0.0010 & 0.2342 & 0.2342 & 1.3554 & 99.62 \\ 
 & & SRRsM(f) & -0.0054 & 0.3524 & 0.3524 & 1.6175 & 97.93 \\ 
 & & SRRsM(u) &  0.0027 & 0.3591 & 0.3591 & 1.7493 & 98.54 \\ 
 & & SR &  -0.0004 & 0.3640 & 0.3640 & 1.7489 & 98.36 \\ \hline
 & unequal & SRRoM & 0.0010 & 0.2454 & 0.2454 & 1.3787 & 99.44 \\ 
 & & SRRsM(f) &  -0.0052 & 0.3678 & 0.3678 & 1.6468 & 97.44 \\ 
 & & SRRsM(u) & -0.0016 & 0.3759 & 0.3759 & 1.8001 & 98.30 \\ 
 & & SR &  0.0006 & 0.3816 & 0.3816 & 1.8013 & 98.13 \\ \hline
 50+2 & equal & SRRoM & -0.0029 & 0.2189 & 0.2189 & 1.1714 & 99.29 \\ 
 & & SRRsM(f) &  -0.0057 & 0.4497 & 0.4498 & 1.8618 & 96.21 \\ 
 & & SRRsM(u) & 0.0006 & 0.4553 & 0.4553 & 1.9713 & 96.74 \\ 
 & & SR &  0.0089 & 0.4614 & 0.4615 & 1.9715 & 96.71 \\ \hline
 & unequal & SRRoM & 0.0011 & 0.2317 & 0.2317 & 1.1946 & 99.00 \\ 
 & & SRRsM(f) &  0.0024 & 0.4610 & 0.4610 & 1.8905 & 96.04 \\ 
 & & SRRsM(u) &  -0.0041 & 0.4716 & 0.4717 & 2.0133 & 96.59 \\ 
 & & SR & 0.0010 & 0.4691 & 0.4691 & 2.0134 & 96.69 \\ \hline
\end{tabular}
}
\begin{tablenotes}
Note: SRRoM, stratified rerandomization based on overall Mahalanobis distance; SRRsM, stratified rerandomization based on stratum-specific Mahalanobis distance; SRRsM(f), SRRsM for fair comparison; SRRsM(u), SRRsM for unfair comparison; SR, stratified randomization; SD, standard deviation; RMSE, root mean squared error; CI length, mean confidence interval length; CP, empirical coverage probability.
\end{tablenotes}
\end{center}	
\end{table}

\begin{table}[H]
\caption{{ Simulation results for Case 3} \label{tab21}}
\begin{center}
{
	\begin{tabular}{ccccccccc}
$ \nk $&Propensity score& Method & Bias & SD & RMSE & CI length & CP (\%) \\[5pt] \hline
	100 &equal&SRRoM & -0.0066 & 0.3803 & 0.3803 & 2.3613 & 99.77 \\ 
	&&SRRsM(f)&  0.0045 & 0.4896 & 0.4896 & 2.7093 & 99.46 \\ 
	&&SRRsM(u)& 0.0022 & 0.3758 & 0.3758 & 2.4404 & 99.89 \\ 
	&&SR&  -0.0115 & 0.8652 & 0.8652 & 3.8709 & 97.19 \\ \hline
	&unequal&SRRoM &-0.0019 & 0.4005 & 0.4005 & 2.3758 & 99.70 \\ 
	&&SRRsM(f)&-0.0000 & 0.5014 & 0.5014 & 2.7249 & 99.49 \\ 
	&&SRRsM(u)&  0.0065 & 0.3938 & 0.3938 & 2.4245 & 99.81 \\ 
	&&SR& -0.0048 & 0.8876 & 0.8876 & 3.9798 & 97.42 \\ \hline
	200 &equal&SRRoM & -0.0000 & 0.2343 & 0.2343 & 1.3582 & 99.61 \\ 
	&&SRRsM(f)& -0.0003 & 0.2562 & 0.2562 & 1.4362 & 99.57 \\ 
	&&SRRsM(u)& -0.0003 & 0.2331 & 0.2331 & 1.3702 & 99.56 \\ 
	&&SR&  0.0056 & 0.3636 & 0.3636 & 1.7552 & 98.65 \\ \hline
	&unequal&SRRoM & 0.0037 & 0.2415 & 0.2416 & 1.3754 & 99.45 \\ 
	&&SRRsM(f)& 0.0039 & 0.2661 & 0.2661 & 1.4518 & 99.43 \\ 
	&&SRRsM(u)& 0.0026 & 0.2411 & 0.2411 & 1.3791 & 99.60 \\ 
	&&SR& -0.0019 & 0.3764 & 0.3764 & 1.8026 & 98.33 \\ \hline
	500 &equal&SRRoM &  -0.0010 & 0.1554 & 0.1554 & 0.8570 & 99.47 \\ 
	&&SRRsM(f)& -0.0011 & 0.1755 & 0.1755 & 0.9213 & 99.14 \\ 
	&&SRRsM(u)& 0.0001 & 0.1541 & 0.1541 & 0.8625 & 99.59 \\ 
	&&SR& -0.0004 & 0.2625 & 0.2625 & 1.1947 & 97.58 \\ \hline
	&unequal&SRRoM & 0.0007 & 0.1588 & 0.1588 & 0.8704 & 99.32 \\ 
	&&SRRsM(f)&  0.0018 & 0.1818 & 0.1818 & 0.9352 & 98.98 \\ 
	&&SRRsM(u)& -0.0022 & 0.1611 & 0.1611 & 0.8740 & 99.29 \\ 
	&&SR& 0.0042 & 0.2701 & 0.2701 & 1.2193 & 97.62 \\ \hline
	\end{tabular}
	}
\begin{tablenotes}
Note: SRRoM, stratified rerandomization based on overall Mahalanobis distance; SRRsM, stratified rerandomization based on stratum-specific Mahalanobis distance; SRRsM(f), SRRsM for fair comparison; SRRsM(u), SRRsM for unfair comparison; SR, stratified randomization; SD, standard deviation; RMSE, root mean squared error; CI length, mean confidence interval length; CP, empirical coverage probability.
\end{tablenotes}
\end{center}	
\end{table}

\begin{table}[H]
\caption{{ Simulation results for Case 4} \label{tab22}}
\begin{center}
{
	\begin{tabular}{ccccccccc}
$ \nk $&Propensity score& Method & Bias & SD & RMSE & CI length & CP (\%) \\[5pt] \hline
	 100 &equal&SRRoM &  0.0008 & 0.6964 & 0.6964 & 3.3621 & 98.37 \\ 
	&&SRRsM(f)& -0.0022 & 0.4254 & 0.4254 & 2.6829 & 99.88 \\ 
	&&SRRsM(u)&-0.0000 & 0.3440 & 0.3440 & 2.5060 & 99.94 \\ 
	&&SR& -0.0108 & 0.7361 & 0.7362 & 3.4941 & 98.09 \\ \hline
	&unequal&SRRoM & 0.0005 & 0.7148 & 0.7148 & 3.3857 & 98.12 \\ 
	&&SRRsM(f)& 0.0001 & 0.4394 & 0.4394 & 2.6818 & 99.83 \\ 
	&&SRRsM(u)& 0.0018 & 0.3484 & 0.3484 & 2.4915 & 99.96 \\ 
 	&&SR& -0.0049 & 0.7412 & 0.7413 & 3.5423 & 98.30 \\ \hline
	 200 &equal&SRRoM &  0.0002 & 0.3877 & 0.3877 & 1.9380 & 98.83 \\ 
	&&SRRsM(f)& 0.0015 & 0.2805 & 0.2805 & 1.6755 & 99.73 \\ 
 	&&SRRsM(u)& -0.0044 & 0.2379 & 0.2379 & 1.5811 & 99.90 \\ 
  	&&SR& -0.0071 & 0.4408 & 0.4408 & 2.1216 & 98.38 \\ \hline
	&unequal&SRRoM & -0.0025 & 0.4088 & 0.4088 & 2.0000 & 98.49 \\ 
	&&SRRsM(f)& -0.0042 & 0.2923 & 0.2924 & 1.6992 & 99.73 \\ 
	&&SRRsM(u)&0.0005 & 0.2523 & 0.2523 & 1.5909 & 99.87 \\ 
	&&SR& -0.0008 & 0.4681 & 0.4681 & 2.2025 & 98.27 \\ \hline
	 500 &equal&SRRoM & 0.0033 & 0.2791 & 0.2791 & 1.2411 & 97.42 \\ 
	&&SRRsM(f)& 0.0003 & 0.1874 & 0.1874 & 0.9524 & 98.97 \\ 
	&&SRRsM(u)& 0.0003 & 0.1604 & 0.1604 & 0.8785 & 99.40 \\ 
	&&SR& 0.0028 & 0.2874 & 0.2874 & 1.2870 & 97.26 \\ \hline
	&unequal&SRRoM & 0.0004 & 0.2856 & 0.2856 & 1.2581 & 97.26 \\ 
	&&SRRsM(f)& 0.0004 & 0.1913 & 0.1913 & 0.9670 & 99.03 \\ 
	&&SRRsM(u)& 0.0034 & 0.1646 & 0.1647 & 0.8918 & 99.21 \\ 
	&&SR& -0.0004 & 0.2933 & 0.2933 & 1.3074 & 97.37 \\ \hline
	\end{tabular}
}
\begin{tablenotes}
Note: SRRoM, stratified rerandomization based on overall Mahalanobis distance; SRRsM, stratified rerandomization based on stratum-specific Mahalanobis distance; SRRsM(f), SRRsM for fair comparison; SRRsM(u), SRRsM for unfair comparison; SR, stratified randomization; SD, standard deviation; RMSE, root mean squared error; CI length, mean confidence interval length; CP, empirical coverage probability.
\end{tablenotes}
\end{center}	
\end{table}

\begin{figure}[ht]
\begin{center}
\includegraphics[width=5.65in]{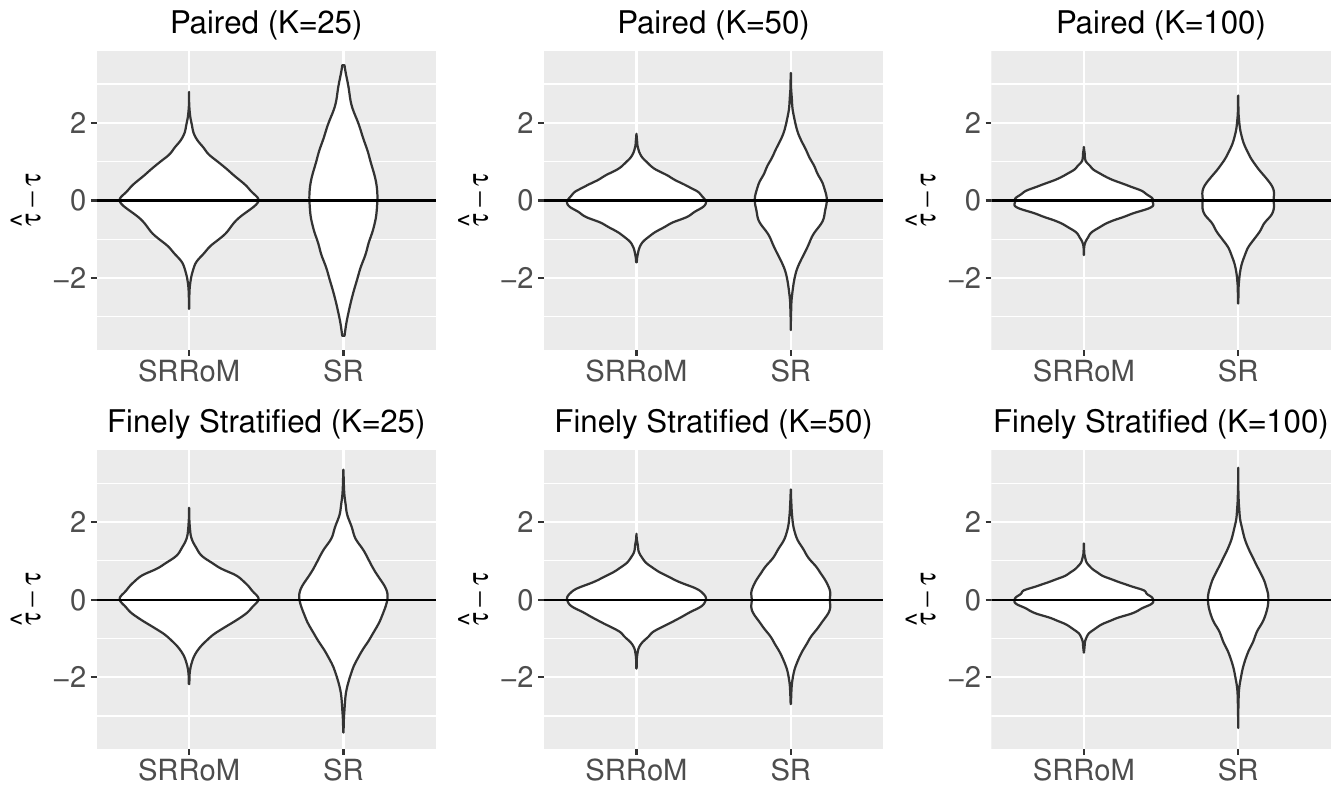}
\end{center}
\caption{Violin plot of the (centered) average treatment effect estimator, $\hat\tau-\tau$, under SRRoM and stratified randomization (SR). The first line shows the results of paired randomized experiments, and the second line shows the results of finely stratified randomized experiments. \label{pair}}
\end{figure}

\begin{table}[H]
\caption{{ Simulation results for paired and finely stratified randomized experiments} \label{tabpair}}
\begin{center}
{
	\begin{tabular}{ccccccccc}
	&$K$& Method & Bias & SD & RMSE & CI length & CP (\%) \\[5pt] \hline
	paired&25&SR & 0.0111 & 1.4414 & 1.4414 & 6.7423 & 97.38 \\ 
	 experiments&&SRRoM & 0.0003 & 0.7560 & 0.7560 & 5.1565 & 99.82 \\ \cline{2-8}
	&50&SR& 0.0149 & 0.9489 & 0.9490 & 4.8359 & 98.59 \\ 
	&&SRRoM& -0.0039 & 0.4897 & 0.4897 & 4.0970 & 100.00 \\ \cline{2-8}
	&100&SR&-0.0035 & 0.7497 & 0.7497 & 3.4744 & 97.64 \\ 
	&&SRRoM&-0.0001 & 0.3890 & 0.3890 & 2.5519 & 99.86  \\ \hline
	finely &25&SR& -0.0172 & 1.0036 & 1.0037 & 4.6598 & 97.14 \\
	stratified &&SRRoM& 0.0033 & 0.6229 & 0.6229 & 3.1964 & 96.27 \\ \cline{2-8}
	experiments&50&SR & 0.0048 & 0.8149 & 0.8149 & 3.7297 & 97.49 \\
	&&SRRoM & -0.0104 & 0.4807 & 0.4809 & 2.6736 & 99.15 \\ \cline{2-8}
	&100&SR& -0.0024 & 0.8887 & 0.8887 & 3.6279 & 95.80 \\ 
	&&SRRoM& -0.0037 & 0.3804 & 0.3804 & 1.9061 & 98.06 \\ \hline
	\end{tabular}
}
\begin{tablenotes}
Note: SRRoM, stratified rerandomization based on overall Mahalanobis distance; SR, stratified randomization; SD, standard deviation; RMSE, root mean squared error; CI length, mean confidence interval length; CP, empirical coverage probability.
\end{tablenotes}
\end{center}	
\end{table}

\section{Proof of results}

\subsection{Some lemmas}

Our proofs rely on some lemmas obtained from \citet{Li2018}, which are presented below without proof.

Our proofs rely on some lemmas from \citet{Li2018}, which are presented below without proof.

\begin{lemma}
\label{lemA1}
Let $L_{\p,a}\sim D_1\mid D^\T D\le a$, where $D=(D_1,\ldots,D_\p)^\T\sim\mathcal{N}(0,I)$. For any $\p$ dimensional unit vector $h$, we have $L_{\p,a}\sim h^\T D\mid D^\T D\le a$.
\end{lemma}

\begin{lemma}
\label{lemA4}
Let $\varepsilon_0 \sim \mathcal{N}(0,1)$, $L_{\p,a}\sim D_1|D'D\leq a$, where $D=(D_1,\ldots,D_\p)^{'}\sim \mathcal{N}(0,I)$, and $\varepsilon_0$ and $L_{\p,a}$ are mutually independent. Then for any $a>0$ and $c \geq 0$,
	\begin{equation}
		\pr(\sqrt{1-\rho^2}\cdot \varepsilon_0+\rho L_{p,a}\geq c)
	\end{equation}
is a nonincreasing function of $\rho$ for $\rho\in[0,1]$.
\end{lemma}

\begin{lemma}
\label{lemA5}
For any $0\leq \pa\leq \tpa\leq 1$ and any $c \geq 0$,
\begin{equation*}
	\pr(|L_{\p,F_\p^{-1}(\pa)}|\leq c)\geq \pr(|L_{\p,F_\p^{-1}(\tpa)}|\leq c)
\end{equation*}
\end{lemma}

\begin{lemma}
\label{lemA6}
For any $\tilde{\p}\ge \p\ge 1$ and any $c \geq 0$,
\begin{equation*}
	\pr(|L_{\p,F_\p^{-1}(\pa)}|\le c)\ge \pr(|L_{\tilde{\p},F_{\tilde\p}^{-1}(\pa)}|\le c)
\end{equation*}
\end{lemma}

\begin{lemma}
\label{lemA7}
Let $\zeta_0,\ \zeta_1,$ and $\zeta_2$ be three jointly independent random variables satisfying: \\
\indent (1) $\zeta_0$ is continuous, symmetric around 0 and unimodal;\\
\indent (2) $\zeta_1$ and $\zeta_2$ are symmetric around 0;\\
\indent (3) $\pr(\zeta_1>c)\le\pr(\zeta_2>c)$ for any $c>0$.\\
Then $\pr(\zeta_0+\zeta_1>c)\le\pr(\zeta_0+\zeta_2>c)$ for any $c>0$.
\end{lemma}

\subsection{Joint asymptotic normality under SR}

\subsubsection{Proof of Proposition \ref{prop0}}
\begin{proof}
According to the definitions in Section $\textnormal{2}$ in the main text,  the stratum-specific average treatment effect for the vector outcomes $R_i(z)$ and its difference-in-means estimator can be expressed as
$$\tau_{[k] \Ymul} = \frac{1}{\nk} \sumk \big \{  \Ymul_{i}(1) - \Ymul_i(0) \big \},$$ 
$$
\hat \tau_{[k] \Ymul} = \frac{1}{\nkt} \sumk  Z_i  \Ymul_{i}(1) -  \frac{1}{\nkc} \sumk ( 1 - Z_i ) \Ymul_i(0) .
$$
In stratified randomization, since we conduct complete randomization in each stratum independently, $\pik = \nk / n$, and $\pk = \nkt/\nk$, then
\[\begin{split}
   & \cov \big\{ n^{1/2}(\htmul-\tmul) \big\}  =  \cov \Big\{ n^{1/2}\skS\pik\htkmul \Big\} 
    = n\skS\pik^2\cov \big(\htkmul \big) \\
    =&\skS\pik\nk\Big\{ \frac{\SkYmul(1)}{\nkt}+\frac{\SkYmul(0)}{\nkc}- \frac{\Sktmul}{\nk} \Big\}  \\
    = & \skS\pik\Big\{\frac{\SkYmul(1)}{\pk}+\frac{\SkYmul(0)}{1-\pk}-\Sktmul\Big\},
\end{split}\]
where the formula of $\cov(\htkmul)$ is obtained from Theorem 3 of \citet{Li2017}.

\end{proof}


\subsubsection{Proof of Proposition \ref{prop::covariance}}
\begin{proof}
    Let $\Ymul_i(z)=(\Yi(z),\Xi^\T)^\T,\ z=0,1$, then we have
    \begin{equation*}
    \begin{split}
        \SkYmul(z)=&\frac{1}{\nk-1}\sumk \big\{\Ymul_i(z)-\bYmulk(z) \big\} \big\{\Ymul_i(z)-\bYmulk(z) \big\}^\T \\
        =&\sumk\pik\left(
        \begin{array}{cc}
            \SkY(z) & \SkXY^\T(z) \\
            \SkXY(z) & \SkX
        \end{array} \right),\quad z=0,1,
    \end{split}
    \end{equation*}
    and \[ \Sktmul=\frac{1}{\nk-1}\sumk\{\tau_{i,\Ymul}-\tau_{[k] \Ymul}  \}\{\tau_{i,\Ymul}-\tau_{[k] \Ymul}  \}^\T=\sumk\pik\left(
        \begin{array}{cc}
            \Skt(z) & 0 \\
            0 & 0
        \end{array} \right),
    \]
    where $\tau_{i,\Ymul} = \Ymul_i(1) - \Ymul_i(0) = ( \tau_i , 0^\T )^\T$. According to Proposition \ref{prop0}, we have
    \begin{equation}  \cov\{ n^{1/2}(\hat\tau-\tau, \htX^\T)^\T \} 
	=\skS\pik \left(
	\begin{array}{cc}
	\frac{\SkY(1)}{\pk}+\frac{\SkY(0)}{1-\pk}-\Skt & \frac{\SkXY^\T(1)}{\pk}+\frac{\SkXY^\T(0)}{1-\pk} \\
	\frac{\SkXY(1)}{\pk}+\frac{\SkXY(0)}{1-\pk} &	\frac{\SkX}{\pk(1-\pk )} \\
	\end{array}
	\right). \nonumber
    \end{equation}

\end{proof}

\subsubsection{Proof of Theorem \ref{thm1CLT}}
\begin{proof}
	It is enough to show that any linear combination of the components of the random vector $n^{1/2}(\htmul-\tmul)$ converges in distribution to a normal distribution. More precisely,  let $\mu=(\mu_1\ldots,\mu_d)^\T \in\mathbb{R}^d$ be a fixed $d$-dimensional vector and $\mu \neq 0$. It is enough to show that  $ n^{1/2} \mu^\T (\htmul-\tmul) $ converges in distribution to a normal distribution with mean zero and variance $\mu^\T \Sigmul^{\infty} \mu$. For this purpose, we define  scalar potential outcomes as 
	\[R_i^{\new}(z)=\sum_{j=1}^d\mu_jR_{i,j}(z),\ z=0,1,\ i=1,\ldots,n,\] where $R_{i,j}(z)$ is the $j$th component of $R_i(z)$. 
The corresponding average treatment effect and its stratified difference-in-means estimator are denoted as $\tau^{\new}$ and $\hat\tau^{\new}$, respectively. Then
	\begin{equation*}
	\begin{split}
		\tau^{\new}= & \skS\pik \cdot \frac{1}{\nk}  \sumk  \big\{ R_i^{\new}(1)-R_i^{\new}(0) \big\} \\
		=&\skS\pik \cdot \frac{1}{\nk} \sumk \bigg\{ \sum_{j=1}^d\mu_jR_{i,j}(1)-\sum_{j=1}^d\mu_jR_{i,j}(0) \bigg\} \\
		=& \sum_{j=1}^d\mu_j  \skS\pik  \cdot \frac{1}{\nk}  \sumk \big\{ R_{i,j}(1)-R_{i,j}(0) \big\} \\
		=& \mu^\T \tmul,
	\end{split}
	\end{equation*}
	and
	\begin{equation*}
	\begin{split}
	\htnew =&\skS\pik\sumk\Big\{\frac{\Zi R_i^{\new}(1)}{\nkt}-\frac{(1-\Zi)R_i^{\new}(0)}{\nkc}\Big\}\\
	=&\skS\pik\sumk\Big\{\frac{\Zi\sum_{j=1}^d\mu_jR_{i,j}(1)}{\nkt}-\frac{(1-\Zi)\sum_{j=1}^d\mu_jR_{i,j}(0)}{\nkc}\Big\}\\
	=&\sum_{j=1}^d\mu_j\skS\pik\sumk\Big\{\frac{\Zi R_{i,j}(1)}{\nkt}-\frac{(1-\Zi)R_{i,j}(0)}{\nkc}\Big\} \\ 
	= &  \mu^\T \htmul.
	\end{split}
	\end{equation*}
	Therefore, $ n^{1/2} \mu^\T (\htmul-\tmul)  = n^{1/2} ( \hat {\tau}^{\new}- {\tau}^{\new}). $ We only need to check that the conditions for obtaining the asymptotic normality of the (scalar) stratified difference-in-means estimator $\hat\tau^{\new}$ hold.  
	
	\begin{lemma}
	\label{thm::liu2019}
Under Condition \ref{cond::propensity}, if $R_i^{\new}(z)$ satisfies the following two conditions:

(a) For $z=0,1$, $\max_{k=1,\ldots,\S}\max_{i\in[k]} \big\{ R_i^{\new}(z) -  \bar{R}_{[k]}^{\new}(z) \big\}^2/n\to 0$;

(b) The covariance matrices $  \skS\pik \SkYnew(1) / \pk$, $ \skS\pik \SkYnew(0) / (1-\pk)$ and $\skS\pik \Sktnew$ have finite limits, and the limit of the variance $\sigma^2_n = \var \{  n^{1/2} \mu^\T (\htmul-\tmul) \}$ is (strictly) positive, where
\[\sigma_n^2=\skS\pik \bigg\{ \frac{\SkYnew(1)}{\pk}+\frac{\SkYnew(0)}{1-\pk}-\Sktnew \bigg\},\]
then $n^{1/2}(\htnew-\tnew)/\sigma_n$ converges in distribution to $\mathcal{N}(0,1)$ as $n$ tends to infinity.
	
	\end{lemma}
	
	\begin{remark}
	Lemma~\ref{thm::liu2019} is a direct result of Theorem 2 of \citet{Liu2019}, so we omit the proof of this lemma.
	\end{remark}

The remaining of the proof is to check that Conditions (a) and (b) hold. According to the definition,  the stratum-specific mean of $R_i^{\new}(z)$ is $\bYmulk^{\new}(z)=\sumk\Ymul^{\new}(z)/\nk$ $(z=0,1)$, and the vector-form stratum-specific mean of the potential outcomes $R_i(z)$  is
    $$  \bYmulk(z) = \frac{1}{\nk} \sumk \Ymul_i(z)=(\bar{R}_{[k],1}(z),\ldots,\bar{R}_{[k],d}(z))^\T ,\quad z=0,1.$$
By Condition \ref{cond::max}, we have for $z=0,1$,
	\begin{equation*}
	\begin{split}
		&\frac{1}{n}\max_{1\le k\le \S}\max_{\ik}\{R_i^{\new}(z)-\bYmulk^{\new}(z)\}^2\\
		=&\frac{1}{n}\max_{1\le k\le \S}\max_{\ik}\Big[\sum_{j=1}^d\mu_j\{R_{i,j}(z)-\bar{R}_{[k],j}(z)\}\Big]^2\\
		\le&\frac{1}{n}\max_{1\le k\le \S}\max_{\ik}\Big(\sum_{j=1}^d\mu_j^2\Big)\sum_{j=1}^d \big\{R_{i,j}(z)-\bar{R}_{[k],j}(z)\big\}^2\\
		\le & d \cdot \Big(\sum_{j=1}^d\mu_j^2 \Big) \cdot \frac{1}{n}\max_{1\le k\le \S}\max_{\ik}  \|\Ymul_i(z)-\bYmulk(z)\|_\infty^2 \to 0. 
	\end{split}
	\end{equation*}
Moreover, the stratum-specific variance of $R_i^{\new}(z)$ in stratum $k$ satisfies
	\begin{equation*}
	\begin{split}
		\SkYnew(z)
		=&\frac{1}{\nk-1}\sumk\{R_i^{\new}(z)-\bYnewk(z)\}^2\\
		=&\frac{1}{\nk-1}\sumk\Big[\sum_{j=1}^d\mu_j\{R_{i,j}(z)-\bar{R}_{[k],j}(z)\}\Big]^2\\
		=&\frac{1}{\nk-1}\sumk\Big[\sum_{j=1}^d\mu_j^2\{R_{i,j}(z)-\bar{R}_{[k],j}(z)\}^2 \\
	\ \ &+\sum_{j\neq l}\mu_j\mu_l\{R_{i,j}(z)-\bar{R}_{[k],j}(z)\}\{R_{i,l}(z)-\bar{R}_{[k],l}(z)\}\Big]\\
		=&\sum_{j=1}^d\mu_j^2\{\SkYmul(z)\}_{j,j}+\sum_{j\neq l}\mu_j\mu_l\{\SkYmul(z)\}_{j,l}
		={\mu}^\T \SkYmul(z){\mu},
	\end{split}
	\end{equation*}
where $(B)_{i,j}$ denotes the $(i,j)$th element of matrix $B$. Thus, by Condition \ref{cond::covariance},
	\begin{equation*}
	\begin{split}
		\skS\pik\frac{\SkYnew(1)}{\pk}=\skS\pik\frac{{\mu}^\T \SkYmul(1){\mu}}{\pk}
		={\mu}^\T \bigg\{ \skS\pik\frac{\SkYmul(1)}{\pk} \bigg\} {\mu}
	\end{split}
	\end{equation*}
	has a finite limit as $n$ tends to infinity. 
	
	Similarly, $\skS\pik \SkYnew(0) / (1-\pk)$ and $\skS\pik \Sktnew$ have finite limits, and 
	\begin{equation}
	\begin{split}
	& \skS\pik\Big\{\frac{\SkYnew(1)}{\pk}+\frac{\SkYnew(0)}{1-\pk}-\Sktnew\Big\} \\
	 = & {\mu}^\T \bigg[\skS\pik\Big\{\frac{\SkYmul(1)}{\pk}+\frac{\SkYmul(0)}{1-\pk}-\Skt\Big\}\bigg]{\mu} \\
	 = &  {\mu}^\T \Sigmul {\mu}
	\label{A}
	\end{split}
	\end{equation} has a  limit $\mu^\T \Sigmul^{\infty} \mu > 0$.

	Thus, by Lemma~\ref{thm::liu2019} and Slutsky's theorem, $n^{1/2} \mu^\T (\htmul-\tmul)$ converges in distribution to $\mathcal{N}(0, \mu^\T \Sigmul^\infty \mu)$.

\end{proof}

\subsubsection{Proof of Theorem \ref{thm:cov-est}}

Before proving Theorem~\ref{thm:cov-est}, we establish the following lemma. Let $X_{ij}$ be the $j$th element of $X_i$, $j=1,\dots,p$. For any pair $(\Ai,\Bi)$ being equal to $(X_{ij}, X_{il}),\ (\Yi(1),X_{ij}),\ (\Yi(0), X_{ij}),\ (\Yi(1),\Yi(1)),$ or $(\Yi(0),\Yi(0))$, $i=1,\dots,n$, $j,l=1,\dots, p$, let $\skAB(z)$ denote the stratum-specific sample covariance between $\Ai $'s and $\Bi $'s in stratum $k$ under treatment arm $z$ for $n_{[k]z} \geq 2$, and let $\SkAB$ denote the stratum-specific population covariance between $\Ai $'s and $\Bi $'s in stratum $k$. Recall that $A_{ss} = \{k :  \nkt =1 \ \textnormal{or}\ \nkc = 1\}$ is the index set of small strata with only one treated or one control unit.
Recall that when $ \nkt = 1$,
	$$s_{[k]XY} (1) = \frac{ \nk }  { (\nk - 1) } \sumk Z_i  ( \Xi -  \bXk) Y_i^\obs =  \frac{ \nk }  { (\nk - 1)   } \cdot \frac{1}{\nkt} \sumk  Z_i   ( \Xi -  \bXk) Y_i(1).$$
Let $\rkz=z\pk+(1-z)(1-\pk)$, $z=0,1,\  k=1,\ldots,\S$. Denote
$$
\hat \sigma^2_{AB}(z) = \sum_{k \notin A_{ss} } \pik \frac{\skAB(z)}{\rkz}\ \ \text{and}\ \ \sigma^2_{AB}(z) = \sum_{ k \notin A_{ss}  }\pik \frac{\SkAB}{\rkz}.
$$
Here, $S_{[k]Y(z)X} = S_{[k]YX}(z) = S_{[k]XY}^\T(z)$ and $S_{[k]Y(z)Y(z)} = S_{[k]Y}^2(z)$ for $z=0,1$.
\begin{lemma}
	Under SRRoM, if Conditions \ref{cond::propensity}, \ref{cond4}, and \ref{cond5} hold, then $ \hat \sigma^2_{AB}(z) -  \sigma^2_{AB}(z) $ converges to zero in probability. Furthermore, if Condition~\ref{cond::small-strata} holds, then $\sum_{k \in A_{ss} } \pik s_{[k]XY} (z) / \rkz    -  \sum_{k \in A_{ss} } \pik S_{[k]XY} (z) / \rkz $  also converges to zero in probability.
	\label{lemsAB}
\end{lemma}
\begin{remark}
The first part of Lemma~\ref{lemsAB} is similar to Lemma A1 of \citet{Liu2019} with the differences being that there is a denominator  $\rkz$ in the weighted summation, and that the randomness originates from stratified rerandomization (SRRoM) instead of stratified randomization. 
\end{remark}


\begin{proof}[Proof of Lemma~\ref{lemsAB}]
	{\bf Step 1}. We prove the first part of Lemma~\ref{lemsAB}. According to Proposition \ref{prop1pa} (we will prove Proposition \ref{prop1pa} later), there exists a constant $c_a$ such that $\pr(\Mtau)\ge c_a>0$ when $n$ is sufficiently large. Then by the property of conditional  expectation,
	\begin{equation*}
	\begin{split}
		& E\big[\{ \hat \sigma^2_{AB}(z) -  \sigma^2_{AB}(z)\}^2\big] \\
		=&\ \pr(\Mtau)E\big[\{ \hat \sigma^2_{AB}(z) -  \sigma^2_{AB}(z)\}^2\mid\Mtau\big] +\pr(\Mtau^c)E\big[\{ \hat \sigma^2_{AB}(z) -  \sigma^2_{AB}(z)\}^2\mid\Mtau^c\big]\\
		\ge&\ \pr(\Mtau)\cdot E\big[\{ \hat \sigma^2_{AB}(z) -  \sigma^2_{AB}(z)\}^2\mid \Mtau\big],
	\end{split}
	\end{equation*}
	where $\Mtau^c$ is the complementary set of $\Mtau$.
	Therefore, 
	\begin{equation}
	\begin{split}
	E\big[\{ \hat \sigma^2_{AB}(z) -  \sigma^2_{AB}(z)\}^2\mid\Mtau\big]\le&\  \pr(\Mtau)^{-1}E\big[\{ \hat \sigma^2_{AB}(z) -  \sigma^2_{AB}(z)\}^2\big] \\
	=&\ \pr(\Mtau)^{-1}\var\{ \hat \sigma^2_{AB}(z)\}.
	\label{lemsAB1}
	\end{split}
	\end{equation}
	Let 
	$$ \bAobsk(z)=\frac{1}{\nkz}\sum_{i \in [k]:\ \Zi=z}\Ai \quad  \textnormal{and} \quad \bBobsk(z)=\frac{1}{\nkz}\sum_{i \in [k]: \ \Zi=z}\Bi, \quad z=0,1, $$
	 be the averages of the observed $\Ai $'s and $\Bi $' in stratum $k$ under treatment arm $z$. When $n_{[k]z} \geq 2$,  the stratum-specific sample variance  can be decomposed as
	\begin{equation*}
		\skAB(z)=\frac{\nkz}{\nkz-1}\Big[\frac{1}{\nkz}\sum_{i \in [k]: \ \Zi=z}(\Ai -\bAk)(\Bi -\bBk)-\big\{\bAobsk(z)-\bAk \big\} \big\{\bBobsk(z)-\bBk \big\}\Big],
	\end{equation*}
	then
	\begin{equation}
	\begin{split}
		&\var\{ \hat \sigma^2_{AB}(z) \}=\sum_{ k \notin A_{ss}} \frac{\pik^2}{\rkz^2}\var\{\skAB(z)\}\\
		=&\sum_{ k \notin A_{ss}} \frac{\pik^2}{\rkz^2}\frac{\nkz^2}{(\nkz-1)^2}\var\bigg[\frac{1}{\nkz}\sum_{i \in [k]:\ \Zi=z}(\Ai -\bAk)(\Bi -\bBk) -\{\bAobsk(z)-\bAk\}\{\bBobsk(z)-\bBk\}\bigg] \\
		\le& \skS\frac{\pik^2}{\rkz^2}\frac{\nkz^2}{(\nkz-1)^2} \cdot 2\bigg( \var\Big\{\frac{1}{\nkz}\sum_{i \in [k]:\ \Zi=z}(\Ai -\bAk)(\Bi -\bBk)\Big\} \\
		&\quad \quad \quad \quad \quad \quad \quad \quad \quad \quad \quad +\var\Big[ \big\{\bAobsk(z)-\bAk \big\} \big\{\bBobsk(z)-\bBk \big\}\Big] \bigg).
	\end{split}
	\label{varsAB}		
	\end{equation}
	The first term in (\ref{varsAB}) is upper bounded as follows:
	\begin{equation*}
	\begin{split}
		& 2 \skS\frac{\pik^2}{\rkz^2}\frac{\nkz^2}{(\nkz-1)^2}  \var\bigg[\frac{1}{\nkz}\sum_{i \in [k]: \ \Zi=z}\{\Ai -\bAk\}\{\Bi -\bBk\}\bigg]\\
		\le& 2 \skS\frac{\pik^2}{\rkz^2}\frac{\nkz^2}{(\nkz-1)^2} \Big(\frac{1}{\nkz}-\frac{1}{\nk} \Big)\frac{1}{\nk-1}\sumk(\Ai -\bAk)^2(\Bi -\bBk)^2 \\
		\le& \frac{1}{n}\max_{1\le k\le \S}\max_{\ik}(\Ai -\bAk)^2 \cdot 2 \skS\frac{\pik\nk}{\rkz^2}\frac{\nkz^2}{(\nkz-1)^2}\Big(\frac{1}{\nkz}-\frac{1}{\nk}\Big)S_{[k]B}^2 \\
		\le & 8 \cdot  \frac{1}{n}\max_{1\le k\le \S}\max_{\ik}(\Ai -\bAk)^2 \cdot  \skS\frac{\pik}{\rkz}S_{[k]B}^2 \cdot \frac{1}{\rkz} \Big( \frac{1}{\rkz} - 1  \Big),
	\end{split}
	\end{equation*}
where the last inequality is because of $ \nkz^2 / (\nkz-1)^2 \leq 4$.

The second term in \eqref{varsAB} is upper bounded as follows:
	\begin{equation*}
	\begin{split}
		& 2 \skS\frac{\pik^2}{\rkz^2}\frac{\nkz^2}{(\nkz-1)^2} \var\Big[ \big\{\bAobsk(z)-\bAk \big\} \big\{\bBobsk(z)-\bBk \big\}\Big] \\
		\le&2  \skS\frac{\pi_k^2}{\rkz^2}\frac{\nkz^2}{(\nkz-1)^2}\max_{\ik}(\Ai -\bAk)^2 E\big\{\bBobsk(z)-\bBk \big\}^2 \\
		\le & 2 \cdot \frac{1}{n}\max_{1\le k\le \S}\max_{\ik}(\Ai -\bAk)^2\skS\frac{\pik\nk}{\rkz^2}\frac{\nkz^2}{(\nkz-1)^2}\Big(\frac{1}{\nkz}-\frac{1}{\nk}\Big) S_{[k]B}^2, \\
		\le & 8 \cdot  \frac{1}{n}\max_{1\le k\le \S}\max_{\ik}(\Ai -\bAk)^2 \cdot  \skS\frac{\pik}{\rkz}S_{[k]B}^2 \cdot \frac{1}{\rkz} \Big( \frac{1}{\rkz} - 1  \Big),
	\end{split}
	\end{equation*}
where the last inequality is also because of  $ \nkz^2 / (\nkz-1)^2 \leq 4$.

By Conditions \ref{cond4} and \ref{cond5}, as $n\to\infty$,  $\max_{1\le k\le \S}\max_{\ik}(\Ai -\bAk)^2/n\to0,$ and $ \skS (\pik/ \rkz ) S_{[k]B}^2$  has a finite limit when $B= X_{ij}$ or $Y_i(z)$. By Condition \ref{cond::propensity},  $ (1 / \rkz )  \big( 1 / \rkz - 1  \big)$ is upper bounded by a constant. Thus, 
$$
\frac{1}{n}\max_{1\le k\le \S}\max_{\ik}(\Ai -\bAk)^2 \cdot  \skS\frac{\pik}{\rkz}S_{[k]B}^2 \cdot \frac{1}{\rkz} \Big( \frac{1}{\rkz} - 1  \Big) \rightarrow 0.
$$
Therefore, $\var\{ \hat \sigma^2_{AB}(z) \}\to0$ as $n\to\infty$. Then by Chebyshev's inequality and \eqref{lemsAB1}, we have, for any $\epsilon>0$,
	\begin{equation*}
	\begin{split}
		\pr(| \hat \sigma^2_{AB}(z) -  \sigma^2_{AB}(z)|>\epsilon\mid\Mtau)\le\frac{1}{\epsilon^2}E\big[\{ \hat \sigma^2_{AB}(z) -  \sigma^2_{AB}(z)\}^2\mid\Mtau\big]\to0.
	\end{split}
	\end{equation*}
	Hence $ \hat \sigma^2_{AB}(z) -  \sigma^2_{AB}(z)$ converges to zero in probability.

	{\bf Step 2}. We prove the second part of Lemma~\ref{lemsAB}. We only prove the result for $z=1$ since the proof for $z=0$ is similar. To ease the notation, we assume that $X_i$ is a scalar in the remaining proof of this lemma. For a $p$-dimensional $X_i$, we can apply the following arguments  to each of its element $X_{ij}$, $j=1,\dots,p$. Since when $ \nkt = 1$,
	$$s_{[k]XY} (1) = \frac{ \nk }  { (\nk - 1)   } \cdot \frac{1}{\nkt} \sumk  Z_i   ( \Xi -  \bXk) Y_i(1)  ,$$
	then
	\begin{eqnarray}
	E\{ s_{[k]XY} (1) \} & = &    \frac{ \nk }  { (\nk - 1)  } \sumk E( Z_i )   ( \Xi -  \bXk)  Y_i(1)   \nonumber \\
	& = & \frac{1}{\nk - 1} \sumk  ( \Xi -  \bXk) Y_i(1) =  S_{[k]XY} (1), \nonumber
	\end{eqnarray}
	and 
	\begin{eqnarray}
	\label{eqn::var-nkt=1}
	\var \{ s_{[k]XY} (1) \}  &=& \left( \frac{\nk}{\nk - 1} \right)^2  \cdot \var  \Big\{ \frac{1}{\nkt} \sumik Z_i   ( \Xi -  \bXk) Y_i(1)  \Big\} \nonumber \\
	&\leq&  \left( \frac{\nk}{\nk - 1} \right)^2  \cdot \left(  \frac{1}{\nkt} - \frac{1}{\nk}  \right)  \cdot \frac{1}{\nk - 1} \sumik  ( \Xi -  \bXk)^2 Y_i^2(1)  \nonumber \\
	& \leq & 4  \cdot \max_{\ik}(\Xi -\bXk)^2 \cdot \left(  \frac{1}{\nkt} - \frac{1}{\nk}  \right)  \cdot \frac{1}{\nk - 1} \sumik Y_i^2(1).
	\end{eqnarray}
	Then it suffices for the convergence of $\sum_{k \in A_{ss} } \pik s_{[k]XY} (1) / \pk    -  \sum_{k \in A_{ss} } \pik S_{[k]XY} (1) / \pk$ to show that
	\begin{equation}
	\label{eqn::nkt=1}
	\var \bigg \{ \sum_{k: \ \nkt = 1} \pik \frac{s_{[k]XY} (1)}{\pk}  \bigg \}  \rightarrow 0. \nonumber
	\end{equation}
	By \eqref{eqn::var-nkt=1},  we have
	\begin{eqnarray}
	&& \var \bigg \{ \sum_{k: \ \nkt = 1} \pik \frac{s_{[k]XY} (1)}{\pk}  \bigg \}  \nonumber \\
	& \leq & 4  \cdot  \max_{1\le k\le \S} \max_{\ik}(\Xi -\bXk)^2    \sum_{k: \ \nkt = 1} \frac{\pik^2}{\pk^2}  \left(  \frac{1}{\nkt} - \frac{1}{\nk}  \right)  \frac{1}{\nk - 1} \sumik Y_i^2(1) \nonumber \\
	& \leq & 4 \cdot \frac{1}{n}  \max_{1\le k\le \S} \max_{\ik}(\Xi -\bXk)^2 \cdot \frac{1}{n} \skS \frac{1}{ \pk^2 } \left( \frac{1}{\pk} - 1  \right)  \frac{\nk}{\nk - 1} \sumik Y_i^2(1)  \nonumber  \\
	& \rightarrow & 0,  \nonumber
	\end{eqnarray}
	where the convergence is because of Conditions \ref{cond::propensity}, \ref{cond4}, and \ref{cond::small-strata}.
\end{proof}


Now, we can prove Theorem \ref{thm:cov-est}. Let $o_p(1)$ denote a sequence of random variables converging to zero in probability. We write $b_n = O(1)$ if $| b_n | \leq C$ for a constant $C$. Let $\htrki$  be the $i$th element of $\htrk$. Similarly, we define $\htrkj$, $\htrsi $, $\htrsj $, $\trki$, $\trkj$, $\trsi $, and $\trsj$.  Let $ A_{ij} $ be the $(i,j)$th element of a matrix $A$. 

\begin{proof}[Proof of Theorem \ref{thm:cov-est}]


By definition, we have
\begin{eqnarray}
 \hat  \Sigma_{R,ij}&=& \sum_{k \notin A_{ss}}\pik\Big\{\frac{\skrij(1)}{\pk}+\frac{\skrij(0)}{1-\pk}\Big\}+ \nonumber \\
&&\left(  \frac{n_{ss} }{n} \right)^2  \sum_{k \in A_{ss} }  \frac{ n \nk^2 } { ( n_{ss} - 2 \nk  ) \Big( n_{ss} + \sum_{h \in A_{ss}} \frac{ n_{[h]}^2  }{ n_{ss} - 2  n_{[h]} }  \Big)  }  (\htrki -  \htrsi )(\htrkj -  \htrsj ). \nonumber
\end{eqnarray}
It suffices for Theorem~\ref{thm:cov-est} to show that, element-wise, 
$$
E( \hat  \Sigma_{R} ) = \tilde{\Sigma}_{R} \quad \textnormal{and} \quad  \hat  \Sigma_{R} - \tilde{\Sigma}_{R} = o_p(1).
$$ 
For any $1\le i,j\le d$, similar to the proof of Theorem 3.4.1 in \citet{Pashley2017}, we have
\begin{eqnarray}
    && E \left\{   \left(  \frac{n_{ss} }{n} \right)^2  \sum_{k \in A_{ss} }  \frac{ n \nk^2 } { ( n_{ss} - 2 \nk  ) \Big( n_{ss} + \sum_{h \in A_{ss}} \frac{ n_{[h]}^2  }{ n_{ss} - 2  n_{[h]} }  \Big)  }  ( \htrki -  \htrsi ) ( \htrkj -  \htrsj )  \right\} \nonumber \\
    & = &  \sum_{k \in A_{ss} }  \pik \Big\{ \frac{\Skrij(1)}{\pk} + \frac{\Skrij(0)}{1 - \pk} - \Sktrij \Big\} +  \nonumber \\
    &&  \frac{n_{ss}^2}{n} \sum_{k \in A_{ss} }  \frac{  \nk^2 } { ( n_{ss} - 2 \nk  ) \Big( n_{ss} + \sum_{h \in A_{ss}} \frac{ n_{[h]}^2  }{ n_{ss} - 2  n_{[h]} }  \Big)  }  (\trki- \trsi)(\trkj- \trsj). \nonumber
\end{eqnarray}
Since 
$$
E \left[   \sum_{k \notin A_{ss} } \pik\Big\{\frac{\skrij(1)}{\pk}+\frac{\skrij(0)}{1-\pk}\Big\} \right] =  \sum_{k \notin A_{ss} } \pik\Big\{\frac{\Skrij(1)}{\pk}+\frac{\Skrij(0)}{1-\pk}\Big\}, 
$$
then
$$
E( \hat  \Sigma_{R} ) = \tilde{\Sigma}_{R}.
$$

Next, we show that, element-wise, $ \hat  \Sigma_{R} - \tilde{\Sigma}_{R} = o_p(1)$.  Replacing $Y_i(z)$ by $R_{i,j}(z)$, $z=0,1$, $j=1,\dots,d$ in the proof of Lemma~\ref{lemsAB},  we have
$$
\var \left\{ \sum_{k \notin A_{ss} } \pik \frac{\skrij(1)}{\pk}   \right\}  \rightarrow 0, \quad   \var \left\{ \sum_{k \notin A_{ss} } \pik \frac{\skrij(0)}{1 - \pk}   \right\}  \rightarrow 0.
$$
Therefore, 
\begin{equation}
\label{eqn::b1-term}
\var \left[   \sum_{k \notin A_{ss} } \pik\Big\{\frac{\skrij(1)}{\pk}+\frac{\skrij(0)}{1-\pk}\Big\}  \right] \rightarrow 0.
\end{equation}
We decompose $ \hat \Sigma_{R,ij}  $ as follows:
\begin{eqnarray}
\hat \Sigma_{R,ij} & \triangleq & B_1 + B_2, \nonumber
\end{eqnarray}
where 
$$
B_1 =   \sum_{k \notin A_{ss} } \pik\Big\{\frac{\skrij(1)}{\pk}+\frac{\skrij(0)}{1-\pk}\Big\} 
$$
and
$$
B_2 = \hat \Sigma_{R,ij} - B_1  =  \left(  \frac{n_{ss} }{n} \right)^2  \sum_{k \in A_{ss} }  \frac{ n \nk^2 } { ( n_{ss} - 2 \nk  ) \Big( n_{ss} + \sum_{h \in A_{ss}} \frac{ n_{[h]}^2  }{ n_{ss} - 2  n_{[h]} }  \Big)  } ( \htrki -  \htrsi ) ( \htrkj -  \htrsj ).
$$


By \eqref{eqn::b1-term} and Chebyshev's inequality, we have
$$
B_1 - E(B_1) = o_p(1).
$$
Now it suffices for  $ \hat  \Sigma_{R,ij} - \tilde{\Sigma}_{R,ij} = o_p(1)$ to show that 
$$B_2 - E(B_2) = o_p(1).$$ 

Let
$$
\theta_k =  \frac{  \nk^2 } { ( n_{ss} - 2 \nk  ) \Big( n_{ss} + \sum_{h \in A_{ss}} \frac{ n_{[h]}^2  }{ n_{ss} - 2  n_{[h]} }  \Big)  }.
$$
According to Condition~\ref{cond::propensity}, the sizes of small strata, $\nk$'s, are uniformly upper bounded, i.e.,   $\nk \leq C$ for all $k \in A_{ss}$, 
where $C$ is a constant. Thus, the number of small strata has the same order as $n_{ss}$ and $\theta_k$ has the same order as $1/n_{ss}^2$ (we say that two sequences of real numbers $\{a_n\}_{n\ge1}$ and $\{b_n\}_{n\ge1}$ have the same order, 
if there exist constants $c_1,\ c_2>0$ such that $c_1a_n\le b_n\le c_2a_n$ for $n\ge1$). Recall that $\trsi = \sum_{k \in A_{ss}} (\nk / n_{ss} ) \trki  = E( \htrsi ) $. The term $B_2$ can be further divided into six terms:
\begin{eqnarray}
B_2 & = &   \left(  \frac{n_{ss} }{n} \right)^2  \sum_{k \in A_{ss} }  \frac{ n \nk^2 } { ( n_{ss} - 2 \nk  ) \Big( n_{ss} + \sum_{h \in A_{ss}} \frac{ n_{[h]}^2  }{ n_{ss} - 2  n_{[h]} }  \Big)  }  ( \htrki -  \htrsi )( \htrkj -  \htrsj ) \nonumber \\
& = &     \frac{n^2_{ss} }{n}   \sum_{k \in A_{ss} } \theta_k  ( \htrki - \trki + \trki- \trsi + \trsi -  \htrsi )\times  \nonumber \\
&&( \htrkj - \trkj + \trkj - \trsj + \trsj - \htrsj )\nonumber \\
& = &   \frac{n^2_{ss} }{n}   \sum_{k \in A_{ss} } \theta_k  ( \htrki - \trki  )( \htrkj - \trkj  ) +  \nonumber \\
& &   \frac{n^2_{ss} }{n}   \sum_{k \in A_{ss} } \theta_k  (\trki - \trsi )(\trkj - \trsj ) +  \nonumber \\
&&   \frac{n^2_{ss} }{n}   \sum_{k \in A_{ss} } \theta_k  ( \trsi -  \htrsi )( \trsj -  \htrsj ) +  \nonumber \\
&&  \frac{n^2_{ss} }{n}   \sum_{k \in A_{ss} } \theta_k \left[  ( \htrki - \trki ) ( \trkj - \trsj )  + ( \htrkj - \trkj ) ( \trki - \trsi ) \right] + \nonumber \\
&&  \frac{n^2_{ss} }{n}   \sum_{k \in A_{ss} } \theta_k  \left[ ( \htrki - \trki ) ( \trsj -  \htrsj ) + ( \htrkj - \trkj ) ( \trsi - \htrsi )\right]+\nonumber \\
&&   \frac{n^2_{ss} }{n}   \sum_{k \in A_{ss} } \theta_k  \left[ ( \trki- \trsi )  ( \trsj -  \htrsj ) +( \trkj- \trsj )  ( \trsi -  \htrsi )\right]\nonumber \\
&\triangleq& B_{21} + B_{22} + B_{23} + B_{24} + B_{25} + B_{26}. \nonumber
\end{eqnarray}
\begin{itemize}
\item For the first term $B_{21}$, we have
\begin{eqnarray}
\var( B_{21} ) & = & \frac{n_{ss}^4}{n^2}   \sum_{k \in A_{ss} } \theta_k^2  \var\{    ( \htrki - \trki  )( \htrkj - \trkj  )\} \nonumber \\
& \leq &  \frac{n_{ss}^4}{n^2}   \sum_{k \in A_{ss} } \theta_k^2  E\{ ( \htrki - \trki  )^2( \htrkj - \trkj  ) ^2 \}. \nonumber
\end{eqnarray}
Let $D_{nz} = \max_{k=1,\dots,K} \max_{\ik} \Vert \Ri(z) - \oRk(z)\Vert_\infty $ for $z=0,1$. Note that 
\begin{eqnarray}
\label{eqn::bound-hattau-tau}
| \htrkj - \trkj |  & = & \Big| \frac{1}{\nkt} \sumik Z_i \{  \Rij(1) - \oRkj(1) \} -  \frac{1}{\nkc} \sumik (1 - Z_i) \{\Rij(0) - \oRkj(0) \} \Big|  \nonumber \\
& \leq & \max_{i \in [k]} | \Rij(1) - \oRkj(1) | +  \max_{i \in [k]} |\Rij(0) - \oRkj(0)| \nonumber \\
& \leq & D_{n1}+D_{n0}. \nonumber
\end{eqnarray}
Then,
\begin{eqnarray}
\var( B_{21} ) & \leq &   \frac{n_{ss}^4}{n^2}   \sum_{k \in A_{ss} } \theta_k^2 (D_{n1} + D_{n0})^2 E( \htrki - \trki  )^2  \nonumber \\
& = & O \left[  \frac{(D_{n1} + D_{n0})^2 }{n^2}  \sum_{k \in A_{ss} }  \Big\{ \frac{\Skrii(1)}{\nkt} +  \frac{\Skrii(0)}{\nkc} - \frac{\Sktrii}{\nk}   \Big\} \right]  \nonumber \\
& = & O \left\{  \frac{ (D_{n1} + D_{n0})^2}{n}  \right\}  \nonumber \\
& \rightarrow & 0,  \nonumber 
\end{eqnarray}
where the convergence is because of Condition~\ref{cond::max} . Thus, $B_{21} - E( B_{21} ) = o_p(1)$.
\item For the second term $B_{22}$, we have $\var(B_{22}) = 0$ and $B_{22} - E( B_{22} ) = 0 = o_p(1)$.
\item For the third term $B_{23}$, we have
\begin{eqnarray}
\var( B_{23} ) & = & \frac{n_{ss}^4}{n^2}   \Big( \sum_{k \in A_{ss} } \theta_k \Big)^2  \var\{( \trsi - \htrsi )( \trsj - \htrsj )\} \nonumber \\
& \le &  O \left( \frac{n_{ss}^2}{n^2} \right)  \cdot  E\{( \trsi - \htrsi )^2( \trsj - \htrsj )^2 \} . \nonumber
\end{eqnarray}
Note that for any $j= 1,\cdots,d$,
\begin{eqnarray}
\label{eqn::bound-hattauss-tauss}
|  \trsj -  \htrsj  |  & = & \Big| \sum_{k \in A_{ss}} \frac{\nk}{n_{ss}} (\trkj - \htrkj) \Big|  \leq  \max_{k \in A_{ss}} |\trkj-\htrkj| \leq  D_{n1}+D_{n0}. \nonumber \\
\end{eqnarray}
Therefore,
\begin{eqnarray}
\var( B_{23} ) & = &  O \left\{ \frac{n_{ss}^2 (D_{n1} + D_{n0} )^2 }{n^2} \right\}  \cdot  E\{   ( \trsi - \htrsi  )^2 \} . \nonumber \\
& = &  O \left\{ \frac{n_{ss}^2 (D_{n1} + D_{n0} )^2 }{n^2} \right\}  \cdot O \left( \frac{1}{n_{ss}} \right) \nonumber \\
& = & O \left\{  \frac{(D_{n1} + D_{n0} )^2}{n}  \right\}  \nonumber \\
& \rightarrow & 0. \nonumber
\end{eqnarray}
Thus, $B_{23} - E( B_{23} ) = o_p(1)$.

\item For the fourth term $B_{24}$, because subscripts $i$ and $j$ are exchangeable, we only show the convergence of the first term $B_{241} =  \sum_{k \in A_{ss} } \theta_k  ( \htrki - \trki ) ( \trkj - \trsj ) $. We have
\begin{eqnarray}
\var( B_{241} ) & = &  \frac{n_{ss}^4}{n^2}   \sum_{k \in A_{ss} } \theta_k^2 (\trkj - \trsj )^2 \cdot \var (\htrki - \trki)   \nonumber \\
& = & O \left\{  \frac{( D_{n1} + D_{n0}  )^2}{n} \cdot  \frac{1}{n} \sum_{k \in A_{ss} } (\trkj - \trsj )^2 \right\} . \nonumber
\end{eqnarray}
By Cauchy-Schwartz inequality,
$$
\oRkj^2(z)=\Big\{\sumik\frac{1}{\nk}\Rij(z)\Big\}^2  \le \Big\{\sumik\left(\frac{1}{\nk}\right)^2\Big\}\cdot \Big\{\sumik \Rij^2(z)\Big\} = \frac{1}{\nk} \sumik \Rij^2(z).
$$
Then by Condition \ref{cond::ss},
$$\sum_{k \in A_{ss} } \oRkj^2(z) \le \skS\oRkj^2(z) \le \skS  \sumik \Rij^2(z) = O(n), $$
and
$$\sum_{k \in A_{ss} }\trkj^2 = \sum_{k \in A_{ss} }\{\oRkj(1)-\oRkj(0)\}^2 \le \sum_{k \in A_{ss} }2\{\oRkj^2(1)+\oRkj^2(0)\} = O(n).$$
Using Cauchy-Schwartz inequality again, we have 
$$\trsj^2 = \left(\sum_{k \in A_{ss} } \frac{\nk}{n_{ss}}\trkj\right)^2 \le \Big\{\sum_{k \in A_{ss}} \left(\frac{\nk}{n_{ss}}\right)^2\Big\}\cdot \Big(\sum_{k \in A_{ss}}\trkj^2\Big) = O \left( \frac{1}{n_{ss} }\sum_{k \in A_{ss}}\trkj^2 \right).$$
Therefore,
$$
\sum_{k \in A_{ss} } (\trkj - \trsj )^2 \le \sum_{k \in A_{ss} } 2( \trkj^2 + \trsj^2) =   O(n).
$$
Thus, $\var(B_{241})\rightarrow 0$ and $B_{24} - E( B_{24} ) = o_p(1)$.

\item For the fifth term, because subscripts $i$ and $j$ are also exchangeable, we only show the convergence of the first term
$$
B_{251} =  \frac{ n^2_{ss} }{n}   \sum_{k \in A_{ss} } \theta_k   ( \htrki - \trki )  \cdot ( \trsj - \htrsj ).
$$
By \eqref{eqn::bound-hattauss-tauss} we have  
$$
| B_{251} | \leq  \frac{D_{n1}+D_{n0} } {\sqrt{n}} \cdot  \Big|  \frac{n^2_{ss} }{ \sqrt{n} }   \sum_{k \in A_{ss} } \theta_k  ( \htrki - \trki ) \Big| .
$$
Therefore,
\begin{eqnarray}
\var( B_{251} ) & \leq  & E (B_{251})^2 \nonumber \\
& \leq & \frac{ (D_{n1}+D_{n0})^2 } { n } \cdot  \frac{ n^4_{ss} }{ n } \cdot E \left\{ \sum_{k \in A_{ss} } \theta_k   (\htrki - \trki) \right\}^2  \nonumber \\
& = &  \frac{ (D_{n1}+D_{n0})^2 } { n } \cdot  \frac{n^4_{ss} }{ n } \cdot \sum_{k \in A_{ss} } \theta_k^2  E   (\htrki - \trki)^2  \nonumber \\
& \le & O \left\{  \frac{ (D_{n1} + D_{n0})^2 }{n}  \right\}  \nonumber \\
& \rightarrow & 0,  \nonumber 
\end{eqnarray}
where the convergence is because of Condition \ref{cond::max}. Thus, $B_{25} - E( B_{25} ) = o_p(1)$.

\item For the last term $B_{26}$, similar to $B_{24}$ and $B_{25}$, we have
\begin{eqnarray}
\var( B_{261} ) & = & \frac{n_{ss}^4}{n^2}   \sum_{k \in A_{ss} } \theta_k^2  ( \trki - \trsi )^2 \cdot  \var( \trsj - \htrsj ) \nonumber \\
& = &  O \left\{ \frac{1}{n^2} \cdot \frac{1}{n_{ss}} \sum_{k \in A_{ss} }(\trki-\trsi)^2   \right\}
= O\left( \frac{1}{n^2}\right) \rightarrow  0. \nonumber
\end{eqnarray}
Thus, $B_{26} - E( B_{26} ) = o_p(1)$.
\end{itemize}
Combining the above arguments, we have $B_2 - E(B_2) = o_p(1)$.
\end{proof}




\subsubsection{Proof of Corollary \ref{cor1}}

\begin{proof}
As the covariates can be considered potential outcomes unaffected by the treatment assignment, we can apply Theorem \ref{thm1CLT} to $R_i(z)=(\Yi(z),\ \Xi^\T)^\T$. Conditions \ref{cond::max} and \ref{cond::covariance} can be deduced from Conditions \ref{cond4} and \ref{cond5}, and hence the corollary holds.
\end{proof}

\subsection{SRRoM}

\subsubsection{Proof of Proposition \ref{prop1pa}}

\begin{proof}
    According to Corollary \ref{cor1}, $n^{1/2}\htX\ \dot{\sim}\ \mathcal{N}(0,\Sigma_{xx})$. Then the asymptotic distribution of the Mahalanobis distance is
	\begin{equation*}
		M_{\htX}=(n^{1/2}\htX)^\T \Sigma_{xx}^{-1}(n^{1/2}\htX)\  \dot{\sim}\ \chi^2_{\p}.
	\end{equation*}
    Therefore, the probability of a random assignment being accepted is 
	\[ \pr(M_{\htX}<a)\to p_a=\pr(\chi^2_{\p}<a) \] as $n$ tends to infinity.
\end{proof}

To prove Theorem \ref{thm2}, we need the following Lemma which directly extends the results of \citet{Li2018} from complete rerandomization to stratified rerandomization.  Let $\phi(\eta, A) : \mathbb{R}^{\p}\times \mathbb{R}^{\p\times\p}\to\{0,1\}$ be an indicator function of covariate balance under the criterion $\eta^\T A^{-1}\eta<a$. Then under SRRoM, $\Mtau\iff\phi(n^{1/2}\htX,\Sigma_{xx})=1$.

\begin{lemma}
\label{corA1}
	Under SRRoM, \[n^{1/2}(\hat\tau-\tau, \htX^\T)^\T\mid\Mtau\ \dot\sim\  (A,B^\T)^\T\mid\phi(B,\Sigma_{xx})=1,\]
	where $(A,B^\T)^\T\sim\mathcal{N}(0,\Sigma)$.
\end{lemma}

\begin{proof}
    As $\Mtau$ represents the event that $M_{\htX}=(n^{\frac{1}{2}}\htX)^\T \Sigma_{xx}^{-1}(n^{\frac{1}{2}}\htX)<a$, the covariate balance criterion $\phi(n^{1/2}\htX,\Sigma_{xx})$ satisfies the Condition $\textnormal{A1}$ proposed in \citet{Li2018}. According to Corollary A1 of \citet{Li2018}, this lemma holds.
\end{proof}

\subsubsection{Proof of Theorem \ref{thm2}}
\begin{proof}
    Let $(A,B^\T)^\T\sim\mathcal{N}(0,\Sigma)$ be the same as that in Lemma \ref{corA1}. 
As the linear projection of $A$ on $B$ is $\Sigma_{\tau x}\Sigma_{xx}^{-1}B$, whose variance is $c^2=\Sigma_{\tau x}\Sigma_{xx}^{-1}\Sigma_{x\tau}=\Sigma_{\tau\tau}R^2$, and the projection residual is $\epsilon=A-\Sigma_{\tau x}\Sigma_{xx}^{-1}B\sim\mathcal{N}(0,(1-R^2)\Sigma_{\tau\tau})\sim\{\Sigma_{\tau\tau}(1-R^2)\}^{1/2}\epsilon_0$, we have 
$$
A=\epsilon+\Sigma_{\tau x}\Sigma_{xx}^{-1}B=\epsilon+\Sigma_{\tau x}\Sigma_{xx}^{-1/2}D=\epsilon+ch^\T D,
$$
where $h^\T = \Sigma_{\tau x} \Sigma_{xx}^{-1/2} /c$ is the normalized vector of $\Sigma_{\tau x}\Sigma_{xx}^{-1/2}$ and $D = \Sigma_{xx}^{-1/2} B \sim N(0,I)$. Because $\phi(B,\Sigma_{xx})=1$ if and only if $B^\T \Sigma_{xx}^{-1}B=D^\T D<a$, then according to Lemma \ref{lemA1}, 
\[A\mid\phi(B,\Sigma_{xx})=1 \  \sim\ \epsilon+ch^\T D\mid D^\T D\le a\ \sim\ \epsilon+cL_{\p,a}\ \sim \  \Sigma_{\tau\tau}^{\frac{1}{2}}\big\{(1-R^2)^{\frac{1}{2}}\epsilon_0+(R^2)^{\frac{1}{2}}L_{\p,a}\big\}.\]
Combining the above result with Lemma \ref{corA1}, this theorem holds.

\end{proof}

\subsubsection{Proof of Corollary \ref{cor3}}
\begin{proof}
    According to Theorem \ref{thm2}, the asymptotic expectation of $\ntau$ is the limit of \[\Sigma_{\tau\tau}^{\frac{1}{2}}\big\{(1-R^2)^{\frac{1}{2}} E(\epsilon_0 ) + (R^2)^{\frac{1}{2}} E( L_{\p,a} ) \big\}=0.\] Thus, $\hat\tau$ is an asymptotically unbiased estimator of $\tau$.
    
    The asymptotic variance of $\ntau$ is the limit of \[\Sigma_{\tau\tau}\{(1-R^2) \var( \epsilon_0 )+ R^2  \var( L_{\p,a} ) \} =  \Sigma_{\tau\tau}\{(1-R^2)+ R^2 v_{\p,a}\}= \Sigma_{\tau\tau}\big\{1-(1-v_{\p,a})R^2\big\}.\]
    According to Proposition \ref{prop::covariance}, the asymptotic variance of $\ntau$ under stratified randomization is $\Sigma_{\tau\tau}$. 
    Thus, compared to stratified randomization, the percentage of reduction in asymptotic variance is the limit of \[\big[\Sigma_{\tau\tau}-\Sigma_{\tau\tau}\big\{1-(1-v_{\p,a})R^2\big\}\big]/\Sigma_{\tau\tau}=(1-v_{\p,a})R^2.\]

\end{proof}

\subsubsection{Proof of Proposition \ref{prop::gamma}}

\begin{proof}
Proposition~\ref{prop::gamma} follows directly from Lemma~\ref{lemsAB}. 
\end{proof}



\subsubsection{Proof of Corollary \ref{thm3}}
\begin{proof}
    According to Theorem \ref{thm2}, the asymptotic distribution of $\ntau$ under SRRoM has the same form as that in Theorem $\textnormal{1}$ of \citet{Li2018}. Therefore, the corollary follows from Theorem $\textnormal{2}$ of \citet{Li2018}.

\end{proof}

\subsubsection{Proof of Theorem \ref{thmvaro}}

Before proving Theorem \ref{thmvaro}, we present a useful lemma. Let 
$$
\tilde{\Sigma}_{\tau\tau}= \Sigma_{\tau \tau} + \sum_{k \notin A_{ss}} \pik\Skt + \frac{n_{ss}^2}{n} \sum_{k \in A_{ss} }  \frac{  \nk^2 } { ( n_{ss} - 2 \nk  ) \Big( n_{ss} + \sum_{h \in A_{ss}} \frac{ n_{[h]}^2  }{ n_{ss} - 2  n_{[h]} }  \Big)  }  ( \tau_{[k]} -   \tau_{ss} )^2 \ge \Sigma_{\tau\tau} .
$$
\begin{lemma}
\label{lem::conservative-sb}
Under the conditions of Theorem \ref{thmvaro}, 
$$
E( \hat  \Sigma_{\tau\tau} ) = \tilde{\Sigma}_{\tau\tau} \quad \textnormal{and} \quad  \hat  \Sigma_{\tau\tau} - \tilde{\Sigma}_{\tau\tau} = o_p(1).
$$ 
\end{lemma}


Lemma~\ref{lem::conservative-sb} is a direct result of Theorem~\ref{thm:cov-est}, so we omit its proof. Now, based on Lemmas~\ref{lemsAB} and \ref{lem::conservative-sb}, we can prove Theorem \ref{thmvaro}.

\begin{proof}[Proof of Theorem \ref{thmvaro}]

According to Lemma \ref{lemsAB}, $ \hat\Sigma_{x\tau}$ is a consistent estimator of $\Sigma_{x\tau}=\Sigma_{\tau x}^\T $. According to Lemma \ref{lem::conservative-sb}, $\hat  \Sigma_{\tau\tau}$ is a consistent estimator of $\tilde \Sigma_{\tau\tau}$.
As $\tilde{\Sigma}^{\infty}_{\tau\tau}\ge \Sigma^{\infty}_{\tau\tau}$, where we use the superscript $\infty$ to denote the limit value as $n$ tends to infinity, $\hat \Sigma_{\tau\tau}$ is an asymptotically conservative estimator of $\Sigma_{\tau\tau}$. 
Since $\hat{R}^2=\hat{\Sigma}_{\tau x}\Sigma_{xx}^{-1}\hat \Sigma_{x\tau}/\hat \Sigma_{\tau\tau}$, then by the consistency of $\hat{\Sigma}_{\tau x}$,
$$\hat \Sigma_{\tau\tau}\hat{R}^2-\Sigma_{\tau\tau}R^2 = \hat \Sigma_{\tau x} \Sigma_{xx}^{-1} \hat \Sigma_{x \tau} -   \Sigma_{\tau x} \Sigma_{xx}^{-1} \Sigma_{x \tau}=o_p(1). $$
Now we have
\begin{equation*}
	\hat \Sigma_{\tau\tau}-(1-v_{\p,a})\hat \Sigma_{\tau\tau}\hat{R}^2-\big\{\tilde{\Sigma}_{\tau\tau}-(1-v_{\p,a})\Sigma_{\tau\tau}R^2\big\}=o_p(1).
\end{equation*}

Thus, the probability limit of $\hat \Sigma_{\tau\tau}\{1-(1-v_{\p,a})\hat{R}^2\}$ is larger than or equal to that of $\Sigma_{\tau\tau}\{1-(1-v_{\p,a})R^2\}$, which, according to Corollary \ref{cor3}, is the asymptotic variance of $\ntau$ under SRRoM. 

By continuous mapping theorem, 
\begin{equation}
\label{eqn::limits}
 \hat \Sigma_{\tau\tau}^{\frac{1}{2}}\big\{(1-\hat R^2)^{\frac{1}{2}}\epsilon_0+(\hat R^2)^{\frac{1}{2}}L_{\p,a}\big\} \to \big\{ \tilde \Sigma^{\infty}_{\tau\tau}-\Sigma^{\infty}_{\tau\tau} (R^\infty)^2 \big\}^{\frac{1}{2}}\epsilon_0+ \big\{\Sigma^{\infty}_{\tau\tau} (R^\infty)^2\big\}^{\frac{1}{2}} L_{\p,a}
\end{equation}
in distribution. 
According to Lemma \ref{lemA7}, we have that the $\CI$ quantile range of 
\[\{ \tilde \Sigma^{\infty}_{\tau\tau}-\Sigma^{\infty}_{\tau\tau} (R^\infty)^2 \}^{ 1/2 }\epsilon_0+ \{\Sigma^{\infty}_{\tau\tau} (R^\infty)^2\}^{ 1/2 }L_{\p,a}\]
includes that of $ \{ \Sigma^{\infty}_{\tau\tau}-\Sigma^{\infty}_{\tau\tau} (R^\infty)^2 \}^{ 1/2 }\epsilon_0+ \{\Sigma^{\infty}_{\tau\tau} (R^\infty)^2 \}^{ 1/2 }L_{\p,a} $, which is identically distributed with the asymptotic distribution of $\ntau$ under SRRoM. Therefore, the limit of the coverage probability of the  confidence interval for $\tau$: 
\begin{equation*}
	\big[\hat\tau-(\hat \Sigma_{\tau\tau}/n)^{\frac{1}{2}}\nu_{1-\alpha/2}(\hat R^2,p_a,p),\ \hat\tau-(\hat \Sigma_{\tau\tau}/n)^{\frac{1}{2}} \nu_{\alpha/2}(\hat R^2,p_a,p)\big]
\end{equation*}
is no less than $\CI$.

\end{proof}

\subsection{SRRsM}

\subsubsection{Proof of Proposition \ref{prop3pa}}
\begin{proof}
    Because the asymptotic probability of accepting a random assignment for stratum $k$ is $p_{a_k}=\pr(\chi^2_{\p}<a_k)$ \citep{Li2018} and the rerandomization is carried out separately and independently across strata, the asymptotic probability of accepting an assignment for all units is $\prod_{k=1}^\S p_{a_k}$.
\end{proof}

\subsubsection{Proof of Theorem \ref{thm5}}
\begin{proof}
    The random variable $\ntau$ can be decomposed as \[n^{\frac{1}{2}}(\hat\tau-\tau) = \skS \pik^{\frac{1}{2}}\nk^{\frac{1}{2}}(\htk-\tauk).\] 
    Applying Theorem $\textnormal{1}$ of \citet{Li2018} (or Theorem \ref{thm2} with $\S = 1$ in the main text)  within each stratum $k$, we have
    \[  \nk^{\frac{1}{2}}(\htk-\tauk)  \mid \mathcal{M}_{[k]} \ \dot{\sim} \  \Vktt^{\frac{1}{2}}\big\{(1-\Rk)^{\frac{1}{2}}\epsilon_0^k+(\Rk)^{\frac{1}{2}}L_{p,a_k}^k\big\}, \quad k = 1,\dots,K, \]
    where $\mathcal{M}_{[k]} = \{ (Z_i)_{i \in [k] } : \ \Mk < a_k \}$ denotes  the event that an assignment within stratum $k$ is accepted by SRRsM. Since the rerandomization is conducted independently across strata, then
    \[ \begin{split}
    n^{\frac{1}{2}}(\hat\tau-\tau)\mid\MI\ \dot{\sim}\
    &\skS \pik^{\frac{1}{2}}\Vktt^{\frac{1}{2}}\Big\{(1-\Rk)^{\frac{1}{2}}\epsilon_0^k+(\Rk)^{\frac{1}{2}}L_{p,a_k}^k\Big\} \\
    \sim \ &\Big\{\skS\pik\Vktt(1-\Rk)\Big\}^{\frac{1}{2}}\epsilon_0+\skS(\pik \Vktt \Rk)^{\frac{1}{2}}L_{p,a_k}^k, 
    \end{split} \]
    where $\epsilon_0,\ \epsilon_0^k\sim\mathcal{N}(0,1)$ and $L_{p,a_k}^k \sim L_{p,a_k},\ k=1,\ldots,\S$ are mutually independent.

\end{proof}

\subsubsection{Proof of Corollary \ref{cor4}}
\begin{proof}
    According to Theorem \ref{thm5}, the asymptotic expectation of $\ntau$ is \[\Big\{\skS\pik\Vktt(1-\Rk)\Big\}^{\frac{1}{2}} E( \epsilon_0 ) + \skS(\pik \Vktt \Rk)^{\frac{1}{2}} E ( L_{p,a_k}^k  ) =0,\] hence $\hat\tau$ is asymptotically unbiased.
    
    The asymptotic variance of $\ntau$ is the limit of \[\skS\pik\Vktt(1-\Rk)+ \skS\pik\Vktt\Rk v_{\p,a_k} =\skS\pik\Vktt\big\{1-(1-v_{\p,a_k})\Rk\big\},\] thus the reduction in asymptotic variance compared to stratified randomization is the limit of $\skS\pik\Vktt(1-v_{\p,a_k})\Rk / \Sigma_{\tau\tau}$.

\end{proof}

\subsubsection{Proof of Theorem \ref{prop3}}
\begin{proof}
When $a_1=\cdots=a_\S=a$, as $v_{\p,a_k} = v_{\p,a}$ and $ \Sigma_{\tau\tau} = \skS\pik\Vktt $, then according to Corollaries \ref{cor3} and \ref{cor4}, the difference between the asymptotic variances of $\ntau$ under SRRsM and SRRoM is the limit of  
\begin{equation}
\label{eqprop4}
\begin{split}
	&\skS\pik\Vktt\big\{1-(1-v_{\p,a_k})\Rk\big\} - \Sigma_{\tau\tau}\big\{1-(1-v_{\p,a})R^2\big\} \\ 
	=&\skS\pik \Vktt(1-v_{\p,a})(R^2-\Rk) \\
	= & (1-v_{\p,a})\Big[\Sigma_{\tau x}\Sigma_{xx}^{-1}\Sigma_{x\tau}-\skS\pik \Vktx \Vkxx^{-1}\Vkxt\Big].
\end{split}
\end{equation}
Let \[b=\big((\pi_{[1]}^{1/2}\Sigma_{[1]xx}^{-1/2}\Sigma_{[1]x\tau})^\T,\ldots,(\pi_{[\S]}^{1/2}\Sigma_{[\S]xx}^{-1/2}\Sigma_{[\S]x\tau})^\T\big)^\T\] and \[d=\big((\pi_{[1]}^{1/2}\Sigma_{[1]xx}^{1/2}\Sigma_{xx}^{-1}\Sigma_{x\tau})^\T,\ldots,(\pi_{[\S]}^{1/2}\Sigma_{[\S]xx}^{1/2}\Sigma_{xx}^{-1}\Sigma_{x\tau})^\T\big)^\T\] be two $(\S\p)\times1$ vectors. By Cauchy-Schwarz inequality and 
$$\Sigma_{xx} = \skS \pik \Sigma_{[k]xx}, \quad  \Sigma_{x \tau} = \skS \pik \Sigma_{[k]x \tau}, $$ 
we have \[ \begin{split}
\bigg(\skS\pik \Vktx \Vkxx^{-1}\Vkxt\bigg) \big(\Sigma_{\tau x}\Sigma_{xx}^{-1}\Sigma_{x\tau}\big)=(b^\T b)(d^\T d)
\ge(b^\T d)^2=\big(\Sigma_{\tau x}\Sigma_{xx}^{-1}\Sigma_{x\tau}\big)^2,
\end{split}\]
where the equality holds if and only if $b=\lambda d$ for some $\lambda$. Since $\Sigma_{\tau x}\Sigma_{xx}^{-1}\Sigma_{x\tau} \geq 0$, then (\ref{eqprop4}) is smaller than or equal to 0. When $b=\lambda d$ for some $\lambda$, $\Vkxt= \lambda \Vkxx\Sigma_{xx}^{-1}\Sigma_{x\tau},\ k=1,\ldots,\S$, hence
\[ \Sigma_{x\tau}= \lambda \skS\pik\Vkxx\Sigma_{xx}^{-1}\Sigma_{x\tau} = \lambda\Sigma_{xx}\Sigma_{xx}^{-1}\Sigma_{x\tau} = \lambda\Sigma_{x\tau}.
\]
Then $\lambda = 1$. Therefore, the equality holds if and only if $\Vkxx^{-1}\Vkxt=\Sigma_{xx}^{-1}\Sigma_{x\tau}$ for $k=1,\ldots,\S$.

Finally, because $v_{\p,a}\to 0$ as $a\to 0$, the same conclusion holds when all the thresholds tend to 0.

\end{proof}

\subsubsection{Proof of Corollary \ref{thm6}}
\begin{proof}
    Let $f(x),\ f_{(1)}(x),$ and $f_1(x)$ be the probability density functions of random variables 
	\begin{equation*}
	\begin{split}
		\skS(\pik \Vktt)^{\frac{1}{2}} \big\{(1-\Rk)^{\frac{1}{2}}\epsilon_0^k+(\Rk)^{\frac{1}{2}}L_{p,a_k}^k \big\},\\
		\sum_{k=2}^\S(\pik \Vktt)^{\frac{1}{2}} \big\{(1-\Rk)^{\frac{1}{2}}\epsilon_0^k+(\Rk)^{\frac{1}{2}}L_{p,a_k}^k \big\},
		\end{split}
	\end{equation*}
	and 
	$$(\pi_{[1]}V_{[1]\tau\tau})^{\frac{1}{2}} \big\{(1-R_{[1]}^2)^{\frac{1}{2}}\epsilon_0^1+(R_{[1]}^2)^{\frac{1}{2}}L_{\p,a_1}^1 \big\},$$
	respectively. For notation simplicity, denote $\P(R_{[1]}^2,c)$ as the probability 
	\[\pr\big[(\pi_{[1]}V_{[1]\tau\tau})^{\frac{1}{2}} \big\{(1-R_{[1]}^2)^{\frac{1}{2}}\epsilon_0^1+(R_{[1]}^2)^{\frac{1}{2}}L_{\p,a_1}^1\big\}\ge c\big],\]
	which is a function of $R_{[1]}^2$ and $c$.
	Then
	\begin{equation}
	\label{eqpr}
	\begin{split}
	&\pr\Big[\Big\{\skS\pik\Vktt(1-\Rk)\Big\}^{\frac{1}{2}}\epsilon_0+\skS(\pik \Vktt \Rk)^{\frac{1}{2}}L_{p,a_k}^k\ge c\Big]\\
	=&\int_{c}^{+\infty}f(x)\dd x = \int_{c}^{+\infty}\int_{-\infty}^{+\infty}f_1(x-y)f_{(1)}(y)\dd y\dd x \\
	=&\int_{-\infty}^{+\infty}f_{(1)}(y)\dd y\int_{c}^{+\infty}f_1(x-y)\dd x \\
	=&\int_{-\infty}^{+\infty}f_{(1)}(y)\P(R_{[1]}^2,c-y)\dd y.
	\end{split}
	\end{equation}
	Since $\epsilon_0^k$ and $L^k_{\p,a}$ are symmetric and unimodal around $0$, $f_{(1)}(x)$ and $f_1(x)$ are also symmetric and unimodal around $0$. By Lemma \ref{lemA4}, $\P(R_{[1]}^2,y)-\P(\tilde R_{[1]}^2,y)\ge0$ when $y\ge0$ for $0\le R_{[1]}^2\le\tilde{R}_{[1]}^2\le1$. Then when $y\ge 0$,
	\begin{equation*}
	\begin{split}
	\P(R_{[1]}^2,-y)-\P(\tilde R_{[1]}^2,-y)
	=&\big\{1-\P(R_{[1]}^2,y)\big\} - \big\{1-\P(\tilde R_{[1]}^2,y)\big\} \\
	=&\P(\tilde R_{[1]}^2,y)-\P(R_{[1]}^2,y)\le0.
	\end{split}
	\end{equation*} 
	Thus, 
	\begin{equation*}
		\begin{split}
		&\pr\Big[\Big\{\skS\pik\Vktt(1-\Rk)\Big\}^{\frac{1}{2}}\epsilon_0+\skS(\pik \Vktt \Rk)^{\frac{1}{2}}L_{p,a_k}^k\ge c\Big]\\
		-&\pr\Big[\Big\{\skS\pik\Vktt(1-\tRk)\Big\}^{\frac{1}{2}}\epsilon_0+\skS(\pik \Vktt\tRk)^{\frac{1}{2}}L_{p,a_k}^k\ge c\Big] \\
		=&\int_{-\infty}^{+\infty}f_{(1)}(y)\big\{\P(R_{[1]}^2,c-y)-\P(\tilde R_{[1]}^2,c-y)\big\}\dd y \\
		=&\int_{-\infty}^{+\infty}f_{(1)}(y-c)\big\{\P(R_{[1]}^2,y)-\P(\tilde R_{[1]}^2,y)\big\}\dd y \\
		=&\int_{-\infty}^{0}f_{(1)}(y-c)\big\{\P(R_{[1]}^2,y)-\P(\tilde R_{[1]}^2,y)\big\}\dd y + \int_{0}^{+\infty}f_{(1)}(y-c)\big\{\P(R_{[1]}^2,y)-\P(\tilde R_{[1]}^2,y)\big\}\dd y \\
		=&\int_{-\infty}^{0}f_{(1)}(y-c)\big\{\P(R_{[1]}^2,y)-\P(\tilde R_{[1]}^2,y)\big\}\dd y - \int_{-\infty}^{0}f_{(1)}(y+c)\big\{\P(R_{[1]}^2,y)-\P(\tilde R_{[1]}^2,y)\big\}\dd y \\
		=&\int_{-\infty}^{0}\big\{f_{(1)}(y-c)-f_{(1)}(y+c)\big\}\big\{\P(R_{[1]}^2,y)-\P(\tilde R_{[1]}^2,y)\big\}\dd y \ge0.
		\end{split}
	\end{equation*}
	Hence the probability in (\ref{eqpr}) is a nonincreasing function of $R_{[1]}^2$. Similarly, the conclusion holds  for $R_{[2]}^2,\ldots,R_{[\S]}^2$. Therefore, the quantile $q_{1-\alpha/2}(R_{[1]}^2,\ldots,R_{[\S]}^2,p_{a_1},\ldots,p_{a_\S},p)$ is a nonincreasing function of $\Rk\ (k=1,\ldots,\S)$ with $R_{[l]}^2$'s $(l\neq k)$, $p_{a_k}$'s, and $\p$ being fixed.
	
By Lemma \ref{lemA5}, for $0\le \pa\le\tpa\le1$ and $c\ge0$, \[\pr(L_{\p,F_\p^{-1}(\pa)}\ge c)\le\pr(L_{\p,F_\p^{-1}(\tpa)}\ge c),\] then by Lemma \ref{lemA7}, \[\begin{split}
	&\pr\Big[\Big\{\skS\pik\Vktt(1-\Rk)\Big\}^{\frac{1}{2}}\epsilon_0+\sum_{k=2}^\S(\pik \Vktt \Rk)^{\frac{1}{2}}L_{p,a_k}^k \\
	&\ \ \ \ +(\pi_{[1]}V_{[1]\tau\tau}R_{[1]}^2)^{\frac{1}{2}}L_{\p,F_{\p}^{-1}(\pa)}^1\ge c\Big] \\
	\le\ &\pr\Big[\Big\{\skS\pik\Vktt(1-\Rk)\Big\}^{\frac{1}{2}}\epsilon_0+\sum_{k=2}^\S(\pik \Vktt \Rk)^{\frac{1}{2}}L_{p,a_k}^k \\
	&\ \ \ \ +(\pi_{[1]}V_{[1]\tau\tau}R_{[1]}^2)^{\frac{1}{2}}L_{\p,F_{\p}^{-1}(\tpa)}^1\ge c\Big].
	\end{split}\]
	Hence the quantile $q_{1-\alpha/2}(R_{[1]}^2,\ldots,R_{[\S]}^2,p_{a_1},\ldots,p_{a_\S},p)$ is a nondecreasing function of $p_{a_k}\ (k=1,\ldots,\S)$ with $p_{a_l}$'s $(l\neq k)$, $R_{[l]}^2$'s, and $\p$ being fixed.
	
	By Lemma \ref{lemA6}, for $\tilde\p\ge\p\ge1$ and $c\ge0$, \[\pr(L_{\p,F_{\p}^{-1}(\pa)}\ge c)\le\pr(L_{\tilde\p,F_{\tilde\p}^{-1}(\pa)}\ge c),\] then by Lemma \ref{lemA7}, 
	\[\begin{split}
	&\pr\Big[\Big\{\skS\pik\Vktt(1-\Rk)\Big\}^{\frac{1}{2}}\epsilon_0+\skS(\pik \Vktt \Rk)^{\frac{1}{2}}L_{\p,F_{\p}^{-1}(p_{a_k})}^k\ge c\Big] \\
	\le\ &\pr\Big[\Big\{\skS\pik\Vktt(1-\Rk)\Big\}^{\frac{1}{2}}\epsilon_0+\sum_{k=2}^\S(\pik\Vktt\Rk)^{\frac{1}{2}}L_{\p,F_{\p}^{-1}(p_{a_k})}^k \\
	&\ \ \ \ +(\pi_{[1]}V_{[1]\tau\tau}R_{[1]}^2)^{\frac{1}{2}}L_{\tilde\p,F_{\tilde\p}^{-1}(p_{a_1})}^1\ge c\Big] \\
	\le\ &\pr\Big[\Big\{\skS\pik\Vktt(1-\Rk)\Big\}^{\frac{1}{2}}\epsilon_0+\sum_{k=3}^\S(\pik \Vktt \Rk)^{\frac{1}{2}}L_{\p,F_{\p}^{-1}(p_{a_k})}^k \\
	&\ \ \ \ +\sum_{k=1}^2(\pik \Vktt \Rk)^{\frac{1}{2}}L_{\tilde\p,F_{\tilde\p}^{-1}(p_{a_k})}^k\ge c\Big] \\
	\le\ &\cdots\le\ \pr\Big[\Big\{\skS\pik\Vktt(1-\Rk)\Big\}^{\frac{1}{2}}\epsilon_0+\skS(\pik\Vktt\Rk)^{\frac{1}{2}}L_{\tilde\p,F_{\tilde\p}^{-1}(p_{a_k})}^k\ge c\Big].
	\end{split}\]
	Hence the quantile $q_{1-\alpha/2}(R_{[1]}^2,\ldots,R_{[\S]}^2,p_{a_1},\ldots,p_{a_\S},\p)$ is a nondecreasing function of $\p$ with $p_{a_k}$'s and $\Rk$'s being fixed.

\end{proof}

\subsubsection{Proof of Theorem \ref{thmvars}}
\begin{proof}
The asymptotic conservativeness of the estimator for the variance of $\ntau$ under SRRsM is because of the conservativeness of the variance estimators in each stratum $k$ \citep{Li2018}. Let $\tVktt=\Vktt+\Skt-S_{[k]\tau|x}^2\ge \Vktt$. Similar to \eqref{eqn::limits}, by continuous mapping theorem we have
\[\begin{split} &\Big\{\skS\pik\hVktt(1-\hRk)\Big\}^{\frac{1}{2}}\epsilon_0+\skS(\pik\hVktt\hRk)^{\frac{1}{2}}L_{p,a_k}^k \\
\to &\Big[\skS\pik \big\{ \tilde\Sigma^{\infty}_{[k]\tau\tau}-\Sigma^{\infty}_{[k]\tau\tau}(R_{[k]}^{\infty})^2 \big\} \Big]^{\frac{1}{2}}\epsilon_0+\skS \big\{ \pik \Sigma^{\infty}_{[k]\tau\tau} R_{[k]}^{\infty})^2 \big\}^{\frac{1}{2}}L_{p,a_k}^k
\end{split}\] in distribution as $n\to\infty$. By Lemma \ref{lemA7},
\[\begin{split}
&\pr\bigg(   \Big[\skS\pik \big\{ \tilde\Sigma^{\infty}_{[k]\tau\tau}-\Sigma^{\infty}_{[k]\tau\tau}(R_{[k]}^{\infty})^2 \big\} \Big]^{\frac{1}{2}}\epsilon_0 + \big\{ \pi_{[1]}V_{[1]\tau\tau}^{\infty}(R_{[1]}^{\infty})^2 \big\}^{\frac{1}{2}}L_{\p,a_1}^1\ge c\bigg) \\
\ge & \pr\bigg(   \Big[\skS\pik \big\{ \Sigma^{\infty}_{[k]\tau\tau}-\Sigma^{\infty}_{[k]\tau\tau}(R_{[k]}^{\infty})^2 \big\} \Big]^{\frac{1}{2}}\epsilon_0 + \big\{ \pi_{[1]}V_{[1]\tau\tau}^{\infty}(R_{[1]}^{\infty})^2 \big\}^{\frac{1}{2}}L_{\p,a_1}^1\ge c\bigg), \\
\end{split}\] for any $c>0$. Applying Lemma \ref{lemA7} again, we have
\[\begin{split}
& \pr\bigg(   \Big[\skS\pik \big\{ \tilde\Sigma^{\infty}_{[k]\tau\tau}-\Sigma^{\infty}_{[k]\tau\tau}(R_{[k]}^{\infty})^2 \big\} \Big]^{\frac{1}{2}}\epsilon_0 + \sum_{k=1}^2 \big\{ \pik \Sigma^{\infty}_{[k]\tau\tau}(R_{[k]}^{\infty})^2 \big\}^{\frac{1}{2}}L_{p,a_k}^k \ge c\bigg) \\
\ge & \pr\bigg(   \Big[\skS\pik \big\{ \Sigma^{\infty}_{[k]\tau\tau}-\Sigma^{\infty}_{[k]\tau\tau}(R_{[k]}^{\infty})^2 \big\} \Big]^{\frac{1}{2}}\epsilon_0 + \sum_{k=1}^2 \big\{ \pik \Sigma^{\infty}_{[k]\tau\tau}(R_{[k]}^{\infty})^2 \big\}^{\frac{1}{2}}L_{p,a_k}^k \ge c\bigg), \\
\end{split}\] for any $c>0$. Apply Lemma \ref{lemA7} for $\S$ times, we have that the limit of the coverage probability of the confidence interval of $\tau$ is no less than $\CI$.

\end{proof}

\subsection{SRRdM}

\subsubsection{Proof of Proposition \ref{props1}}
\begin{proof}
    First we define individual level pseudo potential outcomes $W$ as $\Wi(1)=\pk\Xi/p_1$ and $\Wi(0)=(1-\pk)\Xi/p_0$ for $i\in{[k]}\ (k=1,\ldots,K)$,
	and let $\Ymul_i(z)=(\Yi(z),\Wi(z)^\T)^\T$. The average treatment effect of $W$ is
\begin{equation*}
\begin{split}
	{\tW}=& \frac{1}{n}\skS \sumk \Big( \frac{\pk \Xi}{ p_1 } - \frac{ ( 1 - \pk ) \Xi }{ p_0 } \Big) \\
	=  & \frac{1}{n}\skS \sumk\frac{\pk-p_1}{p_1p_0}\Xi
	=\frac{1}{p_1p_0}\skS\pik (\pk-p_1)\bXk,
\end{split}
\end{equation*}
which is fixed and known in the design stage, and the stratified difference-in-means estimator for $\tW$ is
\begin{equation*}
\begin{split}
	\htW=&\skS\pik \Big\{\frac{1}{\nkt}\sumk\Wi(1)\Zi-\frac{1}{\nkc}\sumk\Wi(0)(1-\Zi)\Big\}\\
	=&\skS\sumk\Big\{\frac{1}{n_1}\Xi\Zi-\frac{1}{n_0}\Xi(1-\Zi)\Big\}\\
	=&\ttX.
\end{split}
\end{equation*} 
   Recall that  $\rkz=z\pk+(1-z)(1-\pk)$, $z=0,1,\  k=1,\ldots,\S$. Since 
	\[\begin{split}
	&\SkWY(z)=\frac{1}{\nk-1}\sumk\{\Wi(z)-\bWk(z)\}\{\Yi(z)-\bYk(z)\}^\T \\
	=&\frac{1}{\nk-1}\sumk\frac{\rkz}{p_z}(\Xi-\bXk)\{\Yi(z)-\bYk(z)\}^\T=\frac{\rkz}{p_z}\SkXY(z),
	\end{split}\]
	and
	\[\begin{split}
	&\SkW(z)=\frac{1}{\nk-1}\sumk\{\Wi(z)-\bWk(z)\}\{\Wi(z)-\bWk(z)\}^\T \\
	=&\frac{1}{\nk-1}\sumk\Big(\frac{\rkz}{p_z}\Big)^2(\Xi-\bXk)(\Xi-\bXk)^\T=\Big(\frac{\rkz}{p_z}\Big)^2\SkX,
	\end{split}\]
    then the stratum-specific covariance of $\Ymul_i(z)=(\Yi(z),\Wi(z)^\T)^\T$ are
    \[
	\SkYmul(z)=\left(\begin{array}{cc}
	    \SkY(z) & \frac{\rkz}{p_z}\SkXY^\T(z) \\
	    \frac{\rkz}{p_z}\SkXY(z) & \big(\frac{\rkz}{p_z}\big)^2\SkX
	\end{array}\right),\ z=0,1.
	\]
	As $\tau_{i,\Ymul}=(\taui,\ (\pk-p_1)/(p_1p_0)\Xi^\T)^\T$, the stratum-specific covariance of $\tmul$ is
	\[
	\Sktmul=\left(\begin{array}{cc}
	    \Skt & \frac{\pk-p_1}{p_1p_0}\big\{\SkXY^\T(1)-\SkXY^\T(0)\big\} \\
	    \frac{\pk-p_1}{p_1p_0}\big\{\SkXY(1)-\SkXY(0)\big\} & \big(\frac{\pk-p_1}{p_1p_0}\big)^2\SkX
	\end{array}\right).
	\]
	Thus, according to Proposition \ref{prop0}, the upper left block of $\cov\{n^{1/2}(\htmul-\tmul)\}$ is \[ \frac{\SkY(1)}{\pk}+\frac{\SkY(0)}{1-\pk}-\Skt, \] the upper right block is 
	\[\begin{split} &\frac{\SkXY^\T(1)}{p_1}+\frac{\SkXY^\T(0)}{p_0}-\frac{(\pk-p_1)\{\SkXY^\T(1)-\SkXY^\T(0)\}}{p_1p_0} \\
	=&\frac{(1-\pk)\SkXY^\T(1)}{p_1p_0}+\frac{\pk\SkXY^\T(0) }{p_1p_0}, \end{split}\]
	and the lower right block is
	\[ \frac{\pk}{p_1^2}\SkX+\frac{(1-\pk)}{p_0^2}\SkX-\Big(\frac{\pk-p_1}{p_1p_0}\Big)^2\SkX=\frac{\pk(1-\pk)}{p_1^2p_0^2}\SkX. \]
	Therefore,
	$\cov\{n^{1/2}(\hat\tau-\tau,\ \ttX^\T)^\T\}=\cov\{n^{1/2}(\hat\tau-\tau,\ (\htW^\T-\tW^\T)^\T\}=U.$

\end{proof}

\subsubsection{Proof of Corollary \ref{cor2}}
\begin{proof}
Let $\Wi(z)$ be the same as in the proof of of Proposition \ref{props1}. We can apply Theorem \ref{thm1CLT} to $\Ymul_i(z)=(\Yi(z),\ \Wi(z)^\T)^\T$. Conditions \ref{cond::max} and \ref{cond::covariance} can be deduced from Conditions \ref{cond4} and \ref{cond6}, thus $n^{1/2}(\hat\tau-\tau,\ \htW^\T-\tW^\T)^\T$ converges in distribution to $\mathcal{N}(0,U^\infty)$. Therefore, the corollary holds because $n^{1/2}\tW$ converges to $\omega$ as $n\to\infty$.

\end{proof}

\subsubsection{Proof of Proposition \ref{prop2pa}}
\begin{proof}
    According to Corollary \ref{cor2}, $n^{1/2}\ttX\ \dot{\sim}\ \mathcal{N}( \omega ,U_{xx})$. Then the asymptotic distribution of the Mahalanobis distance is
	\begin{equation*}
		M_\ttX=(n^{1/2}\ttX)^\T U_{xx}^{-1}(n^{1/2}\ttX)\  \dot{\sim}\ \chi^2_{\p}(\omega^\T U_{xx}^{-1}\omega).
	\end{equation*}
	Therefore, the probability of a random assignment being accepted is 
	\[ \pr(M_\ttX<a)\to p_a'=\pr\{\chi^2_{\p}(\omega^\T U_{xx}^{-1}\omega)<a\} \] as $n$ tends to infinity.

\end{proof}

\subsubsection{Proof of Theorem \ref{thm4}}

\begin{lemma}
\label{corA11}
	Under SRRdM, \[
	n^{1/2}(\hat\tau-\tau, \ttX^\T)^\T\mid\Md\ \dot\sim\  (A,B^\T)^\T\mid\phi(B,U_{xx})=1,\]
	where $(A,B^\T)^\T\sim\mathcal{N}((0,\omega^\T)^\T,U)$.
\end{lemma}

\begin{proof}
    As $\Md \iff \phi(n^{1/2}\ttX,U_{xx})=1$, the proof of this lemma is similar to that of Lemma \ref{corA1}, so we omit it.

\end{proof}

\begin{proof}[Proof of Theorem \ref{thm4}]
Let $(A,B^\T)^\T\sim\mathcal{N}((0,\omega^\T)^\T,U)$ be the same as in Lemma \ref{corA11}. Then according to Lemma \ref{corA11},
\[
n^{1/2}(\hat{\tau}-\tau,\htW^\T)^\T\mid\Md\ \dot{\sim}\ (A, B_0^\T+\omega^\T)^\T \mid\phi(B,U_{xx})=1,
\]
where $B_0 = B - \omega \sim \mathcal{N}(0, U_{xx})$. As under SRRdM, $\phi(B,U_{xx})=1$ if and only if $(B_0+\omega)^\T{U}_{xx}^{-1}(B_0+\omega)<a$, to show the asymptotic biasedness of $\hat\tau$, we compute the expectation of $A\mid (B_0+\omega)^\T({U}_{xx})^{-1}(B_0+\omega)<a$.

Let $\epsilon=A-U_{\tau x}U_{xx}^{-1}(B_0+\omega)$ be the residual from the linear projection of $A$ on $B_0+\omega$. Then $\epsilon\sim\mathcal{N}(U_{\tau x}U_{xx}^{-1}\omega,(1-R^2)U_{\tau\tau})$ and $\epsilon$ is independent of $B_0+\omega$. Let $D_\omega=U_{xx}^{-1/2}(B_0+\omega)$, then $D_\omega\sim\mathcal{N}(U_{xx}^{-1/2}\omega,I)$ is independent of $\epsilon$, and $A=\epsilon+U_{\tau X}U_{xx}^{-1/2}D_\omega$.

Thus
\begin{equation*}
\begin{split}
	&E\{A\mid (B_0+\omega)^\T{U}_{xx}^{-1}(B_0+\omega)<a\}\\
	=\ &E(\epsilon)+E\{U_{\tau x}U_{xx}^{-1/2}D_\omega\mid D_\omega^\T D_\omega<a\} \\
	=\ &U_{\tau x}U_{xx}^{-1}\omega+U_{\tau x}U_{xx}^{-1/2}E(D_\omega\mid D_\omega^\T D_\omega<a)\\
	=\ &U_{\tau x}U_{xx}^{-1/2}\big\{U_{xx}^{-1/2}\omega+E(D_\omega\mid D_\omega^\T D_\omega<a)\big\},
\end{split}
\end{equation*}
which is usually not equal to $0$ when $\omega\neq0$.

To obtain the equivalence of SRRdM and SRRoM,  it is enough to show that $M_{\htX}=M_\ttX$ when the propensity scores are identical across strata. When $\pk=p_1= \nt / n \ (k=1,\ldots,\S)$, 
	\[\begin{split}
	\htX=&\skS\frac{\nk}{n}\biggl\{\frac{1}{\nkt}\sumk\Zi\Xi-\frac{1}{\nkc}\sumk(1-\Zi)\Xi\biggr\} \\
	=&\frac{1}{n}\Big\{\frac{1}{p_1}\sum_{i=1}^n\Zi\Xi-\frac{1}{p_0}\sum_{i=1}^n(1-\Zi)\Xi\Big\}=\ttX,
	\end{split}\] 
and
    \begin{equation*}
	U_{xx}=\skS\pik \frac{\pk(1-\pk)}{p_1^2p_0^2}\SkX=\skS\frac{\pik}{\pk ( 1 - \pk )}\SkX = \Sigma_{xx},
    \end{equation*}
    and therefore $M_{\htX}=n\htX^\T\Sigma_{xx}^{-1}\htX=n\ttX^\T U_{xx}^{-1}\ttX=M_\ttX$.
\end{proof}

\subsection{Another conservative variance estimator}
\subsubsection{Proof of Theorem \ref{thmvaro2}}
The proof is similar to the proof of Theorem \ref{thmvaro}. First we present a lemma. Let
$$
\tilde{\Sigma}_{\tau\tau}^*= \Sigma_{\tau \tau} + \sum_{k \notin A_{ss}} \pik\Skt + \sum_{j=1}^J\frac{m_j^2K_j}{n_{ss}^2(K_j-1)}\skmj (\tauk-\tmj)^2 .
$$
\begin{lemma}
\label{lem::conservative-sb2}
Under the conditions of Theorem \ref{thmvaro2}, 
$$
E( \hat  \Sigma_{\tau\tau}^* ) = \tilde{\Sigma}_{\tau\tau}^* \quad \textnormal{and} \quad  \hat  \Sigma_{\tau\tau}^* - \tilde{\Sigma}_{\tau\tau}^* = o_p(1).
$$ 
\end{lemma}
\begin{proof}[Proof of Lemma~\ref{lem::conservative-sb2}]
According to Corollary 3.4.1 of \citet{Pashley2017}, we have
\begin{eqnarray}
&& E \left\{   \left(  \frac{n_{ss} }{n} \right)^2 \frac{n}{n_{ss}^2 }\sum_{j=1}^J(m_jK_j)^2 \frac{1}{K_j(K_j-1)}\sum_{k\in A_{ss}:\nk=m_j}(\htk-\htmj)^2 \right\} \nonumber \\
& = &  \sum_{k \in A_{ss} }  \pik \Big\{ \frac{S_{[k]Y}^2(1)}{\pk} + \frac{S_{[k]Y}^2(0)}{1 - \pk} - \Skt \Big\} +  \sum_{j=1}^J\frac{m_j^2K_j}{n_{ss}^2(K_j-1)}\skmj (\tauk-\tmj)^2. \nonumber
\end{eqnarray}
Since 
$$
E \left[   \sum_{k \notin A_{ss} } \pik\Big\{\frac{\skY(1)}{\pk}+\frac{\skY(0)}{1-\pk}\Big\} \right] =  \sum_{k \notin A_{ss} } \pik\Big\{\frac{\SkY(1)}{\pk}+\frac{\SkY(0)}{1-\pk}\Big\}, 
$$
then
$$
E( \hat  \Sigma_{\tau\tau}^* ) = \tilde{\Sigma}_{\tau\tau}^*.
$$
Next, we show that $ \hat  \Sigma_{\tau\tau}^* - \tilde{\Sigma}_{\tau\tau}^* = o_p(1)$.  We have shown that 
\begin{equation}
\label{eqn::b1-term2}
\var \left[   \sum_{k \notin A_{ss} } \pik\Big\{\frac{\skY(1)}{\pk}+\frac{\skY(0)}{1-\pk}\Big\}  \right] \rightarrow 0.
\end{equation}
By definition,
\begin{eqnarray}
\hat \Sigma_{\tau\tau}^* &= &   \sum_{k \notin A_{ss} } \pik\Big\{\frac{\skY(1)}{\pk}+\frac{\skY(0)}{1-\pk}\Big\}  \nonumber \\
&& +  \left(  \frac{n_{ss} }{n} \right)^2 \frac{n}{n_{ss}^2 }\sum_{j=1}^J(m_jK_j)^2 \frac{1}{K_j(K_j-1)}\sum_{k\in A_{ss}:\nk=m_j}(\htk-\htmj)^2 \nonumber \\
& \triangleq & B_1 + B_2,
\end{eqnarray}
where 
$$
B_1 =   \sum_{k \notin A_{ss} } \pik\Big\{\frac{\skY(1)}{\pk}+\frac{\skY(0)}{1-\pk}\Big\}, \quad B_2  =  \hat \Sigma_{\tau\tau}^* - B_1.
$$
By \eqref{eqn::b1-term2} and Chebyshev's inequality, we have $B_1 - E(B_1) = o_p(1).$
Now it suffices for  $ \hat  \Sigma^*_{\tau\tau} - \tilde{\Sigma}^*_{\tau\tau} = o_p(1)$ to show that 
$B_2 - E(B_2) = o_p(1)$. 
Let
$$
\theta_j =  \frac{m_j^2K_j}{n_{ss}^2(K_j-1)},\ j=1,\ldots,J.
$$
Then $\theta_k$ has the same order as $1/n_{ss}^2$. The term $B_2$ can be further divided into six terms:
\begin{eqnarray}
B_2 & = &  \left(  \frac{n_{ss} }{n} \right)^2 \frac{n}{n_{ss}^2 }\sum_{j=1}^J(m_jK_j)^2 \frac{1}{K_j(K_j-1)}\sum_{k\in A_{ss}:\nk=m_j}(\htk-\htmj)^2 \nonumber \\
& = &     \frac{n^2_{ss} }{n}   \sum_{j=1}^J \skmj \theta_j(\htk-\tauk+\tauk-\tmj+\tmj-\htmj)^2 \nonumber \\
& = &   \frac{n^2_{ss} }{n}   \sum_{k \in A_{ss} } \theta_j  ( \htk - \tauk  )^2 +  \nonumber \\
& &   \frac{n^2_{ss} }{n}   \sum_{k \in A_{ss} } \theta_j  ( \tauk - \tmj )^2 +  \nonumber \\
&&   \frac{n^2_{ss} }{n}   \sum_{k \in A_{ss} } \theta_j  ( \tmj -  \htmj )^2 +  \nonumber \\
&&2 \cdot \frac{n^2_{ss} }{n}   \sum_{k \in A_{ss} } \theta_j   ( \htk - \tauk ) ( \tauk - \tmj )  + \nonumber \\
&& 2 \cdot \frac{n^2_{ss} }{n}   \sum_{k \in A_{ss} } \theta_j   ( \htk - \tauk ) ( \tmj -  \htmj ) + \nonumber \\
&&2 \cdot  \frac{n^2_{ss} }{n}   \sum_{k \in A_{ss} } \theta_j   ( \tauk - \tmj )  ( \tmj -  \htmj ) \nonumber \\
&\triangleq& B_{21} + B_{22} + B_{23} + B_{24} + B_{25} + B_{26}. \nonumber
\end{eqnarray}
\begin{itemize}
\item For the first term $B_{21}$, we have
\begin{eqnarray}
\var( B_{21} ) & = & \frac{n_{ss}^4}{n^2}   \sum_{j=1}^J\skmj \theta_j^2  \var\{   ( \htk - \tauk  )^2 \} \nonumber \\
& \leq &  \frac{n_{ss}^4}{n^2}   \sum_{j=1}^J\skmj \theta_j^2   E\{   ( \htk - \tauk  )^4 \} \nonumber \\
& \leq &   \frac{n_{ss}^4}{n^2}   \sum_{j=1}^J\skmj \theta_j^2   (D_{n1} + D_{n0})^2  E\{   ( \htk - \tauk  )^2 \}  \nonumber \\
& = & O \left[  \frac{  (D_{n1} + D_{n0})^2 }{n^2}  \sum_{k\in A_{ss}}   \Big \{  \frac{\SkY(1) }{\nkt} +  \frac{\SkY(0) }{\nkc} -  \frac{\Skt }{\nk}    \Big\}  \right]  \nonumber \\
& = & O \left\{  \frac{ (D_{n1} + D_{n0})^2 }{n}  \right\}  \nonumber \\
& \rightarrow & 0,  \nonumber 
\end{eqnarray}
where $D_{nz} = \max_{k=1,\dots,K} \max_{\ik} |  Y_i(z) - \Yk(z) | $ for $z=0,1$. Thus, $B_{21} - E( B_{21} ) = o_p(1)$.
\item For the second term $B_{22}$, we have $\var(B_{22}) = 0$ and $B_{22} - E( B_{22} ) = 0 = o_p(1)$.
\item For the third term $B_{23}$, we have
\begin{eqnarray}
\var( B_{23} ) & = & \frac{n_{ss}^4}{n^2}  \sum_{j=1}^J ( K_j  \theta_j )^2  \var\{   ( \tmj -  \htmj  )^2 \} \nonumber \\
& = &  O \left( \frac{1}{n^2} \right) \cdot \sum_{j=1}^J K_j^2  E\{   ( \tmj -  \htmj  )^4 \} . \nonumber
\end{eqnarray}
Note that
\begin{eqnarray}
\label{eqn::bound-htmj-tmj}
|  \tmj -  \htmj  |  & = & \Big| \skmj \frac{1}{K_j} ( \tauk - \htk ) \Big|  \leq  \max_{k} | \tauk - \htk | \leq  D_{n1}+D_{n0},
\end{eqnarray}
therefore,
\begin{eqnarray}
\var( B_{23} ) & = &  O \left\{ \frac{(D_{n1} + D_{n0} )^2 }{n^2} \right\}  \cdot \sum_{j=1}^J K_j^2 E\{   ( \tmj -  \htmj  )^2 \} . \nonumber \\
& = &  O \left\{ \frac{ (D_{n1} + D_{n0} )^2 }{n^2} \right\}  \cdot\sum_{j=1}^J K_j^2 O \left( \frac{1}{m_jK_j} \right) \nonumber \\
& = & O \left\{  \frac{ (D_{n1} + D_{n0})^2 }{n}  \right\}  \nonumber \\
& \rightarrow & 0. \nonumber
\end{eqnarray}
Thus, $B_{23} - E( B_{23} ) = o_p(1)$.
\item For the fourth term $B_{24}$, we have
\begin{eqnarray}
\var( B_{24} ) & = &  \frac{4 n_{ss}^4}{n^2}   \sum_{j=1}^J \skmj \theta_j^2   ( \tauk - \tmj )^2  \cdot \var (   \htk - \tauk  )   \nonumber \\
& = & O \left(  \frac{1}{n^2}  \right)\cdot  \sum_{j=1}^J \skmj ( \tauk - \tmj )^2  \cdot \var (   \htk - \tauk  ). \nonumber
\end{eqnarray}
Since
$$ \tmj^2 = \left( \skmj \frac{1}{K_j}\tauk\right)^2 \le \left( \skmj \frac{1}{K_j^2} \right) \cdot \Big(\skmj\tauk^2\Big) \le \skmj\tauk^2, $$
then 
\begin{eqnarray}
&  & \sum_{j=1}^J \skmj ( \tauk - \tmj )^2  \cdot \var (   \htk - \tauk  ) \nonumber \\
& \le &\sum_{j=1}^J \skmj 2(\tauk^2 + \skmj\tauk^2) \cdot \var (   \htk - \tauk  )\nonumber \\
& \le & \sum_{j=1}^J \Big[ 2\skmj \{\tauk^2\var (   \htk - \tauk  )\} + 2\Big(\skmj\tauk^2\Big) \Big\{ \skmj \var (   \htk - \tauk  ) \Big\} \Big] \nonumber \\
& \le & 2 \sum_{k\in A_{ss}} \tauk^2\var (   \htk - \tauk  ) + 2 \Big(\sum_{k\in A_{ss}}\tauk^2\Big) \Big\{\sum_{k\in A_{ss}}\var (   \htk - \tauk  )\Big\} \nonumber \\
&\le & 4\Big(\sum_{k\in A_{ss}}\tauk^2\Big) \Big\{\sum_{k\in A_{ss}}\var (   \htk - \tauk  )\Big\} \nonumber \\
&= & O(n)\cdot O \left(\frac{1}{n}\right) = O(1). \nonumber
\end{eqnarray}
Thus, $\var( B_{24} )\rightarrow 0$ and $B_{24} - E( B_{24} ) = o_p(1)$.
\item For the fifth term
$$
B_{25} =  \frac{ 2 n^2_{ss} }{n}   \sum_{j=1}^J \skmj  \theta_j   ( \htk - \tauk )  \cdot ( \tmj -  \htmj ),
$$
by \eqref{eqn::bound-htmj-tmj} we have  
$$
| B_{25} | \leq  \frac{D_{n1}+D_{n0} } {\sqrt{n}} \cdot \Big|   \frac{ 2 n^2_{ss} }{ \sqrt{n} } \sum_{j=1}^J \skmj \theta_j   ( \htk - \tauk ) \Big|.
$$
Therefore,
\begin{eqnarray}
\var( B_{25} ) & \leq  & E (B_{25})^2 \nonumber \\
& \leq & \frac{ (D_{n1}+D_{n0})^2 } { n } \cdot  \frac{ 4 n^4_{ss} }{ n } \cdot E \left\{ \sum_{j=1}^J \skmj \theta_j   ( \htk - \tauk ) \right\}^2  \nonumber \\
& = &  \frac{ (D_{n1}+D_{n0})^2 } { n } \cdot  \frac{ 4 n^4_{ss} }{ n } \cdot \sum_{j=1}^J \skmj \theta_j^2  E\{   ( \htk - \tauk  )^2 \} \nonumber \\
& = & O \left\{  \frac{ (D_{n1} + D_{n0})^2 }{n^2}  \right\}  \nonumber \\
& \rightarrow & 0.  \nonumber 
\end{eqnarray}
Thus, $B_{25} - E( B_{25} ) = o_p(1)$.
\item For the last term $B_{26}$, we have
\begin{eqnarray}
\var( B_{26} ) & = & \frac{4 n_{ss}^4}{n^2}  \sum_{j=1}^J \skmj \theta_j^2  ( \tauk - \tmj )^2 \cdot  \var( \tmj -  \htmj ) \nonumber \\
& = &  O \left(  \frac{1}{n^2}\right) \cdot \sum_{j=1}^J  \var( \tmj -  \htmj ) \skmj ( \tauk - \tmj )^2, \nonumber
\end{eqnarray}
where
\begin{equation*}
\begin{split}
&\skmj  ( \tauk - \tmj )^2 \le \skmj 2(\tauk^2+\tmj^2) \\
\le&\ 2\skmj\tauk^2 + 2K_j\skmj\tauk^2 \\
\le&\ 2(1+K_j) \skmj \tauk^2.
\end{split}
\end{equation*}
Since $\sum_{k\in A_{ss}}\tauk^2 = O(n)$ (see the proof of Theorem~\ref{thm:cov-est}), then
\begin{eqnarray}
\var( B_{26} ) & = & O \left(  \frac{1}{n^2}\right) \cdot \sum_{j=1}^J  \var( \tmj -  \htmj ) 2(1+K_j) \skmj \tauk^2 \nonumber \\
&= & O \left(  \frac{1}{n^2}\right) \cdot\sum_{j=1}^J \skmj \tauk^2 \nonumber \\
&=& O\left( \frac{1}{n}\right) \rightarrow 0.\nonumber
\end{eqnarray}
Thus, $B_{26} - E( B_{26} ) = o_p(1)$.
\end{itemize}
Combining the above arguments, we have $B_2 - E(B_2) = o_p(1)$.
\end{proof}
Based on Lemmas~\ref{lemsAB} and \ref{lem::conservative-sb2}, Theorem \ref{thmvaro2} follows similarly as the proof of Theorem \ref{thmvaro}.

\end{singlespace}

\end{document}